\documentclass[11pt,draftcls,onecolumn]{IEEEtran}
\usepackage{amssymb, amsmath, graphicx, cite, color, epsfig}

\newcommand{\reals}{\mathbb{R}}

\newcommand{\complex}{\mathbb{C}}

\newcommand{\Fig}   {\mbox{Fig.} }

\newcommand{\eqdef}{ := }
\newcommand{\twolines}[2]{\genfrac{}{}{0pt}{}{#1}{#2}}

\newcommand{\thmend}{\hspace*{\fill}~\QEDopen\par\endtrivlist\unskip}

\newcommand{\na} { {\sf a} }
\newcommand{\nb} { {\sf b} }
\newcommand{\nc} { {\sf c} }
\newcommand{\nd} { {\sf d} }
\newcommand{\nr} { {\sf r} }

\newcommand{\nzero} { {\sf 0} }
\newcommand{\none} { {\sf 1} }
\newcommand{\ntwo} { {\sf 2} }
\newcommand{\nm} { {\sf m} }
\newcommand{\prot}{ {\sf P} }
\newcommand{\pa} { {(1)} }
\newcommand{\pb} { {(2)} }
\newcommand{\pc} { {(3)} }
\newcommand{\pd} { {(4)} }

\newtheorem{theorem}{Theorem}
\newtheorem{lemma}[theorem]{Lemma}

\newtheorem{remark}[theorem]{Remark}

\allowdisplaybreaks[1]
\begin{document}
\title{Achievable rate regions and outer bounds for a multi-pair bi-directional relay network}
\author{Sang Joon Kim, Besma Smida and Natasha Devroye%
\thanks{ Sang Joon Kim was with the School
of Engineering and Applied Sciences, Harvard University, Cambridge,
MA 02138. Email:~sangkim@fas.harvard.edu.
Besma Smida is with the Department of Electrical and Computer
Engineering, Purdue University Calumet, Hammond, IN 46323. E-mail:~Besma.Smida@calumet.purdue.edu.
%vahid@deas.harvard.edu.
Natasha Devroye is with the Department of Electrical and Computer
Engineering, University of Illinois at Chicago, Chicago, IL 60607.
Email:~devroye@ece.uic.edu.}}

\maketitle

\begin{abstract}
In a bi-directional relay channel, a pair of nodes wish to exchange independent messages over a shared wireless half-duplex channel with the help of relays. Recent work has mostly considered information theoretic limits of the bi-directional relay channel with two terminal nodes (or end users) and one relay. In this work we consider bi-directional relaying with one base station, multiple terminal nodes and one relay, all of which operate in half-duplex modes. We assume that each terminal node communicates with the base-station in a bi-directional fashion through the relay and do not place any restrictions on the channels between the users, relays and base-stations; that is, each node has a direct link with every other node. %: no information is exchanged directly between the end users.
%with the base station bi-directionally and that no information is directly exchanged between end users.
Our contributions are three-fold: 1) the introduction of four new temporal protocols which fully exploit the two-way nature of the data and outperform simple routing or multi-hop communication schemes by carefully combining network coding, random binning and user cooperation which exploit over-heard and own-message side information, 2) derivations of inner and outer bounds on the capacity region of the discrete-memoryless multi-pair two-way network, and 3) a numerical evaluation of the obtained achievable rate regions and outer bounds in Gaussian noise which illustrate the performance of the proposed protocols compared to simpler schemes, to each other, to the outer bounds, which highlight the relative gains achieved by network coding, random binning and compress-and-forward-type cooperation between terminal nodes.
%We introduce three temporal protocols and derive achievable rate regions and outer bounds for the proposed protocols with decode-and-forward relaying. Then we provide a comprehensive treatment of the proposed protocols in Gaussian channel, illustrating the achievable rate regions, outer bounds and relative performance at different SNRs.
\end{abstract}


%%%%%%%%%%%%%%%%%%%%%%%%%%%%%%%%%%%%%%%%%%%%%%%%%%%%%%%%%%%%%

\section{Introduction}

%%%%%%%%%%%%%%%%%%%%%%%%%%%%%%%%%%%%%%%%%%%%%%%%%%%%%%%%%%%%%
\label{sec:intro}

% What problem we will consider
% Why this problem is interesting

\subsection{Goal and motivation}
In this work, we will derive achievable rate regions and outer bounds for a multi-pair bi-directional relay network. The most basic  bi-directional network consists of a pair of \emph{terminal nodes} that wish to exchange independent messages.
In wireless channels or mesh networks, this communication may take place with the help of {\it relay} nodes which do not wish to transmit any information of their own. The most basic bi-directional {\it relay} network thus consists of a pair of \emph{terminal nodes} that wish to exchange messages through the relay of a single relay. While the capacity of this channel is still unknown in general, as will be outlined in the Related Work sub-section next, it has been of great recent interest due to its relevance in wireless networks of the future.

The single relay,  single pair bi-directional relay channel has been extended in a number of ways: 1) the consideration of a single bi-directional link using {\it multiple relays}, 2) the consideration of {\it multiple bi-directional links} sharing a single, common relay or, most generally, 3) the consideration of  multiple bi-directional links which communicate with the help of multiple relays.

 The relay network considered in this paper falls into the second category and consists of a base station (node $\nzero$) which wishes to communicate simultaneously in a bi-directional fashion with multiple terminal nodes (node $\none$, $\cdots$, node $\nm$) with the help of  one relay node (node $\nr$). Due to limitations of current technology, all nodes are assumed to be half-duplex and thus cannot transmit and receive simultaneously. This network topology is motivated by recent pushes to extend the coverage of wireless networks. For example, in a cellular scenario, a relay station is able to enhance the connectivity between a base station and terminals at its cell boundary. The relays may be connected to the base station using a wireless link rather than a wired one,  resulting in savings to the operators' backhaul costs. Another motivating example is satellite communication: satellites can be used to relay signals from one ground station to multiple vehicular terminals on or close to the earth's surface. In this work, we determine bounds on the capacity regions  - which  may serve as guides and benchmarks in the eventual design of - such multi-pair two-way communication networks aided by a single relay node.

% Related work
\subsection{Related work} Two-way communications were first considered by Shannon himself  \cite{Shannon:1961}, in which he introduced inner and outer bounds on the capacity region of the two-way channel where two full-duplex nodes (which may transmit and receive simultaneously) wish to exchanges messages. Since full-duplex operation is, with current technology, of limited practical significance, in this work we assume that the nodes are {\em half-duplex}, i.e. at each point in time, a node can either transmit or receive symbols, but not both.
%While the capacity region of the two-way channel is known for certain two-way channels \cite{Han:1984, Shannon:1961}  the capacity region of the two-way channel (with full-duplex operation) is in general unknown \cite{Meulen:1977, Dueck:1979}.

The two-way relay channel or bi-directional relay channel is the logical extension of the classical relay channel \cite{Meulen:1971, Cover:1979, Kramer:Gastpar:Gupta} for one-way point-to-point   communication aided by a relay to allow for two-way communication.
Alternatively, it may be seen as the natural extension of the two-way channel which allows communication to take place with the help of a single relay. This channel has been of great interest of late and has been considered from a number of different perspectives.
A large body of work concerning the bi-directional relay channel - which we do not attempt to fully summarize -  has emerged, which may be differentiated roughly based on combinations of assumptions that are made on the type of relaying (CF, DF, AF, de-noise, mixed, lattice codes) and on the duplex abilities of nodes (half-duplex or full-duplex). We highlight some of the work under different assumptions before proceeding to describe extensions to multiple terminal nodes and/or multiple relay nodes.

1. {\bf Relaying type:} The simplest of relaying types is Amplify-and-forward (AF), in which relays are not required to do any processing on the received signal but re-scale and re-transmit it. %Its benefits lie in its simplicity and lack of processing needed at the relay.
%The price paid is that noise, as well as the desired signal will be amplified in each additional hop, making it less attractive for multi-hop scenarios.
AF schemes are often used as a benchmark against which to compare the performance of more complex relaying types \cite{SKim:2007, Rankov:2007,Popovski:ICC, Popovski:2006a, Popovski:2006b, Ho:2008}. Decode and forward (DF) relaying assumes the relay is able to decode all messages before re-transmitting them.
%In order to perform any form of network coding at the message level (rather than at the signal level in analog network coding \cite{Katti:2007}) decode and forward relays must be assumed.
Examples of work which assume bi-directional DF relaying include \cite{Larsson:2005, Larsson:2006, Oechtering:2007, Oechtering:2008, Oechtering:2009, OB08pacm, Wyrembelski:2008, SKim:2007, Kim:ISIT2009, Schnurr:2007, Kim:sarnoff, Rankov:2006}.
Using DF relaying allows the use of network coding at the message level for the broadcast phase and prevents the re-transmission and possible amplification of noise, at the cost of forcing the relays to decode the messages, possibly reducing the permissible transmission rates. Compress and forward (CF), as first introduced in the context of the classical relay channel \cite{Cover:2006}, and considered in the bi-directional relaying context in \cite{Rankov:2006, Schnurr:2007, Kim:sarnoff} and the conceptually related de-noise and forward  \cite{Popovski:ICC}, \cite{Popovski:2006a}, \cite{Popovski:2006b}, requires the relay to re-transmit a quantized or compressed version of the received signal. This scheme has the advantage that the rate need not be lowered so as to allow the relay to fully decode it, but may still mitigate some of the noise amplification effects seen in AF relaying by judicious choice of the quantizer or compressor. In the authors' previous work, \cite{Kim:sarnoff, SKim:2007} the {\it mixed} forwarding scheme is also proposed, in which the streams of information traveling in the two directions are treated differently, i.e. one direction may use DF while the other uses CF, exploiting the intuitive fact that DF is preferable when a relay is ``close'' to the source while CF is generally preferable, rate-wise when the relay is ``close'' to the destination \cite{Kramer:Gastpar:Gupta, SKim:2007}.
% Expanding on the intuition gained from the classical relay channel \cite{Kramer:Gastpar:Gupta}, and as further shown in the bi-directional relaying context \cite{SKim:2007}, when a relay is close (or alternatively sees a weak channel to the source relative to the destination) to the destination it is preferable to perform CF, while if it closer (or sees a better channel to the source relative to the destination)  to the source DF is a better choice. Thus, mixed schemes in general must be considered.

2. {\bf Duplexing:} Both full-duplex as well as half-duplex nodes and their corresponding achievable rate regions have been considered for bi-directional relaying. In \cite{Wu:2005, Larsson:2005, Larsson:2006, Kim:ICDCSW07, SKim:2007, Oechtering:2007, Rankov:2007, Popovski:ICC, Popovski:2006a, Popovski:2006b, Vaze:2009, Yuen:VTC, Schnurr:2007, Kim:sarnoff} half-duplex nodes are assumed. This forces communication to take place over a number of phases, using different temporal {\it protocols}. A temporal protocol specifies which nodes simultaneously transmit at which time.
Three of the most commonly considered protocols are depicted in \ref{fig:TWRC}: the naive 4 phase protocol, the 3 phase time-division broadcast channel (TDBC) protocol, as well as the 2phase Multiple-access broadcast channel (MABC) protocol, where each ``layer'' of nodes describes a different temporal phase. It is interesting to note that in half-duplex protocols, the TDBC allows a destination to obtain {\it side-information}, or extra overheard knowledge, about the other user's message during the phase in which the message is destined to the relay. This is not possible in 2 phase MABC protocols in which both nodes transmit simultaneously to the relay in one phase and are thus unable to overhear any of the other relay's message.
% In \cite{Kim:ICDCSW07, Kim:2008} inner and outer bounds for the three phase TDBC protocol as well as a hybrid four phase protocol in which both TDBC and MABC  schemes are combined are also derived. It is shown that neither TDBC or MABC dominates the other for all possible channels, and, surprisingly, that the hybrid protocol combining the two is able to achieve points strictly outside the {\it outer} bounds of the individual TDBC and MABC protocols.
In \cite{Kim:ICDCSW07, SKim:2007} it is shown that neither TDBC or MABC dominate each other for all channel gains and SNRs. In \cite{Kim:sarnoff, SKim:2008a} a comprehensive treatment of CF, DF, AF and mixed forwarding schemes under  both the MABC and TDBC protocols highlights the significant interplay between relaying types and protocols.
In \cite{Oechtering:2007, Oechtering:2008, Oechtering:2009, OB08pacm, Wyrembelski:2008} the authors have thoroughly analyzed the broadcast phase
%from a number of angles and assumptions (thus assuming half-duplex, decode and forward relaying)
 of the bi-directional relay channel. The full-duplex scenarios have been considered somewhat less: in \cite{Rankov:2006} the authors derived achievable rate regions for the {\it restricted} two-way relay channel using DF, CF and AF schemes, in which the terminals may not cooperate in transmitting their messages. In \cite{nam:2009bit} full-duplex nodes are considered in order to analyze the system from a finite bit perspective.

\begin{figure}
\centerline{\epsfig{figure=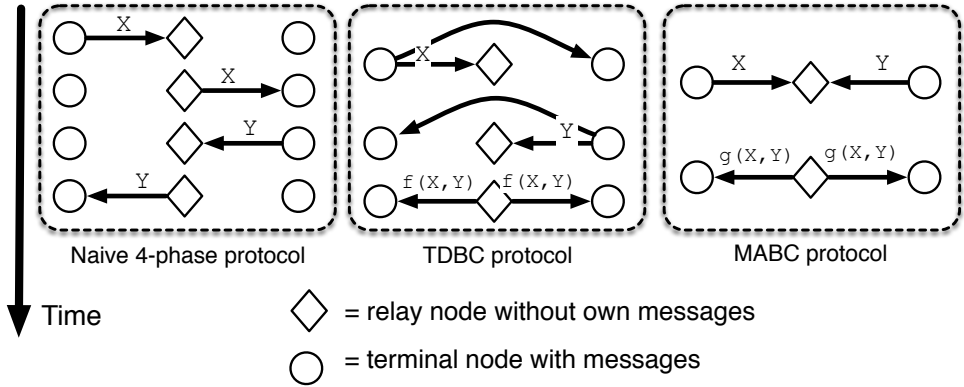, width=13cm}}
\caption{Three most common half-duplex bi-directional relaying protocols: a naive 4 phase protocol, the 3 phase TDBC (Time Division Broadcast Channel) protocol and the 2 phase MABC (Multiple Access Broadcast Channel) protocol.}
\label{fig:TWRC}
\end{figure}

%A number of extensions to include multiple terminal nodes or multiple relay nodes have been recently considered.
{\bf Extensions to multiple-relays and multiple terminal nodes.}
The bi-directional relay channel has been extended to include multiple relays \cite{SKim:2008b, OB08brru, Kim:ISIT2009, Pooniah:2008, Vaze:2008, Vaze:2008codeornot, Lutz:2009} and multiple terminals nodes (or multiple bi-directional data streams)  \cite{ Ghozlan_MIMO_switch, Chen:2008:CISS,Avestimehr:2009:ITW, Gunduz:2009:ISIT}. A variety of promising methods have been employed, but to date the capacity regions of these extensions are still unknown.

The study of a single bi-directional communication link aided by multiple relays has been approached from a number of angles \cite{OB08brru, Kim:ISIT2009, Pooniah:2008, Vaze:2008, Vaze:2008codeornot, Lutz:2009}, which seek to address how to ``best'' employ the relays.
How/which relays to select has been considered in unidirectional networks in for example \cite{Luo:VTC, Madan:2008, Bletsas:2006} and was considered for the first time for bi-directional relay networks in \cite{OB08brru}.  Alternatively, multiple relays could amplify and forward the received signals in a multi-hop fashion, or may decode and cooperatively re-encode and re-transmit the received signals. In the authors' previous work in \cite{Kim:ISIT2009, SKim:2008b} three classes of multiple relays protocols for the half-duplex, non-adaptive DMC and AWGN channel models are considered and inner and outer bounds on the capacity region are derived.

%In \cite{OB08brru}  the selection of a single DF relay out of $N$ possible relays under a MABC protocol is investigated. One of the central conclusions is that a single relay is not in general able to achieve all points in the achievable rate region,  i.e. different points in the region are achieved through the help of different relays and in general there is no `best' relay to select.  Time-sharing between different relays is shown to increase the achievable rate region but how to select this relay is not considered.

Bi-directional relay channels with multiple bi-directional communication links have been much less considered than their single link counterpart:

$\bullet$ In \cite{Ghozlan_MIMO_switch} an interference network, with no direct links between terminal nodes,  in which $K$ half-duplex single-antenna source/destination pairs wish to exchange messages in a bi-directional fashion is investigates  from a diversity-multiplexing gain perspective in the delay-limited high SNR regime.

$\bullet$ The authors of  \cite{Chen:2008:CISS} consider a similar channel model and propose the use of a CDMA strategy to support multiple users so as to guarantee QoS to different users.
%Significant power savings over one-way communication designs are demonstrated through the use of XORed symbols and shared spreading signatures.

$\bullet$ In \cite{Avestimehr:2009:ITW}  multiple bi-directional pairs communicate over a shared relay in the absence of a direct link between end nodes. Under a linear deterministic channel interaction model,  an interesting equation-forwarding strategy is shown to be capacity-achieving. This intuition is transfered to the two-pair full-duplex bi-directional Gaussian relay network in \cite{Avestimehr:2009:ISIT}, where a carefully constructed superposition scheme of random and lattice codes was used to achieved rates within 2 bits of the outer cut-set bound.

$\bullet$ Finally, in \cite{Gunduz:2009:ISIT}, an arbitrary number of clusters (nodes within a cluster all wish to exchange messages)  of arbitrary numbers of full-duplex nodes are assumed to communicate simultaneously through the use of a single relay in AWGN. Nodes are not able to hear each other,
%and the relaying schemes considered in general are AF, DF, CF and lattice codes when multiple two-way channels all interfere with each other (K clusters, each of size 1). I
and it is shown that CF achieves within a constant number of bits from capacity regardless of SNR; interesting conclusions are also drawn with respect to lattice coding versus CF.

In all four examples of multi-pair bi-directional communication with a single relay, no direct link between the terminal nodes is assumed to exist. This simplifies the analysis as the tradeoff between relayed information and directly communicated information is avoided. No ``overheard'' side information from one terminal to the other is possible, and the only side information available in this channel is each node's own message.

% Our contributions
\subsection{Our contributions}
We consider a bi-directional relay network with one base station, multiple terminal nodes and one relay, all of which operate in half-duplex mode.  The physical layout is shown in Fig. \ref{fig:channel}. We assume that each terminal node communicates with the base-station in a bi-directional fashion through the relay and do not place any restrictions on the channels between the users, relays and base-stations; that is, each node has a direct link with every other node.  The desired bi-directional links may be seen from the messages $W_{i,j}$ from node $i$ destined to node $j$, and $\tilde{W}_{i,j}$ the estimate at node $j$ of the message $W_{i.j}$ that is wishes to decode from node $i$. The base-station is denoted as as node with index $0$. Two elements of the formulated problem are markedly different from prior work in this area:

$\bullet$ the assumption that one end of the bi-directional links is a single base-station rather than independent nodes as in \cite{Gunduz:2009:ISIT, Avestimehr:2009:ISIT} and \cite{Ghozlan_MIMO_switch}.

$\bullet$ we place no assumptions on the channels between nodes - i.e. our nodes can all hear each other. This allows for the possibility of causal cooperation between nodes as well as direct transmission between the base-station and the nodes, using the relay only when beneficial.

Our central contributions are three-fold:

\begin{figure}
\centerline{\epsfig{figure=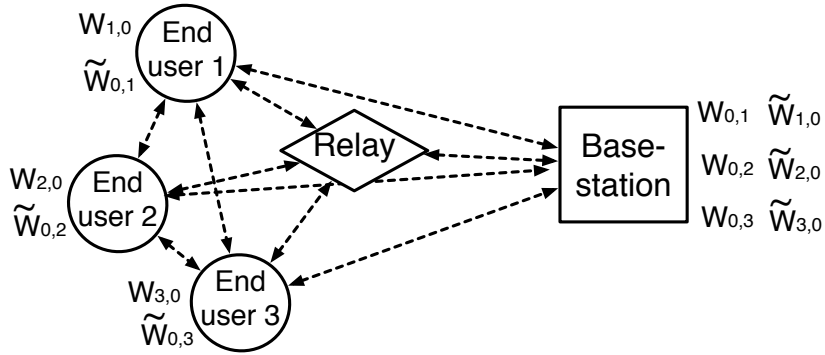, width=11cm}}
\caption{Our physical channel model consists of multiple independent bi-directional desired communication flows (indicated by arrows) between multiple terminal nodes and a single base-station. Communications may be aided using one relay node. We note that communication need not pass through the relay as direct links between the base-station and terminals exist. $W_{i,j}$ denotes the message from node $i$ to node $j$, while $\tilde{W}_{i,j}$ is the estimate at node $j$ of the message $W_{i,j}$ that it wishes to decode.}
\label{fig:channel}
\end{figure}

1.  We propose four temporal protocols which we call the FMABC \emph{(Full Multiple Access Broadcast)}, PMABC \emph{(Partial Multiple Access Broadcast)} and FTDBC \emph{(Full Time Division Broadcast)},  PTDBC \emph{(Partial Time Division Broadcast)}protocols: the first two are extensions of the MABC protocol of \cite{SKim:2007}, pictured in Fig.\ref{fig:TWRC} to multiple terminal-node pairs, while the last two are extensions of the TDBC protocol of \cite{SKim:2007}, also pictured in \Fig \ref{fig:TWRC} to multiple terminal nodes. %The relay may process and forward the received signals in a number of ways.
 %Standard forwarding strategies include decode-and-forward (DF), amplify-and-forward (AF), compress-and-forward (CF), and de-noise and forward.

2. Our central goal is to determine inner and outer bounds on the capacity region of the multi-pair bi-directional relay network. To derive achievable rate regions we suggest a number of achievable strategies, the key elements of which are:
\begin{itemize}
\item  {\it Forwarding:} We consider the decode-and-forward (DF) forwarding strategy at the relay node;
%. AF, CF and mixed-forwarding strategies at the relay are left for future work; the regions are sufficiently complex as is. At the terminal nodes, we may
and employ a Compress and Forward (CF) - type cooperation strategy at terminal nodes.
\item {\it Multiple-access:} In contrast to the ``naive'' protocol of Fig. \ref{fig:TWRC}, we will consider schemes in which multiple nodes transmit simultaneously to a single node as in a multiple-access channel.
\item  {\it Marton's broadcast region:} Due to the presence of a base-station with multiple messages (one to each of the terminal nodes) as well as a relay with multiple decoded messages (traveling in both directions),
%a multi-user broadcasting strategy will be of significant use in deriving achievable rate region. We
%we generalize Marton's broadcasting scheme \cite{marton} to $>2$ messages and note that
we use a  {\it modified} version of a generalization of Marton's broadcasting scheme \cite{marton} to $>2$ messages/users, which takes into account own-message side-information at each node.
\item {\it Random binning:} %Random binning is not only used in a Marton-like fashion but is
is  used to exploit side-information from direct links as in the TDBC protocol in \cite{SKim:2007}, with the added
%, we apply random binning to
%by combin two information flows from the direct links and the relaying link. More
challenge that the direct links involve a general Marton broadcasting scheme.
% that both links are not just a single transmission but more complicated general broadcasting.
\item {\it Network coding:}
%As in many of the classical bi-directional relay channel where network coding at the relay was able to efficiently combine the two flows of information during the final relay broadcast phase, w
We will use network coding on a {\it flow-by-flow} basis.
%to improve achievable rates in this multi-flow problem.
The different bi-directional flows are then combined using a random binning / broadcasting strategy.
\item {\it User cooperation:}
%As our channel model allows end users to over-hear the transmissions of the base-station and other end-nodes when they are not transmitting,
``Over-heard'' side information at terminal nodes is used to allow terminal nodes to cooperate - using a compress and forward strategy -  in transmitting their messages.
%Cooperation is enabled through a  compress and forward strategy. We will derive achievable rate regions with and without cooperation between terminal nodes.
\end{itemize}

% Paper outline
\subsection{Outline}
This paper is structured as follows: in Section \ref{sec:prelim}.
we introduce our notation and the different protocols we will be examining.
%an achievable rate region for the discrete memoryless broadcast channel with multiple terminals.
 In Section \ref{sec:marton} we introduce the extended Marton's bound for the general broadcast channel with more than two  receivers. In Section \ref{sec:bounds} we derive  achievable rate regions for the proposed protocols with DF relaying and CF-based terminal node cooperation, while in Section
\ref{sec:outer_bounds} we derive outer bounds for those protocols. In Section \ref{sec:Gaussian}, we apply the derived performance bounds to the Gaussian noise channel. In Section \ref{sec:regions}, we numerically compute
these bounds in the Gaussian noise channel and compare the results
for different powers and channel conditions, followed by the conclusion in Section \ref{sec:conclusion}.
The proofs of technical contributions are provided in the Appendices, while those for the bi-direcitonal inner and outer bounds are provided in the body of the document.

%%%%%%%%%%%%%%%%%%%%%%%%%%%%%%%%%%%%%%%%%%%%%%%%%%%%%%%%%%%%%

\section{Preliminaries}

%%%%%%%%%%%%%%%%%%%%%%%%%%%%%%%%%%%%%%%%%%%%%%%%%%%%%%%%%%%%%
\label{sec:prelim}

\subsection{Notations and Definitions}
We consider a base station (node $\nzero$), a set of terminal nodes
${\cal B} \eqdef \{\none,\ntwo,\cdots,\nm \}$  and a relay $\nr$ which
aids in the communication between the terminal nodes and the base station. We define ${\cal M} \eqdef {\cal B} \cup \{\nzero\} = \{\nzero, \none,\ntwo,\cdots,\nm\}$. We use $R_{i,j}$ to denote the
rate of communication from node $i$ to node $j$, i.e. the message between node $i$ and node $j$, $W_{i,j}$, lies in the set $   {\cal S}_{i,j} \eqdef \{0, \ldots, \lfloor 2^{nR_{i,j}} \rfloor - 1\}$.  Similarly, $R_{S,T}$ is the vector of rates from set $S$ to set $T$ where $S,T \subseteq {\cal M}$ at which the messages  $W_{S,T}\eqdef \{W_{i,j} | i\in S, j\in T,~ S,T\subseteq {\cal M}\} $ may be reliably communicated.
We assume that each end user communicates with the base station bi-directionally and that no information is directly exchanged between end users: i.e.
%We note that this assumption is removed in ongoing work, in which more general protocols are derived \cite{Kim:inprep}.
 every pair of terminal nodes $\nzero$ and $i \in {\cal B}$ wish to exchange independent messages while $R_{i,j} = 0$ (or is undefined) for all $i,j \in {\cal B}$. Thus, there are a total of $2m$ messages in our network: $m$ from node $\nzero$ to each node $i \in {\cal B}$, and $m$ from each node $i\in {\cal B}$ to node $\nzero$, as shown in Fig. \ref{fig:channel}\footnote{We allow information exchanges between nodes in ${\cal B}$ in the cooperation protocols. However these messages are not induced from the system, but just used as temporary information for decoding original messages between node $\nzero$ and ${\cal B}$.}.

Communication takes place over a number
of channel uses, $n$ and rates are achieved in the classical
asymptotic sense as $n\rightarrow \infty$. At channel use $k$, we
use $X_i^{k}$ to denote the input distribution and $Y_i^{k}$ to
denote the distribution of the received signal of node $i$.
During phase $\ell$ we use $X_i^{(\ell)}$ to denote the
input distribution and $Y_i^{(\ell)}$ to denote the distribution of
the received signal of node $i$.
Because of the
half-duplex constraint, not all nodes transmit/receive during all
phases and we use the dummy symbol $\varnothing$ to denote that
there is no input or no output at a particular node during a
particular phase, i.e. we must have
$X_i^{(\ell)} = \varnothing$ or $Y_i^{(\ell)} = \varnothing$ for all
$\ell$ phases. For convenience, we drop the notation $\varnothing$
from entropy and the mutual information terms when a node is not
transmitting or receiving.  $\Delta_{i,n}$ is the phase
duration of phase $i$ with block size $n$ and $\Delta_i$ is the
phase duration of phase $i$ when $n\rightarrow \infty$. It is
also convenient to define $X_{S}^k \eqdef \{X_i^k | i\in S\}$, the
set of input distributions by all nodes in the set $S$ at time $k$
and similarly $X_{S}^{(\ell)} \eqdef \{X_i^{(\ell)}|i\in S\}$, a set
of input distributions during phase $\ell$.

Each node $i$ has channel input alphabet ${\cal X}^*_i = {\cal
X}_i \cup \{ \varnothing \}$ and channel output alphabet ${\cal
Y}^*_i = {\cal Y}_i \cup \{ \varnothing \}$, which are related through a discrete memoryless channel\footnote{Arguments and extensions to Gaussian noise channels will be addressed in Section \ref{sec:Gaussian}.}.
 Lower case letters $x_i$
denote instances of the upper case $X_i$ which lie in the
calligraphic alphabets ${\cal X}_i^*$. Boldface ${\bf x}_i$
represents a vector indexed by time at node $i$. Finally, it is
convenient to denote by ${\bf x}_S \eqdef \{{\bf x}_i | i\in S\}$, a
set of vectors indexed by time. We also use the notation ${\bf x}_S(w_{S,T})$ to denote the dependence of ${\bf x}_S$ on the message set $w_{S,T}$.

For a block length $n$, encoders and decoders are functions
$X_i^k(W_{\{i\},{\cal M}},Y_i^1,\cdots, Y_i^{k-1})$ producing an encoded message at node $i$, and
$\tilde{W}_{i,j}(Y_j^1, \cdots, Y_j^n, W_{\{j\},{\cal M}})$ producing a decoded message or error at node $j$ when it wishes to decode  the message $W_{i,j}$ from node $i$. We define error events
$E_{S, T} \eqdef \{{W}_{i, j} \neq \tilde{W}_{i,j}(\cdot) | i\in S, j\in
T\}$ for decoding the messages ${W}_{ S,\{j\}}$ at nodes $j\in {T}$ at
the end of the block of length $n$, and $E_{ S, T}^{(\ell)}$ as the
error event at nodes $j\in {T}$ in which nodes $j\in {T}$ independently attempt to
decode ${W}_{S,T}$ at the end of phase $\ell$ using a joint
typicality decoder.
A set of rates $R_{i,j}$ is said to be achievable for a protocol with phase durations $\Delta_\ell$  if there exist
encoders/decoders of block length $n=1,2,\cdots $  with both
$P[E_{i,j}]\rightarrow 0$  and $\Delta_{\ell,n}\rightarrow \Delta_{\ell}$ as $n\rightarrow \infty$ for all $\ell$.  An achievable rate region (resp. capacity region) is the closure of a set of (resp. all) achievable rate
tuples for $\Delta_{\ell}$.

Let $A(U)$ be the set of $\epsilon$-strongly-typical
${\bf u}$ sequences. Similarly, $A^{(\ell)}(UV)$ represents the set of $\epsilon$-strongly-typical
$({\bf u}^{(\ell)},{\bf v}^{(\ell)})$ sequences of length $n \cdot \Delta_{\ell,n}$
according to the distributions $U$ and $V$ in phase $\ell$. Also define the event $D^{(\ell)}({\bf u},{\bf v})\eqdef
\{({\bf u}^{(\ell)},{\bf v}^{(\ell)}) \in A^{(\ell)}(UV)\}$ and use
 ${\bar D}$ to denote the complement of the event $D$. Similarly, ${\bar E}$ is defined.
We denote $\bigotimes$ as the cartesian product, i.e., $\bigotimes_{i=1}^3 {\cal X}_i = {\cal X}_1 \times {\cal X}_2 \times {\cal X}_3 $.
Finally, let  $S(j) \eqdef \{i | i<j , i\in S\}$.

\subsection{Protocols for multiple terminal-node pairs}

The total transmission time is divided into two time {\it divisions}, each of which may consist of one or more {\it phases}.  During the first time division - called the \emph{multiple access} division, the terminal nodes transmit  to the relay. During the second time division - called the \emph{broadcast} division - the relay transmits to the terminal nodes.
In the multiple access period, we consider four cases: 1) all terminal nodes transmit for the whole duration, 2) $\nzero$ uses the whole duration and the other terminal nodes $\none,\cdots,\nm$ transmit sequentially, 3) all nodes transmit sequentially and 4) $\nzero$ first transmits and the other terminal nodes multiple access. We denote the first protocol as \emph{Full Multiple Access Broadcast} (FMABC) protocol, the second protocol as \emph{Partial Multiple Access Broadcast} (PMABC) protocol, the third one as \emph{Full Time Division Broadcast} (FTDBC) protocol, and the last one as \emph{Partial Time Division Broadcast} (PTDBC) protocol.
A more visual description of which nodes transmit when is provided in  \Fig \ref{fig:protocol}.

For comparison purposes in our simulations, we also introduce what we call the \emph{simplest sequential protocol} where all terminal nodes sequentially transmit information to the relay, i.e., $\nzero\rightarrow \nr~,~\none\rightarrow \nr,\cdots,\nm\rightarrow\nr$, then the relay sequentially transmits them to the proper destinations, i.e., $\nr\rightarrow \nzero~,~\nr\rightarrow \none,\cdots,\nr\rightarrow\nm$.

\begin{figure}[t]
 \begin{center}
  \epsfig{figure=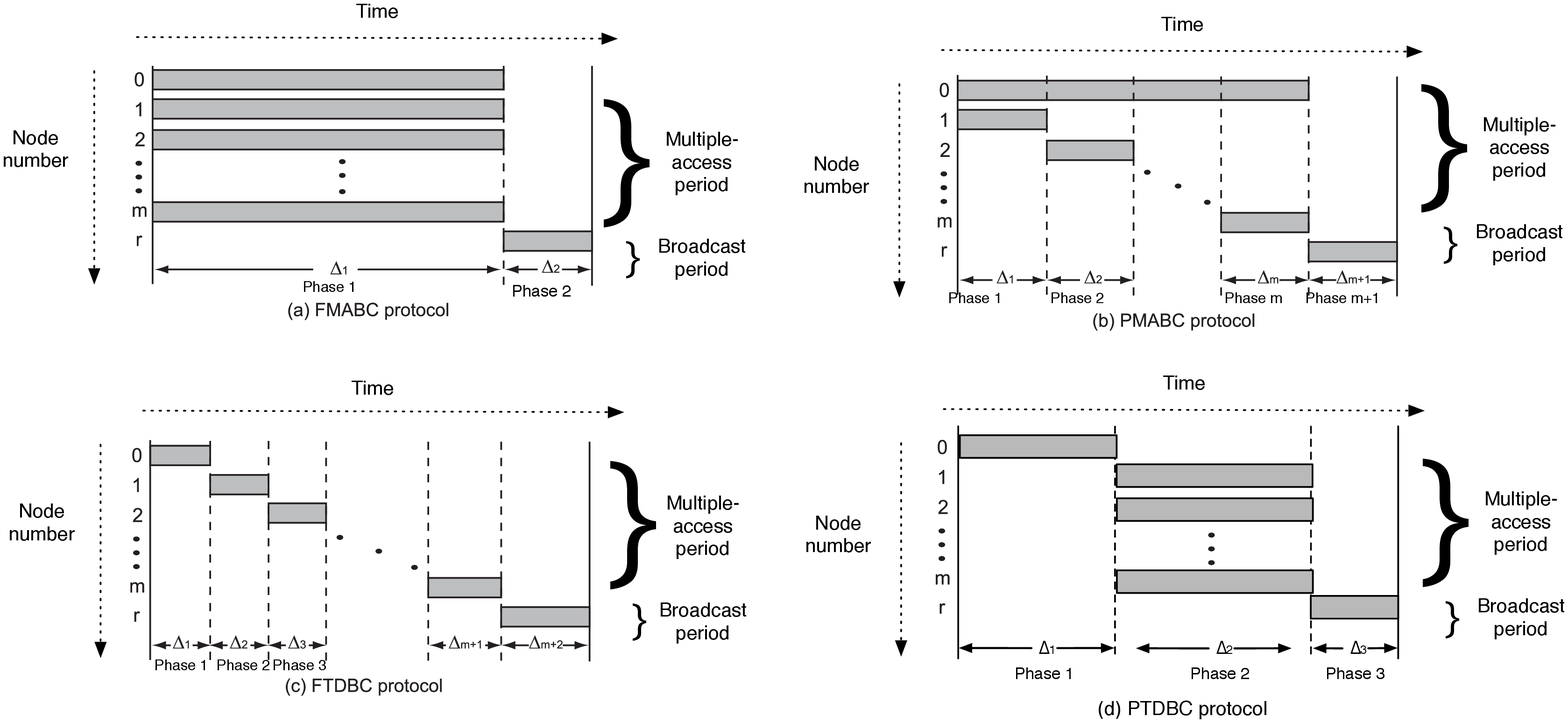, width=16cm}
  \caption{Four proposed half-duplex protocols - the time phases of the different protocols are seen; the encoders and decoders in the different phases may vary.}
  \label{fig:protocol}
 \end{center}
\end{figure}

%%% qlqchose ici
The FMABC, PMABC, FTDBC, and PTDBC protocols describe the temporal phases or divisions of the transmission scheme but not what each node sends, or how its messages are encoded during those phases.
%To improve the performance of the previous basic protocols and derive achievable rate regions,
We will exploit  three types of coding schemes -  network coding, random binning and user cooperation:
%We indicate ``-N'', ``-R'' and ``-C'' for network coding, random binning and user cooperation. For example, PMABC-NRC means the PMABC protocol with network coding, random binning and user cooperation, while FTDBC-NR means the FTDBC protocol with network coding and random binning.

\subsubsection{Network Coding}
 We will use network coding on a \emph{flow-by-flow} (each flow consists of two bi-directional messages $W_{i,0}$ and $W_{0,i}$)  basis to improve achievable rates.
The relay $\nr$ uses \emph{decode and forward} (DF) scheme for broadcasting the received signals.  The relay $\nr$ estimates $\{w_{\nzero,i}\}$ and $\{w_{i,\nzero}\}$, at the end of the multiple access period, and constructs $w_{\nr_i} = w_{\nzero,i} \oplus w_{i,\nzero}$, $\forall i \in {\cal B}$. Next, the DF relay $\nr$ constructs $w_\nr = (w_{\nr_\none},w_{\nr_\ntwo},\cdots,w_{\nr_\nm})$ and broadcasts ${\bf x}_\nr(w_\nr)$ during the broadcast period.

\subsubsection{Random binning}
Random binning is not only used in a Marton-like fashion but is further used to
exploit side-information from direct links in the PMABC, FTDBC and PTDBC protocols. We apply random
binning to combine, at an end user,  the information received along from the direct link,  and that received along the relaying link.
 Similar to the TDBC protocol of the single pair case in \cite{SKim:2007}, we use random binning when the relay is transmitting to the destination nodes. Also, as in the Multi-Hop Multi-Relay protocol in \cite{SKim:2008b}, each terminal node listens to the other terminal nodes' signals whenever it is not transmitting itself, thereby building up side-information which may be exploited in decoding its own message after the relay broadcasting division/phase.

\subsubsection{Cooperation between terminal nodes}
As our channel model allows end users to over-hear the transmissions of the
base-station and other end-nodes when they are not transmitting, "over-heard" side information is
available at terminal nodes. This may be used to allow terminal nodes to cooperate in transmitting their messages. Cooperation is enabled through a compress and forward strategy in which each terminal node in ${\cal B}$ compresses the signals received during the relay broadcast period using an auxiliary message set, which it then transmits during the next multiple access period. If other nodes can decode this auxiliary message,  they are able to obtain the compressed received signals which in turn may be used to decode messages from the relay.

To concretely illustrate how our cooperation strategy operates, we describe cooperation for the PMABC protocol; the FTDBC protocol can be similarly constructed, and cooperation is impossible under the FMABC and PTDBC protocols (as there are no overhead messages).  We apply the \emph{sliding window} and \emph{Compress and Forward} schemes when node $i$ $(\in {\cal B})$ is transmitting: first, we divide the total time duration into $K+1$ slots and each slot consists of $m+1$ phases. Every message $w_{i,j}$ is also divided into $K$ blocks as $\{w_{i,j|(1)},\cdots,w_{i,j|(K)}\}$, and node $i$ transmits $\{w_{i,j|(k)}\}$ during slot $k$ and phase $i$ (for PMABC).
After relay $\nr$ broadcasts ${\bf x}_\nr$ during slot $k$ and phase $m+1$, node $i$ compress ${\bf y}_i$ to ${\hat {\bf y}}_i$ with auxiliary message set $\{w_{\{i\},{\cal B}}\}$. Then node $i$ broadcasts ${\bf x}_i(w_{i,\nzero|(k+1)},w_{\{i\},{\cal B}|(k)})$ during slot $k+1$ and phase $i$ except for the first and last slots. During the first and last slot, $i$ sends ${\bf x}_i(w_{i,\nzero|(1)},1)$ and ${\bf x}_i(1,w_{\{i\},{\cal B}|(K)})$, respectively.

In general, joint typicality is non-transitive. However, by using strong joint-typicality, and the fact that for the
distributions of interest $x \rightarrow y \rightarrow \hat{y}$, we will be able to argue
joint typicality between ${\bf x}$ and $\bf {\hat{y}}$ by the {\it Markov lemma} (Lemma 4.1
in \cite{berger:1977}). If node $j$ $(\in {\cal B}~, j\neq i)$ can decode $\tilde{w}_{\{i\},{\cal B}|(k)}$ at the end of slot $k+1$ and phase $i$, node $j$ can use the sequence ${\hat {\bf y}}^{(m+1)}_i(w_{\{i\},{\cal B}|(k)})$ for decoding $\tilde{w}_{\nzero,j|(k)}$. Let ${\cal J}_j$ be the set of nodes whose message can be decoded by node $j$, i.e., ${\cal J}_j = \{i|\tilde{w}_{\{i\},{\cal B}|(k)} = w_{\{i\},{\cal B}|(k)}~,\forall k\in[1,K+1]\}$. Then node $j$ uses the jointly typical sequences $({\bf x}^{(m+1)}_\nr(w_\nr),{\bf y}^{(m+1)}_{j},{\hat {\bf y}}^{(m+1)}_{{\cal J}_j}(w_{{\cal J}_j,{\cal B}|(k)}))$ to decode $\tilde{w}_{\nzero,j|(k)}$. \Fig \ref{fig:cooperation} illustrates an example of the PMABC protocol with $m=2$ terminal nodes (and hence four messages).

\begin{figure}[t]
 \begin{center}
  \epsfig{figure=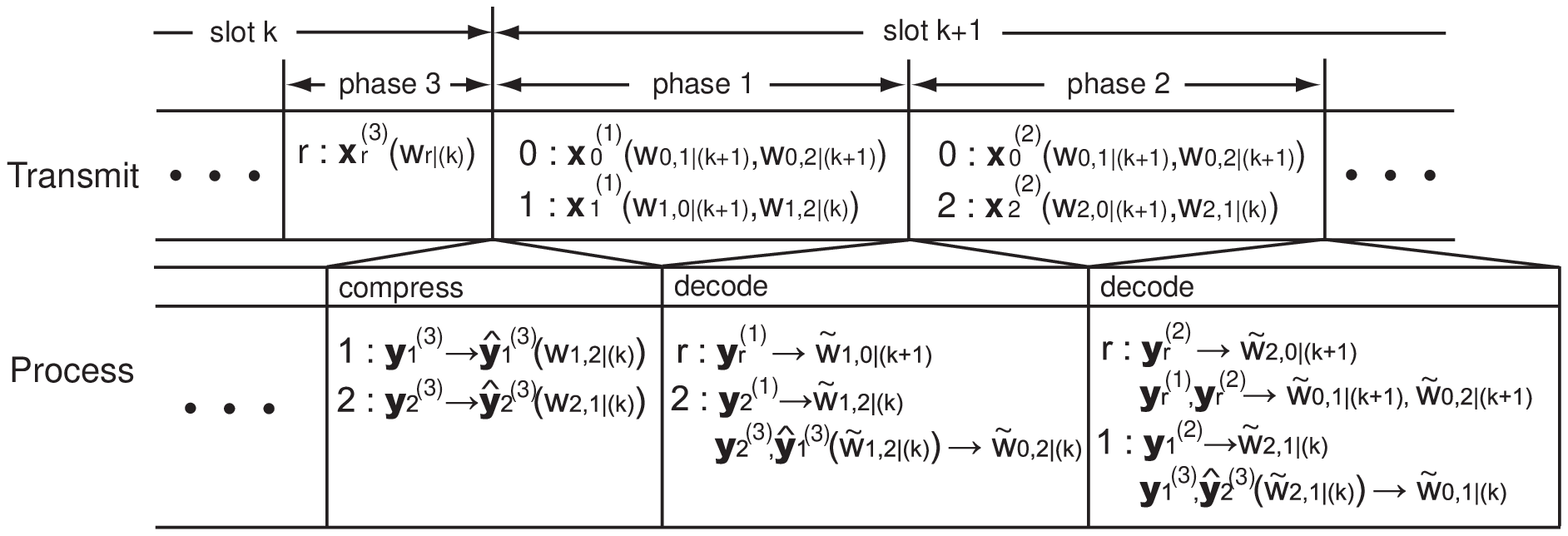, width=6in}
  \caption{An example of the PMABC protocol with $m=2$ and cooperation.}
  \label{fig:cooperation}
 \end{center}
\end{figure}

 In the remainder of this paper, we will be considering a number of different protocols combined with different encoding/decoding schemes. One of our goals is to determine whether all encoding schemes are useful and if so, under what channel conditions certain protocols/encoding schemes out-perform others. In Table \ref{table:protocol}, we summarize the considered protocols and corresponding coding schemes. The indices/notation should be read as ``N'' means Network coding, ``R'' means Random binning, and ``C'' means Cooperation between end nodes.
 % - the Table highlights which encoding schemes are used in which protocols.
  We will derive achievable rate regions for all protocols in the left hand column of Table \ref{table:protocol} and derive an outer bound for FMABC, PMABC, FTDBC and PTDBC separately. All outer bounds are variations of cut-set outer bounds along the lines of \cite{SKim:2007, SKim:2008a, SKim:2008b}.

\begin{table}
\caption{Protocols and coding schemes}
\label{table:protocol}
\centering
\begin{tabular}{l||c|c|c|c|c}
  \hline
  Protocol & Multiple Access & Marton's Broadcast & Network coding & Random binning & User cooperation \\
  \hline
  Simplest & -- & -- & -- & -- & -- \\
  \hline
  FMABC & X & X & -- & -- & --\\
  FMABC-N & X & X & X & -- & --\\
  \hline
  PMABC & X & X & -- & -- & --\\
  PMABC-NR & X & X & X & X & --\\
  PMABC-NRC & X & X & X & X & X\\
  \hline
  FTDBC & -- & X & -- & -- & --\\
  FTDBC-NR & -- & X & X & X & --\\
  FTDBC-NRC & -- & X & X & X & X\\
  \hline
  PTDBC & X & X & -- & -- & --\\
  PTDBC-NR & X & X & X & X & --\\
  \hline
\end{tabular}
\end{table}

%%%%%%%%%%%%%%%%%%%%%%%%%%%%%%%%%%%%%%%%%%%%%

\section{Extension of Marton's inner bound}

%%%%%%%%%%%%%%%%%%%%%%%%%%%%%%%%%%%%%%%%%%%%
\label{sec:marton}
Due to the presence of $>2$ independent messages at the relay who will wish to broadcast messages to the terminal and base-station nodes, for completeness, we first present a simple extension of Marton's achievable rate region for the 2 user discrete memoryless broadcast channel \cite{marton} to $m$ independent receivers in the absence of common information. This Theorem will be used in the derivation of the achievable rate regions - where it will be carefully combined with the unique side-information available in a bi-directional relay channel, i.e. each node has knowledge of its own message.
%  The key idea of binning the messages remains the same and results in a sum-rate constraint (see equation \eqref{eq:sum}).

\begin{theorem}
\label{theorem:gbc}
[Extended Marton] An achievable rate region for sending independent information over the broadcast channel $X \rightarrow (Y_1,Y_2,\cdots, Y_m)$ is the closure of all points $(R_1,\cdots, R_m)$ satisfying
\begin{align}
%R_i &< I(U_i;Y_i)~~~~~~~~~~~~~~~~~~~~~~~~~~~~~~~~~~~~\forall i\in[1,m]\\
\sum_{i\in S} R_i &< \sum_{i\in S} I(U_i;Y_i) -  I(U_i;U_{S(i)})~~~~~~~~~\forall S\subseteq {\cal B}\label{eq:sum}
\end{align}
over all joint distributions
$p(u_1,\cdots,u_m,x)p(y_1,\cdots,y_m|x)$, over the alphabet $\bigotimes_{i=1}^m {\cal U}_i \times {\cal X}\times \bigotimes_{i=1}^m {\cal Y}_i$. \thmend
\end{theorem}

\begin{proof}
{\em Random code generation: } For $\epsilon >0$, $n>0$ and $i\in [1,m]$, generate  $n$-length sequences ${\bf u}_i(w_{i0})$, $w_{i0}\in \{0,1,\cdots, \lfloor 2^{n(I(U_i;Y_i)-\epsilon)}\rfloor -1 \}$, each with probability
\begin{align}
P({\bf u}_i) =\left\{
              \begin{array}{ll}
                \frac{1}{\|A(U_i) \|}, & {\bf u}_i\in A(U_i) \\
                0, & \hbox{otherwise.}
              \end{array}
            \right.
\end{align}
where $\|A\|$ is the size of the set $A$. Also define bin $B_j^i \eqdef \{w_{i0} | w_{i0} \in [j\cdot\lfloor2^{n(I(U_i;Y_i)-R_i-\epsilon)}\rfloor, (j+1)\cdot\lfloor2^{n(I(U_i;Y_i)-R_i-\epsilon)}\rfloor-1]\}$ for $j\in \{0,1,\cdots,\lfloor 2^{nR_i}\rfloor -1\}$.

{\em Encoding: }
To transmit a pair of messages $(w_1,\cdots,w_m)$, pick a pair $(w_{10},\cdots,w_{m0})\in\bigotimes_{i=1}^m B^i_{w_i}$ which satisfies ${\bf u}_S(w_{S0})\in A(U_S) $ $\forall S\subseteq {\cal B}$, $|S|>1$, where $w_{S0}\eqdef \{w_{i0}|i\in S\}$. Such a $(w_{10},\cdots,w_{m0})$ exists with high probability if
\begin{align}
\sum_{i\in S} R_i &< \sum_{i\in S}\left( I(U_i;Y_i) - I(U_i;U_{S(i)}) \right) -|S|\epsilon -\delta(\epsilon)  ~~~~ \forall S\subseteq {\cal B}~,~|S|>1
\end{align}
from Lemma \ref{lemma:gbc} in the Appendix \ref{app:lemma}.
Then send an ${\bf x}$ which is jointly typical with $({\bf u}_1(w_{10}),\cdots,{\bf u}_m(w_{m0}))$.

{\em Decoding: } Receiver $i$ decodes $w_{i0}$ using jointly typical decoding of the sequence $({\bf u}_i,{\bf y}_i)$.

{\em Error analysis: }
We have the following error events :

\begin{description}
\item[$E_{\text{en}}$ ]: there does not exist a pair $(w_{10},\cdots,w_{m0}) \in \bigotimes_{i=1}^m B^i_{w_i}$ such that ${\bf u}_S(w_{S0}) \in A(U_S)$ , $\forall S\subseteq {\cal B},$ $|S|>1$.
\item[$E_{\text{emp}}$ ]: $({\bf u}_1(w_{10}),\cdots,{\bf u}_m(w_{m0}),{\bf x},{\bf y}_1,\cdots,{\bf y}_m) \not \in A(U_1 \cdots U_m \, X \, Y_1 \cdots Y_m)$.
\item[$E_i$]: there exists ${\tilde w}_{i0} \neq w_{i0}$ such that $({\bf u}_{i0}({\tilde w}_{i0}),{\bf y}_i)  \in A(U_i Y_i)$.
\end{description}

Then,
%edit:sj start
\begin{align}
P[E] \leq & P[E_{\text{en}}]+P[E_{\text{emp}}]+\sum_{i=1}^m P[E_i]\\
\leq & \epsilon' + \sum_{i=1}^m 2^{n(R_i - I(U_i;Y_i) + \epsilon)} \label{eq:gbc:1}
\end{align}
%edit:sj end

Since $\epsilon > 0$ is arbitrary, the conditions of Theorem \ref{theorem:gbc} and the AEP property guarantee that the right hand sides of \eqref{eq:gbc:1} as $n \rightarrow \infty$.
\end{proof}

%%%%%%%%%%%%%%%%%%%%%%%%%%%%%%%%%%%%%%%%%%%%%%%%%%%%%%%%%%%%%%%%%

\section{Achievable rate regions}

%%%%%%%%%%%%%%%%%%%%%%%%%%%%%%%%%%%%%%%%%%%%%%%%%%%%%%%%%%%%%%%%%
\label{sec:bounds}
 We first derive achievable rate regions for the {\it simplest,} FMABC, PMABC, FTDBC, and PTDBC protocols using only conventional (for comparison) multiple access channel (MAC) and broadcast channel (BC) coding techniques - along the lines of Theorem 15.3.6 in \cite{Cover:2006} for the MAC and the extended Marton's region of Theorem \ref{theorem:gbc} for the BC, as all protocols may be seen as combinations of MACs and BCs in the different phases. Due to their simplicity, the proofs for these simple cases are omitted and may be obtained as extensions of \cite{SKim:2007,SKim:2008a, SKim:2008b}.

 Subsequently, we ask whether  these rate regions may be improved upon by using the more elaborate coding techniques which exploit over-heard side information and own-message side-information as previously described: network coding, random binning and user cooperation. We then obtain achievable rate regions for the different protocols using different combinations of encoding schemes.
 The proofs of the schemes are included in the Appendix.

\subsection{Simplest Protocol}
\begin{theorem}
\label{theorem:simple}
An achievable rate region of the half-duplex bi-directional relay channel
under the simplest protocol with decode and forward relaying is the closure of the set of all points $(R_{\nzero,b},R_{b,\nzero})$ for all $b\in {\cal B}$ satisfying
\begin{align}
R_{\{\nzero\},{\cal B}} &< \Delta_1 I(X_\nzero^\pa ; Y_{\nr}^\pa ) \label{eq:simple:1}\\
R_{i,\nzero} &< \Delta_{i+1} I(X_i^{(i+1)};Y_\nr^{(i+1)})\label{eq:simple:2}\\
R_{{\cal B},\{\nzero\}} &< \Delta_{m+2} I(X_\nr^{(m+2)} ; Y_{\nzero}^{(m+2)} ) \label{eq:simple:3}\\
R_{\nzero,i} &< \Delta_{m+i+2} I(X_\nr^{(m+i+2)};Y_i^{(m+i+2)})\label{eq:simple:4}
\end{align}
for $i\in {\cal B}$ over all joint distributions
$\prod_{i=0}^m p^{(i+1)}(x_i)p^{(m+i+2)}(x_\nr)$, over the alphabet $\bigotimes_{i=0}^m {\cal X}_i \times {\cal X}_\nr$. \thmend
\end{theorem}

\subsection{FMABC Protocol}

\begin{theorem}
\label{theorem:FMABC}
An achievable rate region of the half-duplex bi-directional relay channel
under the FMABC protocol with decode and forward relaying is the closure of the set of all points $(R_{\nzero,b},R_{b,\nzero})$ for all $b\in {\cal B}$ satisfying
\begin{align}
R_{S,{\cal M}} &< \Delta_1 I(X_{S}^\pa ; Y_{\nr}^\pa | X_{\bar{S}}^\pa,Q) \label{eq:FMABC:1}\\
R_{{\cal M},S} &< \sum_{i\in S} \Delta_2 I(U_i^\pb;Y_i^\pb) -  \Delta_2 I(U_i^\pb;U_{S(i)}^\pb)\label{eq:FMABC:2}
\end{align}
for $S\subseteq {\cal M}$ over all joint distributions
$p(q)\prod_{i=0}^m p^\pa(x_i|q) p^\pb(u_0,\cdots,u_m,x_\nr)$, where $U_j$'s are the auxiliary random variables with $|{\cal
Q}| \leq 2^{m+1}-1$ over the alphabet $(\bigotimes_{i=0}^m {\cal X}_i \times {\cal U}_i) \times{\cal X}_\nr \times {\cal Q}$. \thmend
\end{theorem}

\subsection{FMABC-N Protocol}

We consider the FMABC protocol in which Network coding is employed at the relay to combine messages on a flow-by-flow basis - i.e. the message from node $i(\in {\cal B})$ to node $\nzero$ and vice-versa are combined at the relay. In the following theorem, the $U_i$ variables are the auxiliary random variables similar to those seen in the extension of Marton's region of Theorem \ref{theorem:gbc}.
\begin{theorem}
\label{theorem:FMABC-N}
An achievable rate region of the half-duplex bi-directional relay channel
under the FMABC-N protocol with decode and forward relaying is the closure of the set of all points $(R_{\nzero,b},R_{b,\nzero})$ for all $b\in {\cal B}$ satisfying
\begin{align}
R_{S,{\cal M}} &< \Delta_1 I(X_{S}^\pa ; Y_{\nr}^\pa | X_{\bar{S}}^\pa,Q) \label{eq:FMABC:DF:GC:1}\\
R_{\{\nzero\},T} &< \sum_{i\in T} \Delta_2 I(U_i^\pb;Y_i^\pb) -  \Delta_2 I(U_i^\pb;U_{T(i)}^\pb)\label{eq:FMABC:DF:GC:6}\\
R_{T,\{\nzero\}} &< \Delta_2 I(U_T^\pb ; Y_{\nzero}^\pb, U_{\bar T}^\pb ) \label{eq:FMABC:DF:GC:3}
\end{align}
for $S\subseteq {\cal M}$ and $T\subseteq {\cal B}$ over all joint distributions
$p(q)\prod_{i=0}^m p^\pa(x_i|q) p^\pb(u_1,\cdots,u_m,x_\nr)$, where $U_j$'s are the auxiliary random variables with $|{\cal
Q}| \leq 2^{m+1}-1$ over the alphabet $\bigotimes_{i=0}^m {\cal X}_i \times\bigotimes_{j=1}^m {\cal U}_j \times{\cal X}_\nr \times {\cal Q}$. \thmend
\end{theorem}

The rigorous proof is provided in Appendix \ref{app:FMABC-N}. We note that for the FMABC protocol only FMABC and FMABC-N (with network coding) are possible as there is no over-heard side information: during each phase every node is either transmitting or receiving - none are just listening. Hence, Random Binning and Cooperation schemes are impossible. Under the PMABC protocol however, Network Coding, Random Binning and Cooperation are all possible. We describe protocols with Network coding and Random binning, with and without Cooperation next.

\subsection{PMABC Protocol}

\begin{theorem}
\label{theorem:PMABC}
An achievable rate region of the half-duplex bi-directional relay channel
under the PMABC protocol with decode and forward relaying is the closure of the set of all points $(R_{\nzero,b},R_{b,\nzero})$ for all $b\in {\cal B}$ satisfying
\begin{align}
R_{\nzero,i} &< \Delta_i I(X_{\nzero}^{(i)} ; Y_{\nr}^{(i)} | X_{i}^{(i)},Q) \label{eq:PMABC:1}\\
R_{i,\nzero} &< \Delta_i I(X_{i}^{(i)} ; Y_{\nr}^{(i)} | X_{\nzero}^{(i)},Q) \label{eq:PMABC:2}\\
R_{\nzero,i}+ R_{i,\nzero} &< \Delta_i I(X_{\nzero}^{(i)},X_{i}^{(i)} ; Y_{\nr}^{(i)} | Q) \label{eq:PMABC:3}\\
R_{{\cal M},S} &< \sum_{i\in S} \Delta_{m+1} I(U_i^{(m+1)};Y_i^{(m+1)}) -  \Delta_{m+1} I(U_i^{(m+1)};U_{S(i)}^{(m+1)})\label{eq:PMABC:4}
\end{align}
for $i\in {\cal B}$ and $S\subseteq {\cal M}$ over all joint distributions
$p(q)\prod_{i=1}^m p^{(i)}(x_\nzero|q)p^{(i)}(x_i|q) p^{(m+1)}(u_0,\cdots,u_m,x_\nr)$, where $U_j$'s are the auxiliary random variables with $|{\cal
Q}| \leq 3m$ over the alphabet $(\bigotimes_{i=0}^m {\cal X}_i \times {\cal U}_i) \times{\cal X}_\nr \times {\cal Q}$. \thmend
\end{theorem}

\subsection{PMABC-NR Protocol}
We consider the PMABC protocol in which Network coding is employed at the relay to combine messages on a flow-by-flow basis,   along with Random Binning at the base-station node $\nzero$ to allow the end-nodes to exploit information over-heard in the phases during which they are not transmitting.  In the following theorem, the $U_i$ variables are the auxiliary random variables similar to those seen in the extension of Marton's region of Theorem \ref{theorem:gbc}, while $V_{0i}$ are auxiliary random variables used for binning at the base-station node $\nzero$.

\begin{theorem}
\label{theorem:PMABC-NR}
An achievable rate region of the half-duplex bi-directional relay channel
under the PMABC-NR protocol is the closure of the set of all points $(R_{\nzero,b},R_{b,\nzero})$ for all $b\in {\cal B}$ satisfying
\begin{align}
R_{\{\nzero\},T} + R_{S,\{\nzero\}} &< \sum_{s\in S} \Delta_s I(V_{\nzero T}^{(s)},X_s^{(s)};Y_\nr^{(s)},V_{\nzero {\bar {T}}}^{(s)}|Q) + \sum_{s\in {\bar {S}}} \Delta_s I(V_{\nzero T}^{(s)};Y_\nr^{(s)},V_{\nzero {\bar {T}}}^{(s)}|X_s^{(s)},Q) \label{eq:PMABC:GC:1}\\
R_{\{\nzero\},S} &< \sum_{i\in S} \sum_{j=1}^m \left(\Delta_j I(V_{\nzero i}^{(j)};Y_i^{(j)}|Q) - \Delta_j I(V_{\nzero i}^{(j)};V_{\nzero S(i)}^{(j)}|Q)\right) +\nonumber\\
&~~~~~~~~~  \Delta_{m+1} I(U_{i}^{(m+1)};Y_i^{(m+1)}) - \Delta_{m+1} I(U_i^{(m+1)};U_{S(i)}^{(m+1)})\label{eq:PMABC:GC:4}\\
R_{S,\{\nzero\}} &< \Delta_{m+1} I(U_S^{(m+1)} ; Y_{\nzero}^{(m+1)},U_{\bar S}^{(m+1)} )\label{eq:PMABC:GC:5}
\end{align}
for all $S,T\subseteq {\cal B}$ over all joint distributions
$p(q)\cdot\left[\prod_{i=1}^m p^{(i)}(v_{\nzero \none},\cdots,v_{\nzero \nm},x_\nzero|q)p^{(i)}(x_i|q)\right]$ $\cdot p^{(m+1)}(u_1,\cdots,u_m,x_\nr)$, where $V_{\nzero j}$ are the Random binning auxiliary random variables at node $\nzero$, $U_j$'s are the auxiliary Marton-like random variables used at node $\nr$ and $V_{\nzero T} \eqdef\{V_{\nzero s}| s\in T\}$ with $|{\cal Q}| \leq 2^{2m}+2^m$ over the alphabet $\bigotimes_{i=0}^m {\cal X}_i \times\bigotimes_{j=1}^m \left({\cal V}_{\nzero j} \times{\cal U}_j \right)\times{\cal X}_\nr \times {\cal Q}$. \thmend
\end{theorem}

The rigorous proof is provided in Appendix \ref{app:PMABC-NR}.

\subsection{PMABC-NRC Protocol}

We now allow the terminal nodes to {\it Cooperate} with each other in resolving the messages $w_{0,i}, \, \forall i\in {\cal B}$ - in addition to the flow-by-flow Network coding at the relay and the Random Binning at the base-station.
%consider the PMABC protocol in which Network coding is employed at the relay to combine messages on a flow-by-flow basis,   along with Random Binning at the base-station node $\nzero$ to allow the end-nodes to exploit information over-heard in the phases during which they are not transmitting.
In the following theorem:

$\bullet$ The $U_i$  variables are the auxiliary random variables similar to those seen in the extension of Marton's region of Theorem \ref{theorem:gbc}.

$\bullet$ The $V_{0i}$ are auxiliary random variables used for broadcasting at the base-station node $\nzero$.

$\bullet$ The $V_{i1}$ are the auxiliary random variables for transmitting new information from node $i$ to node $\nzero$.

$\bullet$ The $V_{i2}$ are the auxiliary random variables for transmitting cooperative information from node $i$ to other terminal nodes in the set ${\cal I}_i$, defined as the set of nodes which can decode ${\hat {\bf y}}_i^{(m+1)}(w_{\{i\},{\cal B}})$ at the end of transmission of node $i$. In a similar vein, ${\cal J}_i$ as the set of nodes whose quantized channel output is used at node $i$. Thus, ${\cal I}_i$ and ${\cal J}_i$ satisfy
${\cal I}_i = \{j|i \in {\cal J}_j~,~\forall j\}.$ Note that we will derive achievable rate regions for {\it given} sets ${\cal I}_i$ and ${\cal J}_i$ and that the question of which sets are optimal will depend on the choice of metric, and are left open - in evaluating our bounds we union over all possible choices of these sets.

\begin{theorem}
\label{theorem:PMABC-NRC}
An achievable rate region of the half-duplex bi-directional relay channel
under the PMABC-NRC protocol is the closure of the set of all points $(R_{\nzero,b},R_{b,\nzero})$  under given sets ${\cal I}_b$ and ${\cal J}_b$ for all $b\in {\cal B}$ satisfying
\begin{align}
R_{\{\nzero\},S} &< \sum_{j=1}^m \Delta_j I(V_{\nzero S}^{(j)};Y_\nr^{(j)},V_{\nzero {\bar S}}^{(j)}|V_{j1}^{(j)},Q) \label{eq:PMABC:DF:GC:1}\\
R_{i,\nzero} &< \Delta_i I(V_{i1}^{(i)};Y_\nr^{(i)}|Q) \label{eq:PMABC:DF:GC:2}\\
R_{\{\nzero\},S} &< \sum_{i\in S} \left(\sum_{j\in {\cal J}_i} \Delta_j I(V_{\nzero i}^{(j)};Y_i^{(j)}|V_{j2}^{(j)},Q) - \Delta_j I(V_{\nzero i}^{(j)};V_{\nzero S(i)}^{(j)}|Q) + \right.\nonumber\\
&~~~~~~~~~~~\left.\sum_{j\not\in {\cal J}_i} \Delta_j I(V_{\nzero i}^{(j)};Y_i^{(j)}|Q)  - \Delta_j I(V_{\nzero i}^{(j)};V_{\nzero S(i)}^{(j)}|Q)  \right)+ \nonumber \\
&~~~~~~~~ \Delta_{m+1} I(U_i^{(m+1)};Y_i^{(m+1)},{\hat Y}_{{\cal J}_i}^{(m+1)}) -  \Delta_{m+1}I(U_i^{(m+1)};U_{S(i)}^{(m+1)})\label{eq:PMABC:DF:GC:4}\\
R_{S,\{\nzero\}} &< \Delta_{m+1} I(U_S^{(m+1)} ; Y_{\nzero}^{(m+1)}, U_{\bar S}^{(m+1)} )\label{eq:PMABC:DF:GC:5}\\
R_{i,\nzero} &< \Delta_i I(V_{i1}^{(i)};Y_\nr^{(i)}|Q) + \Delta_i I(V_{i2}^{(i)};Y_{{\cal I}_i^{\min}}^{(i)}|Q) - \nonumber\\
&~~~~\Delta_i I(V_{i1}^{(i)};V_{i2}^{(i)}|Q) - \Delta_{m+1} I(Y_i^{(m+1)};{\hat Y}_i^{(m+1)}|Y_{{\cal I}_i^{\min}}^{(m+1)})\label{eq:PMABC:DF:GC:6}
\end{align}
subject to
\begin{align}
\Delta_{m+1} I(Y_i^{(m+1)};{\hat Y}_i^{(m+1)}|Y_j^{(m+1)}) < \Delta_i I(V_{i 2}^{(i)};Y_j^{(i)}|Q) ~~\text{ for all } j \in {\cal I}_i \label{eq:PMABC:DF:GC:7}
\end{align}
where ${\cal I}_i^{\min} = {\text{argmin}}_{j\in {\cal I}_i}\{\Delta_i I(V_{i2}^{(i)};Y_j^{(i)}|Q) - \Delta_{m+1} I(Y_i^{(m+1)};{\hat Y}_i^{(m+1)}| Y_j^{(m+1)})\}$ for all $i\in {\cal B}$ over all joint distributions
$p(q)\cdot\left(\prod_{i=1}^m p^{(i)}(v_{\nzero \none},\cdots,v_{\nzero \nm},x_\nzero|q)p^{(i)}(v_{i 1},v_{i 2},x_i|q) \right)$ $\cdot p^{(m+1)}(u_1,\cdots,u_m,x_\nr) \cdot p^{(m+1)}(y_{\cal B }|x_{\nr})$ $\cdot\left(\prod_{i=1}^m p^{(m+1)}({\hat y}_i|y_i)\right)$, where $V_{\nzero i}, V_{i 1}, V_{i 2}, U_i$'s are the auxiliary random variables and $V_{\nzero T} \eqdef\{V_{\nzero s}| s\in T\}$ with $|{\cal Q}| \leq 2^{m+1} + m^2 + m+2$ over the alphabet $\bigotimes_{i=0}^m {\cal X}_i \times\bigotimes_{j=1}^m \left({\cal V}_{\nzero j}\times {\cal V}_{j 1}\times {\cal V}_{j 2} \times{\cal U}_j\right) \times{\cal X}_\nr\times {\cal Q}$. \thmend
\end{theorem}
\begin{remark}
 \eqref{eq:PMABC:DF:GC:1} and \eqref{eq:PMABC:DF:GC:2} result from the multiple access period
 \footnote{\eqref{eq:PMABC:DF:GC:1} and \eqref{eq:PMABC:DF:GC:2} are suboptimal for multiple access channel. Generally this is the channel in which two transmitters are simultaneously broadcasting to multiple receivers. We introduce a simpler suboptimal scheme here and leave the optimization for the future work.}, while \eqref{eq:PMABC:DF:GC:4} and \eqref{eq:PMABC:DF:GC:5} result from the relay broadcast period. Also, \eqref{eq:PMABC:DF:GC:6} and  \eqref{eq:PMABC:DF:GC:7} result from the cooperation between terminal nodes. The rigorous proof is provided in Appendix \ref{app:PMABC-NRC}.
\end{remark}

\subsection{FTDBC Protocol}

\begin{theorem}
\label{theorem:FTDBC}
An achievable rate region of the half-duplex bi-directional relay channel
under the FTDBC protocol with decode and forward relaying is the closure of the set of all points $(R_{\nzero,b},R_{b,\nzero})$ for all $b\in {\cal B}$ satisfying
\begin{align}
R_{\{\nzero\},{\cal B}} &< \Delta_1 I(X_{\nzero}^\pa ; Y_{\nr}^\pa) \label{eq:FTDBC:1}\\
R_{i,\nzero} &< \Delta_{i+1} I(X_i^{(i+1)}; Y_\nr^{(i+1)}) \label{eq:FTDBC:2}\\
R_{{\cal M},S} &< \sum_{i\in S} \Delta_{m+2} I(U_i^{(m+2)};Y_i^{(m+2)}) -  \Delta_{m+2} I(U_i^{(m+2)};U_{S(i)}^{(m+2)})\label{eq:FTDBC:3}
\end{align}
for $i\in {\cal B}$ and $S\subseteq {\cal M}$ over all joint distributions
$\prod_{i=0}^m p^{(i+1)}(x_i) p^{(m+2)}(u_0,\cdots,u_m,x_\nr)$, where $U_j$'s are the auxiliary random variables over the alphabet $(\bigotimes_{i=0}^m {\cal X}_i \times {\cal U}_i) \times{\cal X}_\nr$. \thmend
\end{theorem}

\subsection{FTDBC-NR Protocol}
We next consider the FTDBC protocol in which Network coding is employed at the relay to combine messages on a flow-by-flow basis,   along with Random Binning at the base-station node $\nzero$ to allow the end-nodes to exploit information over-heard in the phases during which they are not transmitting.  In the following theorem, the $U_i$ variables are the auxiliary random variables similar to those seen in the extension of Marton's region of Theorem \ref{theorem:gbc}, while $V_{0i}$ are auxiliary random variables used for broadcasting at the base-station node $\nzero$.

\begin{figure}[t]
 \begin{center}
  \epsfig{figure=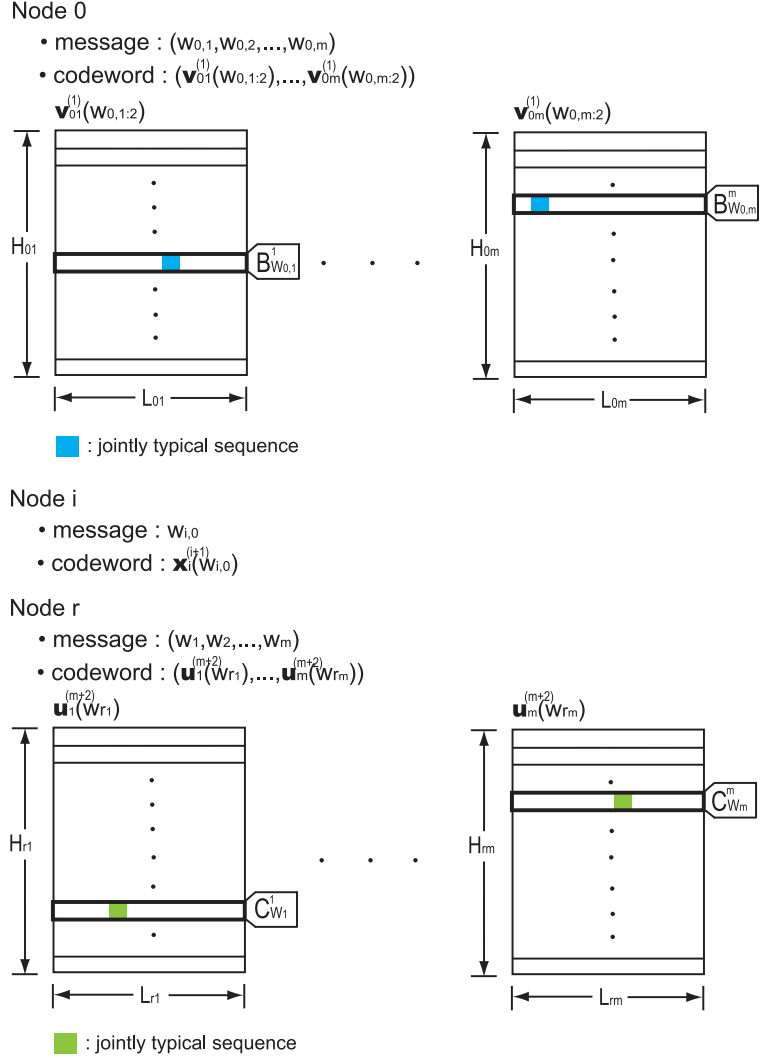, width=10cm}
  \caption{A diagram of encoders in the FTDBC-NR protocol with $\frac1n\log L_{\nzero i} = \Delta_1 I(V_{\nzero i}^\pa;Y_i^\pa) + R_{\nzero,i:1} - R_{\nzero,i}$, $\frac1n \log H_{\nzero,i} = R_{\nzero,i}$, $\frac1n\log L_{\nr i} = \Delta_{m+2} I(U_i^{(m+2)};Y_i^{(m+2)}) - R_{\nzero,i:1}$ and $\frac1n \log H_{\nr i} = \max \{R_{\nzero,i:1},R_{i,\nzero:1}\}$ for all $i \in [1,m]$. Also, $B^i_{w_{\nzero,i}}$ is the subset (bin) of the set $\{w_{\nzero,i:2}\}$ indexed by $w_{\nzero,i}$. Similarly, $C^i_{w_{i}}$ is the subset (bin) of the set $\{w_{\nr_i}\}$ indexed by $w_{i}$. }
  \label{fig:protocol_ftdbc}
 \end{center}
\end{figure}

\begin{theorem}
\label{theorem:FTDBC-NR}
An achievable rate region of the half-duplex bi-directional relay channel
under the FTDBC-NR protocol is the closure of the set of all points  $(R_{\nzero,b},R_{b,\nzero})$ for all $b\in {\cal B}$ satisfying
\begin{align}
R_{\{\nzero\},S} &< \Delta_1 I(V_{\nzero S}^\pa;Y_\nr^\pa, V_{\nzero {\bar S}}^\pa) \label{eq:FTDBC:GC:1}\\
R_{i,\nzero}&< \Delta_{i+1} I(X_i^{(i+1)};Y_\nr^{(i+1)})\label{eq:FTDBC:GC:2}\\
R_{\{\nzero\},S} &< \sum_{i\in S} \Delta_1 I(V_{\nzero i}^\pa;Y_i^\pa) - \Delta_1 I(V_{\nzero i}^\pa;V_{\nzero S(i)}^\pa) + \nonumber\\
&~~~~~~~~~\Delta_{m+2} I(U_i^{(m+2)};Y_i^{(m+2)}) - \Delta_{m+2}I(U_i^{(m+2)};U_{S(i)}^{(m+2)}) \label{eq:FTDBC:GC:4}\\
R_{S,\{\nzero\}} &< \sum_{i\in S} \Delta_{i+1} I(X_i^{(i+1)};Y_\nzero^{(i+1)}) + \Delta_{m+2} I(U_S^{(m+2)};Y_\nzero^{(m+2)},U_{\bar S}^{(m+2)} )\label{eq:FTDBC:GC:5}
\end{align}
for all $i \in {\cal B}$ and $S \subseteq {\cal B}$ over all joint distributions
$p^\pa(v_{\nzero \none},\cdots,v_{\nzero \nm},x_\nzero)$ $ \cdot \left( \prod_{j=1}^m p^{(j+1)}(x_j)\right)\cdot$ \\$p^{(m+2)}(u_\none,\cdots,u_\nm,x_\nr)$, where $V_{\nzero j}, U_j$'s are the auxiliary random variables and $V_{\nzero T} \eqdef\{V_{\nzero s}| s\in T\}$ over the alphabet $\bigotimes_{i=0}^m {\cal X}_i \times\bigotimes_{j=1}^m ({\cal V}_{\nzero j}\times {\cal U}_j) \times{\cal X}_\nr$. \thmend
\end{theorem}

\begin{remark}
\eqref{eq:FTDBC:GC:1} and \eqref{eq:FTDBC:GC:2} correspond to the transmissions from ${\cal M}$ to the relay $\nr$, while \eqref{eq:FTDBC:GC:4} -- \eqref{eq:FTDBC:GC:5} correspond to the relay broadcast phase. An example of the encoding is provided in  \Fig \ref{fig:protocol_ftdbc}, where we see the different auxiliary random variables at nodes $\nzero$, $i$ and the relay $\nr$. The rigorous proof is provided in Appendix \ref{app:FTDBC-NR}.
\end{remark}

\subsection{FTDBC-NRC protocol}

We allow the terminal nodes to {\it Cooperate} with each other in resolving the messages $w_{0,i}, \, \forall i\in {\cal B}$ - in addition to the flow-by-flow Network coding at the relay and the Random Binning at the base-station.
%consider the PMABC protocol in which Network coding is employed at the relay to combine messages on a flow-by-flow basis,   along with Random Binning at the base-station node $\nzero$ to allow the end-nodes to exploit information over-heard in the phases during which they are not transmitting.
In the following theorem the interpretations of the auxiliary random variables $U_i, V_{0i}, V_{i1}, V_{i2}$ are the same as those of Theorem \ref{theorem:PMABC-NRC}.

\begin{theorem}
\label{theorem:FTDBC-NRC}
An achievable rate region of the half-duplex bi-directional relay channel
under the FTDBC-NRC protocol is the closure of the set of all points  $(R_{\nzero,b},R_{b,\nzero})$  under given sets ${\cal I}_b$ and ${\cal J}_b$ for all $b\in {\cal B}$ satisfying
\begin{align}
R_{\{\nzero\},S} &< \Delta_1 I(V_{\nzero S}^\pa;Y_\nr^\pa, V_{\nzero {\bar S}}^\pa) \label{eq:FTDBC:GC:CO:1}\\
R_{i,\nzero}&< \Delta_{i+1} I(V_{i1}^{(i+1)};Y_\nr^{(i+1)})\label{eq:FTDBC:GC:CO:2}\\
R_{\{\nzero\},S} &< \sum_{i\in S} \Delta_1 I(V_{\nzero i}^\pa;Y_i^\pa) - \Delta_1 I(V_{\nzero i}^\pa;V_{\nzero S(i)}^\pa) + \nonumber\\
&~~~~~~~~~\Delta_{m+2} I(U_i^{(m+2)};Y_i^{(m+2)},{\hat Y}_{{\cal J}_i}^{(m+2)}) - \Delta_{m+2}I(U_{i}^{(m+2)};U_{S(i)}^{(m+2)})\label{eq:FTDBC:GC:CO:4}\\
R_{S,\{\nzero\}} &< \sum_{i\in S} \Delta_{i+1} I(V_{i1}^{(i+1)};Y_\nzero^{(i+1)}) + \Delta_{m+2} I(U_S^{(m+2)};Y_\nzero^{(m+2)},U_{\bar S}^{(m+2)} )\label{eq:FTDBC:GC:CO:5}\\
R_{i,\nzero} &< \Delta_{i+1} I(V_{i1}^{(i+1)};Y_\nr^{(i+1)}) + \Delta_{i+1} I(V_{i2}^{(i+1)};Y_{{\cal I}_i^{\min}}^{(i+1)}) - \nonumber\\
&~~~~\Delta_{i+1} I(V_{i1}^{(i+1)};V_{i2}^{(i+1)}) - \Delta_{m+2} I(Y_i^{(m+2)};{\hat Y}_i^{(m+2)}|Y_{{\cal I}_i^{\min}}^{(m+2)}) \label{eq:FTDBC:GC:CO:6}
\end{align}
subject to
\begin{align}
\Delta_{m+2} I(Y_i^{(m+2)};{\hat Y}_i^{(m+2)}|Y_j^{(m+2)}) < \Delta_{i+1} I(V_{i 2}^{(i+1)};Y_j^{(i+1)}) ~~\text{ for all } j \in {\cal I}_i \label{eq:FTDBC:GC:CO:7}
\end{align}
where ${\cal I}_i^{\min} = {\text{argmin}}_{j\in {\cal I}_i}\{\Delta_{i+1} I(V_{i2}^{(i+1)};Y_j^{(i+1)}) - \Delta_{m+2} I(Y_i^{(m+2)};{\hat Y}_i^{(m+2)}| Y_j^{(m+2)})\}$ for all $i \in {\cal B}$ and $S\subseteq {\cal B}$ over all joint distributions
$p^{(1)}(v_{\nzero \none},\cdots,v_{\nzero\nm},x_\nzero)\prod_{j=1}^m p^{(j+1)}(v_{j 1},v_{j 2},x_j)$ $p^{(m+2)}(u_1,\cdots,u_m,x_\nr)$\\$ p^{(m+2)}(y_{\cal B}|x_{\nr})\cdot\prod_{k=1}^m p^{(m+2)}({\hat y}_k|y_k)$, where $V_{\nzero j}, V_{j 1}, V_{j 2}, U_j$'s are the auxiliary random variables and $V_{\nzero T} \eqdef\{V_{\nzero s}| s\in T\}$ over the alphabet $\bigotimes_{i=0}^m {\cal X}_i \times\bigotimes_{j=1}^m ({\cal V}_{\nzero j}\times {\cal V}_{j 1}\times{\cal V}_{j 2} \times{\cal U}_j) \times{\cal X}_\nr$. \thmend
\end{theorem}

\emph{Proof outline :} The proof of Theorem \ref{theorem:FTDBC-NRC} follows the same argument as the proof of Theorem 3 in \cite{SKim:2007} and Theorem \ref{theorem:PMABC-NRC}.

\subsection{PTDBC Protocol}

\begin{theorem}
\label{theorem:PTDBC}
An achievable rate region of the half-duplex bi-directional relay channel
under the PTDBC protocol with decode and forward relaying is the closure of the set of all points $(R_{\nzero,b},R_{b,\nzero})$ for all $b\in {\cal B}$ satisfying
\begin{align}
R_{\{\nzero\},{\cal B}} &< \Delta_1 I(X_{\nzero}^\pa ; Y_{\nr}^\pa) \label{eq:PTDBC:1}\\
R_{T,\{\nzero\}} &< \Delta_{2} I(X_T^\pb; Y_\nr^\pb|X_{\bar T}^\pb,Q) \label{eq:PTDBC:2}\\
R_{{\cal M},S} &< \sum_{i\in S} \Delta_{3} I(U_i^\pc;Y_i^\pc) -  \Delta_{3} I(U_i^\pc;U_{S(i)}^\pc)\label{eq:PTDBC:3}
\end{align}
for $S\subseteq {\cal M}$ and $T\subseteq {\cal B}$ over all joint distributions
$p(q)p^\pa(x_\nzero|q)\prod_{i=1}^m p^\pb(x_i|q) p^\pc(u_0,\cdots,u_m,x_\nr)$, where $U_j$'s are the auxiliary random variables with $|{\cal Q}|\leq 2^m-1$ over the alphabet $(\bigotimes_{i=0}^m {\cal X}_i \times {\cal U}_i) \times{\cal X}_\nr\times {\cal Q}$. \thmend
\end{theorem}

\subsection{PTDBC-NR Protocol}
We next consider the PTDBC protocol in which Network coding is employed at the relay to combine messages on a flow-by-flow basis,   along with Random Binning at the base-station node $\nzero$ to allow the end-nodes to exploit information over-heard in the phases during which they are not transmitting.

\begin{theorem}
\label{theorem:PTDBC-NR}
An achievable rate region of the half-duplex bi-directional relay channel
under the PTDBC-NR protocol is the closure of the set of all points  $(R_{\nzero,b},R_{b,\nzero})$ for all $b\in {\cal B}$ satisfying
\begin{align}
R_{\{\nzero\},S} &< \Delta_1 I(V_{\nzero S}^\pa;Y_\nr^\pa, V_{\nzero {\bar S}}^\pa) \label{eq:PTDBC:GC:1}\\
R_{S,\{\nzero\}}&< \Delta_{2} I(X_S^\pb;Y_\nr^\pb|X_{\bar S}^\pb,Q)\label{eq:PTDBC:GC:2}\\
R_{\{\nzero\},S} &< \sum_{i\in S} \Delta_1 I(V_{\nzero i}^\pa;Y_i^\pa) - \Delta_1 I(V_{\nzero i}^\pa;V_{\nzero S(i)}^\pa) + \Delta_{3} I(U_i^\pc;Y_i^\pc) - \Delta_{3}I(U_i^\pc;U_{S(i)}^\pc) \label{eq:PTDBC:GC:4}\\
R_{S,\{\nzero\}} &<  \Delta_{2} I(X_S^\pb;Y_\nzero^\pb|X_{\bar S}^\pb,Q) + \Delta_{3} I(U_S^\pc;Y_\nzero^\pc,U_{\bar S}^\pc )\label{eq:PTDBC:GC:5}
\end{align}
for all $S \subseteq {\cal B}$ over all joint distributions
$p(q)p^\pa(v_{\nzero \none},\cdots,v_{\nzero \nm},x_\nzero)$ $ \cdot \left( \prod_{j=1}^m p^\pb(x_j|q)\right)\cdot$ \\$p^\pc(u_\none,\cdots,u_\nm,x_\nr)$, where $V_{\nzero j}, U_j$'s are the auxiliary random variables and $V_{\nzero T} \eqdef\{V_{\nzero s}| s\in T\}$ with $|{\cal Q}|\leq 2^m-1$ over the alphabet $\bigotimes_{i=0}^m {\cal X}_i \times\bigotimes_{j=1}^m ({\cal V}_{\nzero j}\times {\cal U}_j) \times{\cal X}_\nr \times {\cal Q}$. \thmend
\end{theorem}

\begin{remark}
\eqref{eq:PTDBC:GC:1} and \eqref{eq:PTDBC:GC:2} correspond to the transmissions from ${\cal M}$ to the relay $\nr$, while \eqref{eq:PTDBC:GC:4} -- \eqref{eq:PTDBC:GC:5} correspond to the relay broadcast phase. The proof of Theorem \ref{theorem:PTDBC-NR} follows the same argument as the proof of Theorem 3 in \cite{SKim:2007} and Theorem \ref{theorem:PMABC-NR}.
\end{remark}

%%%%%%%%%%%%%%%%%%%%%%%%%%%%%%%%%%%%%%%%%%%%%%%%%%%%%%%%%%%%%

\section{Outer Bounds}

%%%%%%%%%%%%%%%%%%%%%%%%%%%%%%%%%%%%%%%%%%%%%%%%%%%%%%%%%%%%%
\label{sec:outer_bounds}
We now derive outer bounds for the FMABC, PMABC, FTDBC and PTDBC protocols using the following
cut-set bound lemma tailored to multi-phase protocols first derived in \cite{SKim:2007}, and included for completeness. Given subsets $S,T
\subseteq {\cal M} = \{1,2,\cdots,m\}$, and $\bar{S}\eqdef {\cal M}\backslash S$, we define $W_{S,T} \eqdef
\{W_{i,j} | i\in S , j\in T\}$ and $R_{S,T} = \lim_{n\rightarrow
\infty }\frac1n H(W_{S,T})$.

\begin{lemma}
\label{lemma:out} If in some network the information rates
$\{R_{i,j}\}$ are achievable for a protocol $\prot$ with relative
phase durations $\{\Delta_\ell\}$, then for every $\epsilon>0$ and all $S
\subset  {\cal M}$
\begin{align}
 R_{S,\bar{S}} \leq
   \sum_\ell \Delta_\ell I(X^{(\ell)}_{S};Y^{(\ell)}_{\bar{S}}|X^{(\ell)}_{\bar{S}},Q)+\epsilon,
\end{align}
for a family of conditional distributions $p^{(\ell)}(x_1, x_2,
\ldots, x_m|q)$ and a discrete time-sharing random variable $Q$ with
distribution $p(q)$. Furthermore, each $p^{(\ell)}(x_1, x_2, \ldots,
x_m|q)p(q)$ must satisfy the constraints of phase $\ell$ of protocol
$\prot$. \thmend
\end{lemma}
The FMABC, PMABC, FTDBC and PTDBC outer bounds are obtained by applying the above lemma to the different protocols; in each protocol the ``cuts'' will look different depending on what nodes are permitted to transmit during each phase, leading to different outer bound regions.

\subsection{FMABC protocol}
\begin{theorem}
\label{theorem:FMABC:out} (Outer bound) The capacity region of the
half-duplex bi-directional relay channel under the FMABC protocol is outer bounded by the set of all points $(R_{\nzero,b},R_{b,\nzero})$ for all $b\in {\cal B}$ satisfying
\begin{align}
R_{\{\nzero\},{\cal B}} &\leq \Delta_1 I(X_\nzero^\pa;Y_\nr^\pa| X_{\cal B}^\pa,Q)\label{eq:FMABC:out:1}\\
R_{{\cal B},\{\nzero\}} &\leq \Delta_2 I(X_\nr^\pb;Y_\nzero^\pb)\label{eq:FMABC:out:2}\\
R_{S,\{\nzero\}} &\leq \Delta_1 I(X_S^\pa;Y_\nr^\pa|X_{\bar{S}}^\pa,X_\nzero^\pa,Q)\label{eq:FMABC:out:3}\\
R_{\{\nzero\},S} &\leq \Delta_2 I(X_\nr^\pb;Y_S^\pb)\label{eq:FMABC:out:4}
\end{align}
for all choices of the joint distribution
$p(q)\prod_{i=0}^m p^\pa(x_i|q)$ $p^\pb(x_{\nr})$ with $|{\cal Q}| \leq 2^{m}$ over the restricted alphabet $\bigotimes_{i=0}^m {\cal X}_i \times {\cal X}_\nr\times {\cal Q}$ for all possible $S\subseteq {\cal B}$. \thmend
\end{theorem}

\begin{proof}
%We use Lemma~\ref{lemma:out} to prove the
%Theorem~\ref{theorem:FMABC:out}.
For every $S\subseteq {\cal B}$,
there exist 4 types of cut-sets: $S_1 =\{\nzero\} \cup S$, $S_2
=\{\nr\} \cup S$, $S_3 = \{\nzero,\nr\}\cup S$ and $S_4 = S$, with corresponding rate pairs $R_{S,\{\nzero\}}$ and $R_{\{\nzero\},S}$. By definition of the FMABC protocol,
\begin{align}
 Y_{\cal M}^\pa &= X_\nr^\pa = \varnothing \label{FMABC:GC:out:1} \\
 X_{\cal M}^\pb &= Y_\nr^\pb = \varnothing. \label{FMABC::GCout:2}
\end{align}

Thus, the corresponding outer bounds for a given subset $S$ are:
\begin{align}
S_1 &: R_{\{\nzero\},{\bar S}} \leq \Delta_1 I(X_\nzero^\pa, X_S^\pa; Y_\nr^\pa
|X_{\bar S}^\pa,Q) + \epsilon, \label{FMABC:GC:out:3}\\
S_2 &: R_{S,\{\nzero\}}\leq \Delta_2 I(X_\nr^\pb; Y_{\nzero}^\pb ,Y_{\bar S}^\pb )+\epsilon,
\label{FMABC:GC:out:4}\\
S_3 &: R_{\{\nzero\},{\bar S}}\leq \Delta_2 I(X_\nr^\pb; Y_{\bar S}^\pb )+\epsilon,
\label{FMABC:GC:out:5}\\
S_4 &: R_{S,\{\nzero\}}\leq \Delta_1 I(X_S^\pa; Y_\nr^\pa |X_{\bar S}^\pa,X_\nzero^\pa ,Q)+\epsilon,
\label{FMABC:GC:out:6}
\end{align}
In \eqref{FMABC:GC:out:3}, we have
\begin{align}
R_{\{\nzero\},{\bar S}} &\leq R_{\{\nzero\},{\cal B}} \\
&\leq \Delta_1 I(X_\nzero^\pa;Y_\nr^\pa|X_{\cal B}^\pa,Q) + \epsilon \\
&= \Delta_1 I(X_\nzero^\pa;Y_\nr^\pa|X_S^\pa,X_{\bar S}^\pa,Q) + \epsilon \\
&\leq \Delta_1 I(X_\nzero^\pa, X_S^\pa; Y_\nr^\pa
|X_{\bar S}^\pa,Q) + \epsilon
\end{align}
Similarly, in \eqref{FMABC:GC:out:4}, we have
\begin{align}
R_{S,\{\nzero\}} &\leq R_{{\cal B},\{\nzero\}} \\
&\leq \Delta_2 I(X_\nr^\pb;Y_\nzero^\pb) + \epsilon \\
&\leq \Delta_2 I(X_\nr^\pb;Y_\nzero^\pb,Y_{\bar S}^\pb) + \epsilon
\end{align}
Thus, the following inequalities capture all possible $S_1$'s and $S_2$'s, respectively.
\begin{align}
R_{\{\nzero\},{\cal B}}&\leq \Delta_1 I(X_\nzero^\pa;Y_\nr^\pa|X_{\cal B}^\pa,Q) + \epsilon \label{FMABC:GC:out:13}\\
R_{{\cal B},\{\nzero\}}&\leq \Delta_2 I(X_\nr^\pb;Y_\nzero^\pb) + \epsilon \label{FMABC:GC:out:12}
\end{align}
Since $\epsilon > 0$ is arbitrary, \eqref{FMABC:GC:out:6},
\eqref{FMABC:GC:out:13}, \eqref{FMABC:GC:out:12} and \eqref{FMABC:GC:out:5} yield Theorem~\ref{theorem:FMABC:out}. By
Fenchel-Bunt's extension of the Carath\'{e}odory theorem in
\cite{Hiriart:2001}, it is sufficient to restrict $|{\cal Q}| \leq
2^{m}$.
\end{proof}

\subsection{PMABC protocol}
\begin{theorem}
\label{theorem:PMABC:out} (Outer bound) The capacity region of the
half-duplex bi-directional relay channel under the PMABC protocol is outer bounded by the set of all points $(R_{\nzero,b},R_{b,\nzero})$ for all $b\in {\cal B}$ satisfying
\begin{align}
R_{\{\nzero\},{\cal B}}&\leq \sum_{i=1}^m \Delta_i I(X_\nzero^{(i)};Y_\nr^{(i)}|X_i^{(i)},Q) \label{eq:PMABC:out:1}\\
R_{{\cal B},\{\nzero\}}&\leq \Delta_{m+1} I(X_\nr^{(m+1)}; Y_\nzero^{(m+1)})\label{eq:PMABC:out:2}\\
R_{S,\{\nzero\}}&\leq \sum_{i\in S} \Delta_i I(X_i^{(i)};Y_\nr^{(i)},Y_{\bar S}^{(i)}|X_\nzero^{(i)},Q)\label{eq:PMABC:out:3}\\
R_{\{\nzero\},S} &\leq \sum_{i\in {\bar S}}\Delta_i I(X_\nzero^{(i)},X_i^{(i)};Y_S^{(i)}|Q) +\sum_{i\in {S}}\Delta_i I(X_\nzero^{(i)};Y_{S\setminus\{i\}}^{(i)}|X_i^{(i)},Q) + \nonumber\\
&~~~~\Delta_{m+1}I(X_\nr^{(m+1)};Y_S^{(m+1)}) \label{eq:PMABC:out:4}
\end{align}
for all choices of the joint distribution
$p(q)\prod_{i=1}^m p^{(i)}(x_\nzero|q)p^{(i)}(x_i|q)p^{(m+1)}(x_{\nr})$ with $|{\cal Q}| \leq 2^{m+1}-1$ over the restricted alphabet $\bigotimes_{i=0}^m {\cal X}_i \times {\cal X}_\nr\times {\cal Q}$ for all possible $S\subseteq {\cal B}$. \thmend
\end{theorem}

\emph{Proof outline :} The proof of Theorem \ref{theorem:PMABC:out} follows the same argument as the proof of Theorem \ref{theorem:FMABC:out}.

\subsection{FTDBC protocol}
\begin{theorem}
\label{theorem:FTCBC:out} (Outer bound) The capacity region of the
half-duplex bi-directional relay channel under the FTDBC protocol is outer bounded by the set of all points $(R_{\nzero,b},R_{b,\nzero})$ for all $b\in {\cal B}$ satisfying
\begin{align}
R_{\{\nzero\},S}&\leq \min \left\{\Delta_1 I(X_\nzero^\pa;Y_\nr^\pa,Y_S^\pa) + \sum_{i\in {\bar S}} \Delta_{i+1} I(X_i^{(i+1)};Y_\nr^{(i+1)},Y_S^{(i+1)}),\right.\nonumber\\
&~~~\left.\Delta_1 I(X_\nzero^\pa;Y_S^\pa) + \sum_{i\in {\bar S}} \Delta_{i+1} I(X_i^{(i+1)};Y_S^{(i+1)}) + \Delta_{m+2} I(X_\nr^{(m+2)} ;Y_S^{(m+2)}) \right\}\label{eq:FTDBC:out:1}\\
R_{S,\{\nzero\}}&\leq \min\left\{\sum_{i\in S} \Delta_{i+1}
I(X_i^{(i+1)}; Y_\nzero^{(i+1)}
,Y_{\bar S}^{(i+1)}) + \Delta_{m+2} I(X_\nr^{(m+2)};Y_\nzero^{(m+2)},Y_{\bar S}^{(m+2)}),\right.\nonumber\\
&~~~\left.\sum_{i\in S}
\Delta_{i+1}I(X_i^{(i+1)};Y_\nzero^{(i+1)},Y_{\bar
S}^{(i+1)},Y_\nr^{(i+1)})\right\} \label{eq:FTDBC:out:2}
\end{align}
for all choices of the joint distribution
$\prod_{i=0}^m p^{(i+1)}(x_i)$ $p^{(m+2)}(x_{\nr})$  over the restricted alphabet $\bigotimes_{i=0}^m {\cal X}_i \times {\cal X}_\nr$ for all possible $S\subseteq {\cal B}$. \thmend
\end{theorem}

\emph{Proof outline :} The proof of Theorem \ref{theorem:FTCBC:out} follows the same argument as the proof of Theorem \ref{theorem:FMABC:out}.

\subsection{PTDBC protocol}
\begin{theorem}
\label{theorem:PTCBC:out} (Outer bound) The capacity region of the
half-duplex bi-directional relay channel under the PTDBC protocol is outer bounded by the set of all points $(R_{\nzero,b},R_{b,\nzero})$ for all $b\in {\cal B}$ satisfying
\begin{align}
R_{\{\nzero\},S}&\leq \min \left\{\Delta_1 I(X_\nzero^\pa;Y_\nr^\pa,Y_S^\pa) + \Delta_{2} I(X_{\bar S}^\pb;Y_\nr^\pb|X_S^\pb,Q),\right.\nonumber\\
&~~~~~~~~~~\left.\Delta_1 I(X_\nzero^\pa;Y_S^\pa) + \Delta_3 I(X_\nr^\pc ;Y_S^\pc) \right\}\label{eq:PTDBC:out:1}\\
R_{S,\{\nzero\}}&\leq \min\left\{ \Delta_2 I(X_S^\pb; Y_\nzero^\pb|X_{\bar S}^\pb,Q) + \Delta_3 I(X_\nr^\pb;Y_\nzero^\pb,Y_{\bar S}^\pb),\right.\nonumber\\
&~~~~~~~~~~\left.\Delta_2I(X_S^\pb;Y_\nzero^\pb,Y_\nr^\pb|X_{\bar
S}^\pb,Q)\right\} \label{eq:PTDBC:out:2}
\end{align}
for all choices of the joint distribution
$p(q)p^\pa(x_\nzero)\prod_{i=1}^m p^\pb(x_i|q)$ $p^\pc(x_{\nr})$  with $|{\cal Q}| \leq 2^m -1$ over the restricted alphabet $\bigotimes_{i=0}^m {\cal X}_i \times {\cal X}_\nr \times {\cal Q}$ for all possible $S\subseteq {\cal B}$. \thmend
\end{theorem}

\emph{Proof outline :} The proof of Theorem \ref{theorem:PTCBC:out} follows the same argument as the proof of Theorem \ref{theorem:FMABC:out}.

These outer bounds will be evaluated in Gaussian noise, as described next.

%%%%%%%%%%%%%%%%%%%%%%%%%%%%%%%%%%%%%%%%%%%%%%%%%%%%%%%%%%%%%

\section{Gaussian noise channel}

%%%%%%%%%%%%%%%%%%%%%%%%%%%%%%%%%%%%%%%%%%%%%%%%%%%%%%%%%%%%%
\label{sec:Gaussian}
In this section, we evaluate the achievable rate regions and outer bounds obtained in the previous section to an AWGN channel.  To do so, we assume an additive white Gaussian noise (AWGN) channel model, assume Gaussian input distributions for the achievability schemes, which may or may not be optimal, and evaluate the mutual information terms. The resulting rate regions are indeed quite complex: our goals are to demonstrate, at least under Gaussian input distributions, that:

$\bullet$ Network coding, Random Binning, and Cooperation all improve the rate regions over simpler schemes which employ standard multiple access and broadcast channel coding schemes.

$\bullet$ There are no inclusion relationships between FMABC, PMABC, FTDBC and PTDBC regions.

Furthermore, we include entirely the somewhat laborious analytical expressions in the Gaussian noise regime; we hope that these expressions may be of use to the further exploration of the obtained regions in for example the high or low SNR regimes, which is the subject of current investigation \cite{kim:inprep}.

\subsection{Channel model}
The corresponding mathematical channel model is, for each channel use $k$ :
\begin{align}
{\bf Y}[k] = {\bf H}{\bf X}[k] + {\bf Z}[k]
\end{align}
where,
\begin{align}
{\bf Y}[k] = \left[
               \begin{array}{c}
                 Y_{\nzero}[k] \\
                 Y_{\none}[k] \\
                 {\vdots} \\
                 Y_{\nm}[k] \\
                 Y_{\nr}[k] \\
               \end{array}
             \right] ,~~~
{\bf X}[k] = \left[
               \begin{array}{c}
                 X_{\nzero}[k] \\
                 X_{\none}[k] \\
                 {\vdots} \\
                 X_{\nm}[k] \\
                 X_{\nr}[k]\\
               \end{array}
             \right] ,~~~
{\bf Z}[k] = \left[
               \begin{array}{c}
                 Z_{\nzero}[k] \\
                 Z_{\none}[k] \\
                 {\vdots} \\
                 Z_{\nm}[k] \\
                 Z_{\nr}[k]\\
               \end{array}
             \right]
\end{align}
and
\begin{align}
{\bf H} = \left[
            \begin{array}{cccc}
              0 & h_{\nzero,\none} & \cdots & h_{\nzero,\nr} \\
              h_{\none,\nzero} & 0 & \cdots & h_{\none,\nr} \\
              \vdots & \ddots & \vdots & \vdots \\
              h_{\nr,\nzero} & h_{\nr,\none} & \cdots & 0 \\
            \end{array}
          \right]
\end{align}
where ${\bf Y}[k]$, ${\bf X}[k]$ and ${\bf Z}[k]$ are in $(\complex^*)^{(m+2) \times 1} = (\complex \cup \{\varnothing\})^{(m+2)\times 1}$, and ${\bf H} \in \complex^{(m+2)\times(m+2)}$. In phase $\ell$, if node $i$ is in transmission mode $X_{i}[k]$ follows the input distribution $X_{i}^{(\ell)} \sim {\cal CN}(0,P_{i})$. Otherwise, $X_{i}[k]= \varnothing$, which means that the input symbol does not exist in the above mathematical channel model.

$h_{i,j}$ is the effective
channel gain between transmitter $i$ and receiver $j$. We assume that each node is
fully aware of the channel gains, i.e., full CSI. The noise at all
receivers is independent, of unit power, additive, white Gaussian,
complex and circularly symmetric. For convenience of analysis, we
also define the function $C(x) \eqdef \log_2(1+x)$.

For the simplest case $m=2$, we get the following achievable rate regions for the Gaussian channel from Theorems \ref{theorem:FMABC} $\thicksim$ \ref{theorem:PTDBC-NR}. For the analysis of the cooperation scheme, we assume $\hat{Y}_{\none}$ are zero mean Gaussians and define $P_y \eqdef E[Y_\none^2]$ , $P_{\hat{y}} \eqdef E[\hat{Y}_\none^2]$ and $\sigma_y \eqdef E[\hat{Y}_\none Y_\none]$. Then the relation between the received $Y_\none$ and the compressed $\hat{Y}_\none$ are given by the following equivalent channel model:
\begin{align}
\hat{Y}_\none [k] = h_{\none \hat{\none}} Y_\none[k] + Z_{\hat{\none}}[k]
\end{align}
where $h_{\none \hat{\none}} = \frac{\sigma_y}{P_y}$ and $Z_{\hat{\none}} \sim {\cal CN}(0,P_{\hat{y}} - \frac{\sigma_y^2}{P_y})$. We note that in the following, $P_{\hat{y}}$ and $\sigma_y$ are unknown variables corresponding to the quantization which we numerically optimize. For convenience of analysis, we assume that $|h_{\nr,\none}|>|h_{\nr,\ntwo}|$, $|h_{\nzero,\none}|>|h_{\nzero,\ntwo}|$ and $|h_{\none,\nr}|>|h_{\none,\ntwo}|$. Since the given channel is degraded as $X_\nr \rightarrow Y_\none \rightarrow Y_\ntwo$, ${\cal I}_\none = \{\ntwo\}$ and ${\cal I}_\ntwo = \varnothing$.

\subsection{Achievable rate regions in the Gaussian channel with $m=2$}
We apply the bounds found in previous sections to the Gaussian channel with three terminal nodes $(\nzero,\none,\ntwo)$ and one relay node $(\nr)$. The detail evaluations of necessary mutual information terms are defined in Appendix \ref{app:evaluation}.
\begin{itemize}
\item Simplest Protocol
\begin{align}
&\text{From \eqref{eq:simple:1}}\left\{
  \begin{array}{l}
    R_{\nzero,\none} + R_{\nzero,\ntwo} < \Delta_1 C(|h_{\nzero,\nr}|^2 P_\nzero)\\
  \end{array}
\right.\\
&\text{From \eqref{eq:simple:2}}\left\{
  \begin{array}{l}
    R_{\none,\nzero} < \Delta_2 C(|h_{\none,\nr}|^2 P_\none)\\
    R_{\ntwo,\nzero} < \Delta_3 C(|h_{\ntwo,\nr}|^2 P_\ntwo)\\
  \end{array}
 \right.\\
&\text{From \eqref{eq:simple:3}}\left\{
  \begin{array}{l}
    R_{\none,\nzero} + R_{\ntwo,\nzero} < \Delta_4 C(|h_{\nr,\nzero}|^2 P_\nr)\\
  \end{array}
  \right.\\
&\text{From \eqref{eq:simple:4}}\left\{
  \begin{array}{l}
    R_{\nzero,\none} < \Delta_5 C(|h_{\nr,\none}|^2 P_\nr)\\
    R_{\nzero,\ntwo} < \Delta_6 C(|h_{\nr,\ntwo}|^2 P_\nr)\\
  \end{array}
\right.
\end{align}
To obtain the regions numerically, we optimize $\Delta_1$, $\Delta_2$, $\Delta_3$, $\Delta_4$, $\Delta_5$ and $\Delta_6$ for the given channel mutual informations to maximize the achievable rate regions.

\item FMABC Protocol
\begin{align}
&\text{From \eqref{eq:FMABC:1}}\left\{
  \begin{array}{l}
    R_{\nzero,\none} + R_{\nzero,\ntwo} < \Delta_1 C(|h_{\nzero,\nr}|^2 P_\nzero)\\
    R_{\none,\nzero} < \Delta_1 C(|h_{\none,\nr}|^2 P_\none)\\
    R_{\ntwo,\nzero} < \Delta_1 C(|h_{\ntwo,\nr}|^2 P_\ntwo)\\
    R_{\nzero,\none} + R_{\nzero,\ntwo} + R_{\none,\nzero} < \Delta_1 C(|h_{\nzero,\nr}|^2 P_\nzero + |h_{\none,\nr}|^2 P_\none)\\
    R_{\nzero,\none} + R_{\nzero,\ntwo} + R_{\ntwo,\nzero}< \Delta_1 C(|h_{\nzero,\nr}|^2 P_\nzero+ |h_{\ntwo,\nr}|^2 P_\ntwo)\\
    R_{\none,\nzero} + R_{\ntwo,\nzero} < \Delta_1 C(|h_{\none,\nr}|^2 P_\none + |h_{\ntwo,\nr}|^2 P_\ntwo)\\
    R_{\nzero,\none} + R_{\nzero,\ntwo}+ R_{\none,\nzero} + R_{\ntwo,\nzero} < \Delta_1 C(|h_{\nzero,\nr}|^2 P_\nzero + |h_{\none,\nr}|^2 P_\none + |h_{\ntwo,\nr}|^2 P_\ntwo)
  \end{array}
\right.\\
&\text{From \eqref{eq:FMABC:2}}\left\{
  \begin{array}{l}
    R_{\none,\nzero} + R_{\ntwo,\nzero} < \Delta_2 C_{B}(P_\nr,h_{\nr,\nzero},{\bar\lambda}_{\nr1}^\pb,{\bar \beta}_\nr^\pb)\\
    R_{\nzero,\none} < \Delta_2 C_{B}(P_\nr,h_{\nr,\nzero},{\bar\lambda}_{\nr2}^\pb,{\bar \beta}_\nr^\pb)\\
    R_{\nzero,\ntwo} < \Delta_2 C_{B}(P_\nr,h_{\nr,\nzero},{\bar\lambda}_{\nr3}^\pb,{\bar \beta}_\nr^\pb)\\
    R_{\none,\nzero} + R_{\ntwo,\nzero} + R_{\nzero,\none} <  \Delta_2 (C_{B}(P_\nr,h_{\nr,\nzero},{\bar\lambda}_{\nr1}^\pb,{\bar \beta}_\nr^\pb) + C_{B}(P_\nr,h_{\nr,\nzero},{\bar\lambda}_{\nr2}^\pb,{\bar \beta}_\nr^\pb)) - \\
    ~~~~~~~~~~~~~~~~~~~~~~~~~~\Delta_2 C_{BE} ({\bar\lambda}_{\nr1}^\pb,{\bar\lambda}_{\nr2}^\pb,{\bar \beta}_\nr^\pb)\\
    R_{\none,\nzero} + R_{\ntwo,\nzero} + R_{\nzero,\ntwo} <  \Delta_2 (C_{B}(P_\nr,h_{\nr,\nzero},{\bar\lambda}_{\nr1}^\pb,{\bar \beta}_\nr^\pb) + C_{B}(P_\nr,h_{\nr,\nzero},{\bar\lambda}_{\nr3}^\pb,{\bar \beta}_\nr^\pb)) - \\
    ~~~~~~~~~~~~~~~~~~~~~~~~~~\Delta_2 C_{BE} ({\bar\lambda}_{\nr1}^\pb,{\bar\lambda}_{\nr3}^\pb,{\bar \beta}_\nr^\pb)\\
    R_{\nzero,\none} + R_{\nzero,\ntwo} <  \Delta_2 (C_{B}(P_\nr,h_{\nr,\nzero},{\bar\lambda}_{\nr2}^\pb,{\bar \beta}_\nr^\pb) + C_{B}(P_\nr,h_{\nr,\nzero},{\bar\lambda}_{\nr3}^\pb,{\bar \beta}_\nr^\pb)) - \\
    ~~~~~~~~~~~~~~~~~~~~~~~~~~\Delta_2 C_{BE} ({\bar\lambda}_{\nr2}^\pb,{\bar\lambda}_{\nr3}^\pb,{\bar \beta}_\nr^\pb)\\
    R_{\none,\nzero} + R_{\ntwo,\nzero} + R_{\nzero,\none} + R_{\nzero,\ntwo} <  \Delta_2 (C_{B}(P_\nr,h_{\nr,\nzero},{\bar\lambda}_{\nr1}^\pb,{\bar \beta}_\nr^\pb) + C_{B}(P_\nr,h_{\nr,\nzero},{\bar\lambda}_{\nr2}^\pb,{\bar \beta}_\nr^\pb) + \\
     ~~~C_{B}(P_\nr,h_{\nr,\nzero},{\bar\lambda}_{\nr3}^\pb,{\bar \beta}_\nr^\pb) - C_{BE2} ({\bar\lambda}_{\nr3}^\pb,{\bar\lambda}_{\nr1}^\pb,{\bar\lambda}_{\nr2}^\pb,{\bar \beta}_\nr^\pb) - C_{BE} ({\bar\lambda}_{\nr1}^\pb,{\bar\lambda}_{\nr2}^\pb,{\bar \beta}_\nr^\pb))\\
  \end{array}
\right.
\end{align}
To obtain the regions numerically, we optimize $\Delta_1$, $\Delta_2$, ${\bf\Lambda}_\nr^\pb$ and ${\bar\beta}_{\nr}^\pb$ for the given channel mutual informations to maximize the achievable rate regions.

\item FMABC-N Protocol
\begin{align}
&\text{From \eqref{eq:FMABC:DF:GC:1}}\left\{
  \begin{array}{l}
    R_{\nzero,\none} + R_{\nzero,\ntwo} < \Delta_1 C(|h_{\nzero,\nr}|^2 P_\nzero)\\
    R_{\none,\nzero} < \Delta_1 C(|h_{\none,\nr}|^2 P_\none)\\
    R_{\ntwo,\nzero} < \Delta_1 C(|h_{\ntwo,\nr}|^2 P_\ntwo)\\
    R_{\nzero,\none} + R_{\nzero,\ntwo} + R_{\none,\nzero} < \Delta_1 C(|h_{\nzero,\nr}|^2 P_\nzero + |h_{\none,\nr}|^2 P_\none)\\
    R_{\nzero,\none} + R_{\nzero,\ntwo} + R_{\ntwo,\nzero}< \Delta_1 C(|h_{\nzero,\nr}|^2 P_\nzero+ |h_{\ntwo,\nr}|^2 P_\ntwo)\\
    R_{\none,\nzero} + R_{\ntwo,\nzero} < \Delta_1 C(|h_{\none,\nr}|^2 P_\none + |h_{\ntwo,\nr}|^2 P_\ntwo)\\
    R_{\nzero,\none} + R_{\nzero,\ntwo}+ R_{\none,\nzero} + R_{\ntwo,\nzero} < \Delta_1 C(|h_{\nzero,\nr}|^2 P_\nzero + |h_{\none,\nr}|^2 P_\none + |h_{\ntwo,\nr}|^2 P_\ntwo)
  \end{array}
\right.\\
&\text{From \eqref{eq:FMABC:DF:GC:6}}\left\{
  \begin{array}{l}
     R_{\nzero,\none} < \Delta_2 C_{B}(P_\nr,h_{\nr,\none},{\bar \lambda}_{\nr1}^\pb,{\bar\beta}_\nr^\pb)\\
    R_{\nzero,\ntwo} < \Delta_2 C_{B}(P_\nr,h_{\nr,\ntwo},{\bar \lambda}_{\nr2}^\pb,{\bar\beta}_\nr^\pb)\\
    R_{\nzero,\none} + R_{\nzero,\ntwo} < \Delta_2(C_{B}(P_\nr,h_{\nr,\none},{\bar \lambda}_{\nr1}^\pb,{\bar\beta}_\nr^\pb) + C_{B}(P_\nr,h_{\nr,\ntwo},{\bar\lambda}_{\nr 2}^\pb,{\bar\beta}_\nr^\pb)- \\
    ~~~~~~~~~~~~~~~~~~~~~~C_{BE}({\bar\lambda}_{\nr 1}^\pb,{\bar \lambda}_{\nr 2}^\pb,{\bar\beta}_\nr^\pb))
  \end{array}
\right.\\
&\text{From \eqref{eq:FMABC:DF:GC:3}}\left\{
  \begin{array}{l}
   R_{\none,\nzero} < \Delta_2  C_{C}(P_\nr,h_{\nr,\nzero},{\bar \lambda}_{\nr 1}^\pb,{\bar \lambda}_{\nr 2}^\pb,{\bar\beta}_\nr^\pb)\\
   R_{\ntwo,\nzero} < \Delta_2  C_{C}(P_\nr,h_{\nr,\nzero},{\bar \lambda}_{\nr 2}^\pb,{\bar \lambda}_{\nr 1}^\pb,{\bar\beta}_\nr^\pb)\\
   R_{\none,\nzero} + R_{\ntwo,\nzero} < \Delta_2 C(|h_{\nr,\nzero}|^2 P_\nr)
  \end{array}
\right.
\end{align}
To obtain the regions numerically, we optimize $\Delta_1$, $\Delta_2$, ${\bf\Lambda}_{\nr}^\pb$ and $\beta_{\nr}^\pb$ for the given channel mutual informations to maximize the achievable rate regions.

\item PMABC Protocol
\begin{align}
&\text{From \eqref{eq:PMABC:1}}\left\{
  \begin{array}{l}
    R_{\nzero,\none} < \Delta_1 C(|h_{\nzero,\nr}|^2 P_\nzero)\\
    R_{\nzero,\ntwo} < \Delta_2 C(|h_{\nzero,\nr}|^2 P_\nzero)\\
  \end{array}
\right.\\
&\text{From \eqref{eq:PMABC:2}}\left\{
  \begin{array}{l}
    R_{\none,\nzero} < \Delta_1 C(|h_{\none,\nr}|^2 P_\none)\\
    R_{\ntwo,\nzero} < \Delta_2 C(|h_{\ntwo,\nr}|^2 P_\ntwo)\\
  \end{array}
\right.\\
&\text{From \eqref{eq:PMABC:3}}\left\{
  \begin{array}{l}
    R_{\nzero,\none} + R_{\none,\nzero} < \Delta_1 C(|h_{\nzero,\nr}|^2 P_\nzero + |h_{\none,\nr}|^2 P_\none)\\
    R_{\nzero,\ntwo} + R_{\ntwo,\nzero} < \Delta_2 C(|h_{\nzero,\nr}|^2 P_\nzero + |h_{\ntwo,\nr}|^2 P_\ntwo)\\
  \end{array}
\right.\\
&\text{From \eqref{eq:PMABC:4}}\left\{
  \begin{array}{l}
    R_{\none,\nzero} + R_{\ntwo,\nzero} < \Delta_3 C_{B}(P_\nr,h_{\nr,\nzero},{\bar\lambda}_{\nr1}^\pc,{\bar \beta}_\nr^\pc)\\
    R_{\nzero,\none} < \Delta_3 C_{B}(P_\nr,h_{\nr,\nzero},{\bar\lambda}_{\nr2}^\pc,{\bar \beta}_\nr^\pc)\\
    R_{\nzero,\ntwo} < \Delta_3 C_{B}(P_\nr,h_{\nr,\nzero},{\bar\lambda}_{\nr3}^\pc,{\bar \beta}_\nr^\pc)\\
    R_{\none,\nzero} + R_{\ntwo,\nzero} + R_{\nzero,\none} <  \Delta_3 (C_{B}(P_\nr,h_{\nr,\nzero},{\bar\lambda}_{\nr1}^\pc,{\bar \beta}_\nr^\pc) + C_{B}(P_\nr,h_{\nr,\nzero},{\bar\lambda}_{\nr2}^\pc,{\bar \beta}_\nr^\pc)) - \\
    ~~~~~~~~~~~~~~~~~~~~~~~~~~\Delta_3 C_{BE} ({\bar\lambda}_{\nr1}^\pc,{\bar\lambda}_{\nr2}^\pc,{\bar \beta}_\nr^\pc)\\
    R_{\none,\nzero} + R_{\ntwo,\nzero} + R_{\nzero,\ntwo} <  \Delta_3 (C_{B}(P_\nr,h_{\nr,\nzero},{\bar\lambda}_{\nr1}^\pc,{\bar \beta}_\nr^\pc) + C_{B}(P_\nr,h_{\nr,\nzero},{\bar\lambda}_{\nr3}^\pc,{\bar \beta}_\nr^\pc)) - \\
    ~~~~~~~~~~~~~~~~~~~~~~~~~~\Delta_3 C_{BE} ({\bar\lambda}_{\nr1}^\pc,{\bar\lambda}_{\nr3}^\pc,{\bar \beta}_\nr^\pc)\\
    R_{\nzero,\none} + R_{\nzero,\ntwo} <  \Delta_3 (C_{B}(P_\nr,h_{\nr,\nzero},{\bar\lambda}_{\nr2}^\pc,{\bar \beta}_\nr^\pc) + C_{B}(P_\nr,h_{\nr,\nzero},{\bar\lambda}_{\nr3}^\pc,{\bar \beta}_\nr^\pc)) - \\
    ~~~~~~~~~~~~~~~~~~~~~~~~~~\Delta_3 C_{BE} ({\bar\lambda}_{\nr2}^\pc,{\bar\lambda}_{\nr3}^\pc,{\bar \beta}_\nr^\pc)\\
    R_{\none,\nzero} + R_{\ntwo,\nzero} + R_{\nzero,\none} + R_{\nzero,\ntwo} <  \Delta_3 (C_{B}(P_\nr,h_{\nr,\nzero},{\bar\lambda}_{\nr1}^\pc,{\bar \beta}_\nr^\pc) + C_{B}(P_\nr,h_{\nr,\nzero},{\bar\lambda}_{\nr2}^\pc,{\bar \beta}_\nr^\pc) + \\
     ~~~C_{B}(P_\nr,h_{\nr,\nzero},{\bar\lambda}_{\nr3}^\pc,{\bar \beta}_\nr^\pc) - C_{BE2} ({\bar\lambda}_{\nr3}^\pc,{\bar\lambda}_{\nr1}^\pc,{\bar\lambda}_{\nr2}^\pc,{\bar \beta}_\nr^\pc) - C_{BE} ({\bar\lambda}_{\nr1}^\pc,{\bar\lambda}_{\nr2}^\pc,{\bar \beta}_\nr^\pc))\\
  \end{array}
\right.
\end{align}
To obtain the regions numerically, we optimize $\Delta_1$, $\Delta_2$, $\Delta_3$, ${\bf\Lambda}_\nr^\pc$ and ${\bar\beta}_{\nr}^\pc$ for the given channel mutual informations to maximize the achievable rate regions.

\item PMABC-NR Protocol
\begin{align}
&\text{From \eqref{eq:PMABC:GC:1}}\left\{
  \begin{array}{l}
    R_{\none,\nzero} < \Delta_1  C(|h_{\none,\nr}|^2 P_\none)\\
    R_{\ntwo,\nzero} < \Delta_2  C(|h_{\ntwo,\nr}|^2 P_\ntwo)\\
    R_{\nzero,\none} < \Delta_1  C_C(P_\nzero,h_{\nzero,\nr},{\bar \lambda}_{\nzero 1}^\pa,{\bar \lambda}_{\nzero 2}^\pa,\beta_\nzero^\pa) + \Delta_2 C_C(P_\nzero,h_{\nzero,\nr},{\bar \lambda}_{\nzero 1}^\pb,{\bar \lambda}_{\nzero 2}^\pb,\beta_\nzero^\pb) \\
    R_{\nzero,\none} + R_{\none,\nzero} < \Delta_1 C_M(P_\none,P_\nzero,h_{\none,\nr},h_{\nzero,\nr},{\bar \lambda}_{\nzero 2}^\pa,{\bar \beta}_{\nzero}^\pa) + \\
    ~~~~~~~~~~~~~~~~~~ \Delta_1  C_C(P_\nzero,h_{\nzero,\nr},{\bar \lambda}_{\nzero 1}^\pa,{\bar \lambda}_{\nzero 2}^\pa,{\bar \beta}_\nzero^\pa) +  \Delta_2 C_C(P_\nzero,h_{\nzero,\nr},{\bar \lambda}_{\nzero 1}^\pb,{\bar \lambda}_{\nzero 2}^\pb,{\bar \beta}_\nzero^\pb)\\
    R_{\nzero,\none} + R_{\ntwo,\nzero} < \Delta_2 C_M(P_\ntwo,P_\nzero,h_{\ntwo,\nr},h_{\nzero,\nr},{\bar \lambda}_{\nzero 2}^\pb,{\bar \beta}_{\nzero}^\pb) + \\
    ~~~~~~~~~~~~~~~~~~ \Delta_1  C_C(P_\nzero,h_{\nzero,\nr},{\bar \lambda}_{\nzero 1}^\pa,{\bar \lambda}_{\nzero 2}^\pa,{\bar \beta}_\nzero^\pa) +  \Delta_2 C_C(P_\nzero,h_{\nzero,\nr},{\bar \lambda}_{\nzero 1}^\pb,{\bar \lambda}_{\nzero 2}^\pb,{\bar \beta}_\nzero^\pb)\\
    R_{\nzero,\none} + R_{\none,\nzero}+ R_{\ntwo,\nzero} <  \Delta_1 C_M(P_\none,P_\nzero,h_{\none,\nr},h_{\nzero,\nr},{\bar \lambda}_{\nzero 2 }^\pa,{\bf\beta}_{\nzero}^\pa) + \\
    ~~~~~~~~~~~~~~~~~~~~~~~~~~\Delta_2 C_M(P_\ntwo,P_\nzero,h_{\ntwo,\nr},h_{\nzero,\nr},{\bar \lambda}_{\nzero 2}^\pb,{\bar \beta}_{\nzero}^\pb) + \\
    ~~~~~~~~~~~~~~~~~~~~~~~~~~ \Delta_1  C_C(P_\nzero,h_{\nzero,\nr},{\bar \lambda}_{\nzero 1}^\pa,{\bar \lambda}_{\nzero 2}^\pa,{\bar \beta}_\nzero^\pa) + \\
    ~~~~~~~~~~~~~~~~~~~~~~~~~~ \Delta_2 C_C(P_\nzero,h_{\nzero,\nr},{\bar \lambda}_{\nzero 1}^\pb,{\bar \lambda}_{\nzero 2}^\pb,{\bar \beta}_\nzero^\pb)\\
    R_{\nzero,\ntwo} < \Delta_1  C_C(P_\nzero,h_{\nzero,\nr},{\bar \lambda}_{\nzero 2}^\pa,{\bar \lambda}_{\nzero 1}^\pa,{\bar \beta}_\nzero^\pa) + \Delta_2 C_C(P_\nzero,h_{\nzero,\nr},{\bar \lambda}_{\nzero 2}^\pb,{\bar \lambda}_{\nzero 1}^\pb,{\bar \beta}_\nzero^\pb) \\
    R_{\nzero,\ntwo} + R_{\none,\nzero} < \Delta_1 C_M(P_\none,P_\nzero,h_{\none,\nr},h_{\nzero,\nr},{\bar \lambda}_{\nzero 1}^\pa,{\bar\beta}_{\nzero}^\pa) + \\
    ~~~~~~~~~~~~~~~~~~ \Delta_1  C_C(P_\nzero,h_{\nzero,\nr},{\bar \lambda}_{\nzero 2}^\pa,{\bar \lambda}_{\nzero 1}^\pa,{\bar \beta}_\nzero^\pa) +  \Delta_2 C_C(P_\nzero,h_{\nzero,\nr},{\bar \lambda}_{\nzero 2}^\pb,{\bar \lambda}_{\nzero 1}^\pb,{\bar \beta}_\nzero^\pb)\\
  \end{array}
\right.\\
&\text{From \eqref{eq:PMABC:GC:1}}\left\{
  \begin{array}{l}
    R_{\nzero,\ntwo} + R_{\ntwo,\nzero} < \Delta_2 C_M(P_\ntwo,P_\nzero,h_{\ntwo,\nr},h_{\nzero,\nr},{\bar \lambda}_{\nzero 1}^\pb,{\bar \beta}_{\nzero}^\pb) + \\
    ~~~~~~~~~~~~~~~~~~ \Delta_1  C_C(P_\nzero,h_{\nzero,\nr},{\bar \lambda}_{\nzero 2}^\pa,{\bar \lambda}_{\nzero 1}^\pa,{\bar\beta}_\nzero^\pa) +  \Delta_2 C_C(P_\nzero,h_{\nzero,\nr},{\bar \lambda}_{\nzero 2}^\pb,{\bar \lambda}_{\nzero 1}^\pb,{\bar\beta}_\nzero^\pb)\\
    R_{\nzero,\ntwo} + R_{\none,\nzero}+ R_{\ntwo,\nzero} <  \Delta_1 C_M(P_\none,P_\nzero,h_{\none,\nr},h_{\nzero,\nr},{\bar \lambda}_{\nzero 1 }^\pa,{\bar \beta}_{\nzero}^\pa) + \\
    ~~~~~~~~~~~~~~~~~~~~~~~~~~\Delta_2 C_M(P_\ntwo,P_\nzero,h_{\ntwo,\nr},h_{\nzero,\nr},{\bar \lambda}_{\nzero 1}^\pb,{\bar \beta}_{\nzero}^\pb) + \\
    ~~~~~~~~~~~~~~~~~~~~~~~~~~ \Delta_1  C_C(P_\nzero,h_{\nzero,\nr},{\bar \lambda}_{\nzero 2}^\pa,{\bar \lambda}_{\nzero 1}^\pa,{\bar\beta}_\nzero^\pa) + \\
    ~~~~~~~~~~~~~~~~~~~~~~~~~~ \Delta_2 C_C(P_\nzero,h_{\nzero,\nr},{\bar \lambda}_{\nzero 2}^\pb,{\bar \lambda}_{\nzero 1}^\pb,{\bar \beta}_\nzero^\pb)\\
    R_{\nzero,\none} + R_{\nzero,\ntwo} < \Delta_1 C(|h_{\nzero,\nr}|^2 P_\nzero) + \Delta_2 C(|h_{\nzero,\nr}|^2 P_\nzero)\\
    R_{\nzero,\none} + R_{\nzero,\ntwo} + R_{\none,\nzero}< \Delta_1 C(|h_{\nzero,\nr}|^2 P_\nzero + |h_{\none,\nr}|^2 P_\none) + \Delta_2 C(|h_{\nzero,\nr}|^2 P_\nzero)\\
    R_{\nzero,\none} + R_{\nzero,\ntwo} + R_{\ntwo,\nzero}< \Delta_1 C(|h_{\nzero,\nr}|^2 P_\nzero) + \Delta_2 C(|h_{\nzero,\nr}|^2 P_\nzero + |h_{\ntwo,\nr}|^2 P_\ntwo)\\
    R_{\nzero,\none} + R_{\nzero,\ntwo} +  R_{\none,\nzero} + R_{\ntwo,\nzero}< \Delta_1 C(|h_{\nzero,\nr}|^2 P_\nzero + |h_{\none,\nr}|^2 P_\none) + \Delta_2 C(|h_{\nzero,\nr}|^2 P_\nzero + |h_{\ntwo,\nr}|^2 P_\ntwo)\\
  \end{array}
\right.\\
&\text{From \eqref{eq:PMABC:GC:4}}\left\{
  \begin{array}{l}
   R_{\nzero,\none} < \Delta_2 C_{BI}(P_\nzero,P_\ntwo,h_{\nzero,\none},h_{\ntwo,\none},{\bar \lambda}_{\nzero 1}^\pb,{\bar\beta}_\nzero^\pb) + \Delta_3 C_{B}(P_\nr,h_{\nr,\none},{\bar \lambda}_{\nr 1}^\pc,{\bar\beta}_\nr^\pc) \\
    R_{\nzero,\ntwo} < \Delta_1 C_{BI}(P_\nzero,P_\none,h_{\nzero,\ntwo},h_{\none,\ntwo},{\bar \lambda}_{\nzero 2}^\pa,{\bar\beta}_\nzero^\pa) + \Delta_3 C_{B}(P_\nr,h_{\nr,\ntwo},{\bar \lambda}_{\nr 2}^\pc,{\bar\beta}_\nr^\pc)\\
   R_{\nzero,\none} + R_{\nzero,\ntwo} < \Delta_1 C_{BI}(P_\nzero,P_\none,h_{\nzero,\ntwo},h_{\none,\ntwo},{\bar \lambda}_{\nzero 2 }^\pa,{\bar\beta}_\nzero^\pa) +\\
   ~~~~~~~~~~~~~~~~~~\Delta_2 C_{BI}(P_\nzero,P_\ntwo,h_{\nzero,\none},h_{\ntwo,\none},{\bar \lambda}_{\nzero 1}^\pb,{\bar \beta}_\nzero^\pb) +\\
   ~~~~~~~~~~~~~~~~~~\Delta_3(C_{B}(P_\nr,h_{\nr,\none},{\bar \lambda}_{\nr 1}^\pc,{\bar\beta}_\nr^\pc) + C_{B}(P_\nr,h_{\nr,\ntwo},{\bar \lambda}_{\nr 2 }^\pc,{\bar\beta}_\nr^\pc))- \\
   ~~~~~~~~~~~~~~~~~~\Delta_1 C_{BE}({\bar \lambda}_{\nzero 1}^\pa,{\bar \lambda}_{\nzero 2}^\pa,{\bar\beta}_\nzero^\pa) - \Delta_2 C_{BE}({\bar \lambda}_{\nzero 1}^\pb,{\bar \lambda}_{\nzero 2}^\pb,{\bar\beta}_\nzero^\pb) - \\
   ~~~~~~~~~~~~~~~~~~\Delta_3 C_{BE}({\bar \lambda}_{\nr 1}^\pc,{\bar \lambda}_{\nr 2}^\pc,{\bar\beta}_\nr^\pc)
  \end{array}
\right.\\
&\text{From \eqref{eq:PMABC:GC:5}}\left\{
  \begin{array}{l}
    R_{\none,\nzero} <  \Delta_3 C_{C}(P_\nr,h_{\nr,\nzero},{\bar \lambda}_{\nr 1}^\pc,{\bar \lambda}_{\nr 2}^\pc,{\bar\beta}_\nr^\pc)\\
    R_{\ntwo,\nzero} <  \Delta_3 C_{C}(P_\nr,h_{\nr,\nzero},{\bar \lambda}_{\nr 2}^\pc,{\bar \lambda}_{\nr 1}^\pc,{\bar\beta}_\nr^\pc)\\
    R_{\none,\nzero} + R_{\ntwo,\nzero} < \Delta_3 C\left(|h_{\nr,\nzero}|^2 P_\nr\right)\\
  \end{array}
\right.
\end{align}
 To obtain the regions numerically, we optimize $\Delta_1$, $\Delta_2$, $\Delta_3$, ${\bf\Lambda}_{\nzero}^\pa$, ${\bar\beta}_\nzero^\pa$, ${\bf\Lambda}_{\nzero}^\pb$, ${\bar\beta}_\nzero^\pb$, ${\bf\Lambda}_{\nr}^\pc$ and ${\bar\beta}_\nr^\pc$, for the given channel mutual informations to maximize the achievable rate regions.

\item PMABC-NRC Protocol
\begin{align}
&\text{From \eqref{eq:PMABC:DF:GC:1}}\left\{
  \begin{array}{l}
    R_{\nzero,\none} < \Delta_1 C_{CI} (P_\nzero,P_\none,h_{\nzero,\nr},h_{\none,\nr},{\bar \lambda}_{\nzero 1}^\pa,{\bar \lambda}_{\nzero 2 }^\pa,{\bar \lambda}_{\none 1}^\pa,{\bar\beta}_\nzero^\pa,{\bar\beta}_\none^\pa) + \\
    ~~~~~~~~~\Delta_2 C_{CI} (P_\nzero,P_\none,h_{\nzero,\nr},h_{\ntwo,\nr},{\bar \lambda}_{\nzero 1}^\pb,{\bar \lambda}_{\nzero 2}^\pb,1,{\bf\beta}_\nzero^\pb,1) \\
    R_{\nzero,\ntwo} < \Delta_1 C_{CI} (P_\nzero,P_\none,h_{\nzero,\nr},h_{\none,\nr},{\bar \lambda}_{\nzero 2}^\pa,{\bar \lambda}_{\nzero 1 }^\pa,{\bar \lambda}_{\none 1}^\pa,{\bar\beta}_\nzero^\pa,{\bar\beta}_\none^\pa) + \\
    ~~~~~~~~~\Delta_2 C_{CI} (P_\nzero,P_\none,h_{\nzero,\nr},h_{\ntwo,\nr},{\bar \lambda}_{\nzero 2}^\pb,{\bar \lambda}_{\nzero 1 }^\pb,1,{\bar \beta}_\nzero^\pb,1) \\
    R_{\nzero,\none} + R_{\nzero,\ntwo} < \Delta_1 C_{M}(P_\nzero,P_\none,h_{\nzero,\nr},h_{\none,\nr},{\bar \lambda}_{\none 1}^\pa,{\bar\beta}_{\none}^\pa) +\Delta_2 C(|h_{\nzero,\nr}|^2 P_\nzero)
  \end{array}
\right.\\
&\text{From \eqref{eq:PMABC:DF:GC:2}}\left\{
  \begin{array}{l}
    R_{\none,\nzero} < \Delta_1 C_{BI}(P_\none,P_\nzero,h_{\none,\nr},h_{\nzero,\nr},{\bar \lambda}_{\none 1}^\pa,{\bar\beta}_\none^\pa)\\
    R_{\ntwo,\nzero} < \Delta_2 C_{BI}(P_\ntwo,P_\nzero,h_{\ntwo,\nr},h_{\nzero,\nr},1,1)
  \end{array}
\right.\\
&\text{From \eqref{eq:PMABC:DF:GC:4}}\left\{
  \begin{array}{l}
  R_{\nzero,\none} < \Delta_2 C_{BI}(P_\nzero,P_\ntwo,h_{\nzero,\none},h_{\ntwo,\none},{\bar \lambda}_{\nzero 1}^\pb,{\bar\beta}_\nzero^\pb) + \Delta_3 C_{B}(P_\nr,h_{\nr,\none},{\bar \lambda}_{\nr 1}^\pc,{\bar\beta}_\nr^\pc)\\
    R_{\nzero,\ntwo} < \Delta_1 C_{BM} (P_\nzero,P_\none,h_{\nzero,\ntwo},h_{\none,\ntwo},{\bar \lambda}_{\nzero 2}^\pa,{\bar \lambda}_{\none 2 }^\pa,{\bar\beta}_\nzero^\pa,{\bar\beta}_\none^\pa) + \\
    ~~~~~~~~~\Delta_3 C_{BC}(P_\nr,h_{\nr,\none},h_{\nr,\ntwo},{\bar \lambda}_{\nr 2}^\pc,{\bar\beta}_{\nr}^\pc,P_{\hat \none},\sigma_\none)\\
   R_{\nzero,\none} + R_{\nzero,\ntwo} < \Delta_1 C_{BM} (P_\nzero,P_\none,h_{\nzero,\ntwo},h_{\none,\ntwo},{\bar \lambda}_{\nzero 2}^\pa,{\bar \lambda}_{\none 2}^\pa,{\bar\beta}_\nzero^\pa,{\bar\beta}_\none^\pa) +\\
   ~~~~~~~~~~~~~~~~~~\Delta_2 C_{BI}(P_\nzero,P_\ntwo,h_{\nzero,\none},h_{\ntwo,\none},{\bar \lambda}_{\nzero 1}^\pb,{\bar\beta}_\nzero^\pb) +\\
   ~~~~~~~~~~~~~~~~~~\Delta_3(C_{B}(P_\nr,h_{\nr,\none},{\bar \lambda}_{\nr 1}^\pc,{\bar\beta}_\nr^\pc) + C_{BC}(P_\nr,h_{\nr,\none},h_{\nr,\ntwo},{\bar \lambda}_{\nr 2}^\pc,{\bar\beta}_{\nr}^\pc,P_{\hat \none},\sigma_\none))- \\
   ~~~~~~~~~~~~~~~~~~\Delta_1 C_{BE}({\bar \lambda}_{\nzero 1}^\pa,{\bar \lambda}_{\nzero 2}^\pa,{\bar\beta}_\nzero^\pa) - \Delta_2 C_{BE}({\bar \lambda}_{\nzero 1}^\pb,{\bar \lambda}_{\nzero 2}^\pb,{\bar\beta}_\nzero^\pb) - \\
   ~~~~~~~~~~~~~~~~~~\Delta_3 C_{BE}({\bar \lambda}_{\nr 1}^\pc,{\bar \lambda}_{\nr 2}^\pc,{\bar\beta}_\nr^\pc)
  \end{array}
\right.\\
&\text{From \eqref{eq:PMABC:DF:GC:5}}\left\{
  \begin{array}{l}
    R_{\none,\nzero} <  \Delta_3 C_{C}(P_\nr,h_{\nr,\nzero},{\bar \lambda}_{\nr 1}^\pc,{\bar \lambda}_{\nr 2}^\pc,{\bar\beta}_\nr^\pc)\\
    R_{\ntwo,\nzero} <  \Delta_3 C_{C}(P_\nr,h_{\nr,\nzero},{\bar \lambda}_{\nr 2}^\pc,{\bar \lambda}_{\nr1}^\pc,{\bar\beta}_\nr^\pc)\\
    R_{\none,\nzero} + R_{\ntwo,\nzero} < \Delta_3 C\left(|h_{\nr,\nzero}|^2 P_\nr\right)\\
  \end{array}
\right.\\
&\text{From \eqref{eq:PMABC:DF:GC:6}}\left\{
  \begin{array}{l}
    R_{\none,\nzero} <  \Delta_1 C_{BI}(P_\none,P_\nzero,h_{\none,\nr},h_{\nzero,\nr},{\bar \lambda}_{\none 1}^\pa,{\bar\beta}_\none^\pa) + \\
    ~~~~~~~~~ \Delta_1 C_{BI}(P_\none,P_\nzero,h_{\none,\ntwo},h_{\nzero,\ntwo},{\bar \lambda}_{\none 2}^\pa,{\bar\beta}_\none^\pa) - \\
    ~~~~~~~~~ \Delta_1 C_{BE}({\bar \lambda}_{\none 1}^\pa,{\bar \lambda}_{\none 2}^\pa,{\bar\beta}_\none^\pa) - \Delta_3 C\left(\frac{(\sigma_\none)^2(1-P^{**})}{P_{\hat{\none}} (|h_{\nr,\none}|^2 +1) -(\sigma_\none)^2}\right)
  \end{array}
\right.\\
&\text{From \eqref{eq:PMABC:DF:GC:7}}\left\{
  \begin{array}{l}
   \Delta_3 C\left(\frac{(\sigma_\none)^2(1-P^{**})}{P_{\hat{\none}} (|h_{\nr,\none}|^2 +1) -(\sigma_\none)^2}\right) <  \Delta_1 C_{BI}(P_\none,P_\nzero,h_{\none,\ntwo},h_{\nzero,\ntwo},{\bar \lambda}_{\none 2}^\pa,{\bar\beta}_\none^\pa)\\
  \end{array}
\right.
\end{align}
where
\begin{align}
P^{**} & = \frac{|h_{\nr,\ntwo}|^2 P_\nr}{|h_{\nr,\ntwo}|^2 P_\nr
+1}\cdot \frac{|h_{\nr,\none}|^2 P_\nr}{|h_{\nr,\none}|^2 P_\nr +1}
\end{align}
 To obtain the regions numerically, we optimize $\Delta_1$, $\Delta_2$, $\Delta_3$, ${\bf\Lambda}_{\nzero}^\pa$, ${\bar\beta}_\nzero^\pa$, ${\bf \Lambda}_{\none}^\pa$, ${\bar\beta}_\none^\pa$, ${\bf\Lambda}_{\nzero}^\pb$, ${\bar\beta}_\nzero^\pb$, ${\bf\Lambda}_{\nr}^\pc$, ${\bar\beta}_\nr^\pc$, $P_{\hat \none}$ and $\sigma_{\none}$ for the given channel mutual informations to maximize the achievable rate regions.

\item FTDBC Protocol
\begin{align}
&\text{From \eqref{eq:FTDBC:1}}\left\{
  \begin{array}{l}
    R_{\nzero,\none} + R_{\nzero,\ntwo} < \Delta_1 C(|h_{\nzero,\nr}|^2 P_\nzero)\\
  \end{array}
\right.\\
&\text{From \eqref{eq:FTDBC:2}}\left\{
  \begin{array}{l}
    R_{\none,\nzero} < \Delta_2 C(|h_{\none,\nr}|^2 P_\none)\\
    R_{\ntwo,\nzero} < \Delta_3 C(|h_{\ntwo,\nr}|^2 P_\ntwo)\\
  \end{array}
\right.\\
&\text{From \eqref{eq:FTDBC:3}}\left\{
  \begin{array}{l}
    R_{\none,\nzero} + R_{\ntwo,\nzero} < \Delta_4 C_{B}(P_\nr,h_{\nr,\nzero},{\bar\lambda}_{\nr1}^\pd,{\bar \beta}_\nr^\pd)\\
    R_{\nzero,\none} < \Delta_4 C_{B}(P_\nr,h_{\nr,\nzero},{\bar\lambda}_{\nr2}^\pd,{\bar \beta}_\nr^\pd)\\
    R_{\nzero,\ntwo} < \Delta_4 C_{B}(P_\nr,h_{\nr,\nzero},{\bar\lambda}_{\nr3}^\pd,{\bar \beta}_\nr^\pd)\\
    R_{\none,\nzero} + R_{\ntwo,\nzero} + R_{\nzero,\none} <  \Delta_4 (C_{B}(P_\nr,h_{\nr,\nzero},{\bar\lambda}_{\nr1}^\pd,{\bar \beta}_\nr^\pd) + C_{B}(P_\nr,h_{\nr,\nzero},{\bar\lambda}_{\nr2}^\pd,{\bar \beta}_\nr^\pd)) - \\
    ~~~~~~~~~~~~~~~~~~~~~~~~~~\Delta_4 C_{BE} ({\bar\lambda}_{\nr1}^\pd,{\bar\lambda}_{\nr2}^\pd,{\bar \beta}_\nr^\pd)\\
    R_{\none,\nzero} + R_{\ntwo,\nzero} + R_{\nzero,\ntwo} <  \Delta_4 (C_{B}(P_\nr,h_{\nr,\nzero},{\bar\lambda}_{\nr1}^\pd,{\bar \beta}_\nr^\pd) + C_{B}(P_\nr,h_{\nr,\nzero},{\bar\lambda}_{\nr3}^\pd,{\bar \beta}_\nr^\pd)) - \\
    ~~~~~~~~~~~~~~~~~~~~~~~~~~\Delta_4 C_{BE} ({\bar\lambda}_{\nr1}^\pd,{\bar\lambda}_{\nr3}^\pd,{\bar \beta}_\nr^\pd)\\
    R_{\nzero,\none} + R_{\nzero,\ntwo} <  \Delta_4 (C_{B}(P_\nr,h_{\nr,\nzero},{\bar\lambda}_{\nr2}^\pd,{\bar \beta}_\nr^\pd) + C_{B}(P_\nr,h_{\nr,\nzero},{\bar\lambda}_{\nr3}^\pd,{\bar \beta}_\nr^\pd)) - \\
    ~~~~~~~~~~~~~~~~~~~~~~~~~~\Delta_4 C_{BE} ({\bar\lambda}_{\nr2}^\pd,{\bar\lambda}_{\nr3}^\pd,{\bar \beta}_\nr^\pd)\\
    R_{\none,\nzero} + R_{\ntwo,\nzero} + R_{\nzero,\none} + R_{\nzero,\ntwo} <  \Delta_4 (C_{B}(P_\nr,h_{\nr,\nzero},{\bar\lambda}_{\nr1}^\pd,{\bar \beta}_\nr^\pd) + C_{B}(P_\nr,h_{\nr,\nzero},{\bar\lambda}_{\nr2}^\pd,{\bar \beta}_\nr^\pd) + \\
     ~~~C_{B}(P_\nr,h_{\nr,\nzero},{\bar\lambda}_{\nr3}^\pd,{\bar \beta}_\nr^\pd) - C_{BE2} ({\bar\lambda}_{\nr3}^\pd,{\bar\lambda}_{\nr1}^\pd,{\bar\lambda}_{\nr2}^\pd,{\bar \beta}_\nr^\pd) - C_{BE} ({\bar\lambda}_{\nr1}^\pd,{\bar\lambda}_{\nr2}^\pd,{\bar \beta}_\nr^\pd))\\
  \end{array}
\right.
\end{align}
To obtain the regions numerically, we optimize $\Delta_1$, $\Delta_2$, $\Delta_3$, $\Delta_4$, ${\bf\Lambda}_\nr^\pd$ and ${\bar\beta}_{\nr}^\pd$ for the given channel mutual informations to maximize the achievable rate regions.

\item FTDBC-NR Protocol
\begin{align}
&\text{From \eqref{eq:FTDBC:GC:1}}\left\{
  \begin{array}{l}
    R_{\nzero,\none} < \Delta_1  C_{C}(P_\nzero,h_{\nzero,\nr},{\bar \lambda}_{\nzero 1}^\pa,{\bar \lambda}_{\nzero 2}^\pa,{\bar\beta}_\nzero^\pa)\\
    R_{\nzero,\ntwo} < \Delta_1  C_{C}(P_\nzero,h_{\nzero,\nr},{\bar \lambda}_{\nzero 2}^\pa,{\bar \lambda}_{\nzero 1}^\pa,{\bar\beta}_\nzero^\pa)\\
    R_{\nzero,\none} + R_{\nzero,\ntwo} < \Delta_1 C(|h_{\nzero,\nr}|^2 P_\nzero)\\
  \end{array}
\right.\\
&\text{From \eqref{eq:FTDBC:GC:2}}\left\{
  \begin{array}{l}
   R_{\none,\nzero} <  \Delta_2 C(|h_{\none,\nr}|^2 P_\none)\\
   R_{\ntwo,\nzero} <  \Delta_3 C(|h_{\ntwo,\nr}|^2 P_\ntwo)\\
  \end{array}
\right.\\
&\text{From \eqref{eq:FTDBC:GC:4}}\left\{
  \begin{array}{l}
   R_{\nzero,\none} < \Delta_1 C_{B}(P_\nzero,h_{\nzero,\none},{\bar \lambda}_{\nzero 1}^\pa,{\bar\beta}_\nzero^\pa) + \Delta_4 C_{B}(P_\nr,h_{\nr,\none},{\bar \lambda}_{\nr 1}^\pd,{\bar\beta}_\nr^\pd) \\
    R_{\nzero,\ntwo} < \Delta_1 C_{B}(P_\nzero,h_{\nzero,\ntwo},{\bar \lambda}_{\nzero 2}^\pa,{\bar\beta}_\nzero^\pa) + \Delta_4 C_{B}(P_\nr,h_{\nr,\ntwo},{\bar \lambda}_{\nr 2}^\pd,{\bar\beta}_\nr^\pd)\\
   R_{\nzero,\none} + R_{\nzero,\ntwo} < \Delta_1(C_{B}(P_\nzero,h_{\nzero,\none},{\bar \lambda}_{\nzero 1}^\pa,{\bar\beta}_\nzero^\pa) + C_{B}(P_\nzero,h_{\nzero,\ntwo},{\bar \lambda}_{\nzero 2}^\pa,{\bar\beta}_\nzero^\pa)) +\\
   ~~~~~~~~~~~~~~~~~\Delta_4(C_{B}(P_\nr,h_{\nr,\none},{\bar \lambda}_{\nr 1}^\pd,{\bar\beta}_\nr^\pd) + C_{B}(P_\nr,h_{\nr,\ntwo},{\bar \lambda}_{\nr 2 }^\pd,{\bar\beta}_\nr^\pd))- \\
   ~~~~~~~~~~~~~~~~~\Delta_1 C_{BE}({\bar \lambda}_{\nzero 1}^\pa,{\bar \lambda}_{\nzero 2}^\pa,{\bar\beta}_\nzero^\pa) - \Delta_4 C_{BE}({\bar \lambda}_{\nr 1}^\pd,{\bar \lambda}_{\nr 2}^\pd,{\bar\beta}_\nr^\pd)
  \end{array}
\right.\\
&\text{From \eqref{eq:FTDBC:GC:5}}\left\{
  \begin{array}{l}
    R_{\none,\nzero} <  \Delta_2 C(|h_{\none,\nzero}|^2 P_\none) + \Delta_4 C_{C}(P_\nr,h_{\nr,\nzero},{\bar \lambda}_{\nr 1}^\pd,{\bar \lambda}_{\nr 2 }^\pd,{\bar\beta}_\nr^\pd)\\
    R_{\ntwo,\nzero} <  \Delta_3 C(|h_{\ntwo,\nzero}|^2 P_\ntwo) + \Delta_4 C_{C}(P_\nr,h_{\nr,\nzero},{\bar \lambda}_{\nr 2}^\pd,{\bar \lambda}_{\nr 1}^\pd,{\bar\beta}_\nr^\pd)\\
    R_{\none,\nzero} + R_{\ntwo,\nzero} < \Delta_2 C(|h_{\none,\nzero}|^2 P_\none)  +  \Delta_3 C(|h_{\ntwo,\nzero}|^2 P_\ntwo) + \Delta_4 C\left(|h_{\nr,\nzero}|^2 P_\nr\right)\\
  \end{array}
\right.
\end{align}
 To obtain the regions numerically, we optimize $\Delta_1$, $\Delta_2$, $\Delta_3$, $\Delta_4$, ${\bf \Lambda}_{\nzero}^\pa$, ${\bar\beta}_\nzero^\pa$, ${\bf\Lambda}_{\nr}^\pd$ and ${\bar\beta}_\nr^\pd$  for the given channel mutual informations to maximize the achievable rate regions.

\item FTDBC-NRC Protocol
\begin{align}
&\text{From \eqref{eq:FTDBC:GC:CO:1}}\left\{
  \begin{array}{l}
    R_{\nzero,\none} < \Delta_1  C_{C}(P_\nzero,h_{\nzero,\nr},{\bar \lambda}_{\nzero 1}^\pa,{\bar \lambda}_{\nzero 2}^\pa,{\bar\beta}_\nzero^\pa)\\
    R_{\nzero,\ntwo} < \Delta_1  C_{C}(P_\nzero,h_{\nzero,\nr},{\bar \lambda}_{\nzero 2}^\pa,{\bar \lambda}_{\nzero 1}^\pa,{\bar\beta}_\nzero^\pa)\\
    R_{\nzero,\none} + R_{\nzero,\ntwo} < \Delta_1 C(|h_{\nzero,\nr}|^2 P_\nzero)\\
  \end{array}
\right.\\
&\text{From \eqref{eq:FTDBC:GC:CO:2}}\left\{
  \begin{array}{l}
   R_{\none,\nzero} <  \Delta_2 C_B(P_\none,h_{\none,\nr},{\bar \lambda}_{\none 1}^\pb,{\bar\beta}_\none^\pb)\\
   R_{\ntwo,\nzero} <  \Delta_3 C(|h_{\ntwo,\nr}|^2 P_\ntwo)\\
  \end{array}
\right.\\
&\text{From \eqref{eq:FTDBC:GC:CO:4}}\left\{
  \begin{array}{l}
   R_{\nzero,\none} < \Delta_1 C_{B}(P_\nzero,h_{\nzero,\none},{\bar \lambda}_{\nzero 1}^\pa,{\bar\beta}_\nzero^\pa) + \Delta_4 C_{B}(P_\nr,h_{\nr,\none},{\bar \lambda}_{\nr 1}^\pd,{\bar\beta}_\nr^\pd) \\
    R_{\nzero,\ntwo} < \Delta_1 C_{B}(P_\nzero,h_{\nzero,\ntwo},{\bar \lambda}_{\nzero 2}^\pa,{\bar\beta}_\nzero^\pa) + \Delta_4 C_{BC}(P_\nr,h_{\nr,\none},h_{\nr,\ntwo},{\bar \lambda}_{\nr 2}^\pd,{\bar\beta}_{\nr}^\pd,P_{\hat \none},\sigma_\none)\\
   R_{\nzero,\none} + R_{\nzero,\ntwo} < \Delta_1(C_{B}(P_\nzero,h_{\nzero,\none},{\bar \lambda}_{\nzero 1}^\pa,{\bar\beta}_\nzero^\pa) + C_{B}(P_\nzero,h_{\nzero,\ntwo},{\bar \lambda}_{\nzero 2}^\pa,{\bar\beta}_\nzero^\pa)) +\\
   ~~~~~~~~~~~~~~~~~\Delta_4(C_{B}(P_\nr,h_{\nr,\none},{\bar \lambda}_{\nr 1}^\pd,{\bar\beta}_\nr^\pd) + C_{BC}(P_\nr,h_{\nr,\none},h_{\nr,\ntwo},{\bar \lambda}_{\nr 2}^\pd,{\bar\beta}_{\nr}^\pd,P_{\hat \none},\sigma_\none))- \\
   ~~~~~~~~~~~~~~~~~\Delta_1 C_{BE}({\bar \lambda}_{\nzero 1}^\pa,{\bar \lambda}_{\nzero 2}^\pa,{\bar\beta}_\nzero^\pa) - \Delta_4 C_{BE}({\bar \lambda}_{\nr 1}^\pd,{\bar \lambda}_{\nr 2}^\pd,{\bar\beta}_\nr^\pd)
  \end{array}
\right.\\
&\text{From \eqref{eq:FTDBC:GC:CO:5}}\left\{
  \begin{array}{l}
    R_{\none,\nzero} <  \Delta_2 C_B(P_\none,h_{\none,\nzero},{\bar \lambda}_{\none 1}^\pb,{\bar\beta}_\none^\pb) + \Delta_4 C_{C}(P_\nr,h_{\nr,\nzero},{\bar \lambda}_{\nr 1}^\pd,{\bar \lambda}_{\nr 2}^\pd,{\bar\beta}_\nr^\pd)\\
    R_{\ntwo,\nzero} <  \Delta_3 C(|h_{\ntwo,\nzero}|^2 P_\ntwo) + \Delta_4 C_{C}(P_\nr,h_{\nr,\nzero},{\bar \lambda}_{\nr 2}^\pd,{\bar \lambda}_{\nr 1 }^\pd,{\bar\beta}_\nr^\pd)\\
    R_{\none,\nzero} + R_{\ntwo,\nzero} < \Delta_2 C_B(P_\none,h_{\none,\nzero},{\bar \lambda}_{\none 1}^\pb,{\bar\beta}_\none^\pb)  +  \Delta_3 C(|h_{\ntwo,\nzero}|^2 P_\ntwo) + \Delta_4 C\left(|h_{\nr,\nzero}|^2 P_\nr\right)\\
  \end{array}
\right.\\
&\text{From \eqref{eq:FTDBC:GC:CO:6}}\left\{
  \begin{array}{l}
    R_{\none,\nzero} <  \Delta_2 C_B(P_\none,h_{\none,\nr},{\bar \lambda}_{\none 1}^\pb,{\bar\beta}_\none^\pb) + \Delta_2 C_B(P_\none,h_{\none,\ntwo},{\bar \lambda}_{\none 2}^\pb,{\bar\beta}_\none^\pb) - \\
    ~~~~~~~~~~\Delta_2 C_{BE}({\bar \lambda}_{\none 1}^\pb,{\bar \lambda}_{\none 2}^\pb,{\bar\beta}_\none^\pb)
    - \Delta_4 C\left(\frac{(\sigma_\none)^2(1-P^{**})}{P_{\hat{\none}} (|h_{\nr,\none}|^2 +1) -(\sigma_\none)^2}\right)\\
  \end{array}
\right.\\
&\text{From \eqref{eq:FTDBC:GC:CO:7}}\left\{
  \begin{array}{l}
   \Delta_4 C\left(\frac{(\sigma_\none)^2(1-P^{**})}{P_{\hat{\none}} (|h_{\nr,\none}|^2 +1) -(\sigma_\none)^2}\right) <  \Delta_2 C_B(P_\none,h_{\none,\ntwo},{\bar \lambda}_{\none 2}^\pb,{\bar\beta}_\none^\pb)\\
  \end{array}
\right.
\end{align}
where
\begin{align}
P^{**} & = \frac{|h_{\nr,\ntwo}|^2 P_\nr}{|h_{\nr,\ntwo}|^2 P_\nr
+1}\cdot \frac{|h_{\nr,\none}|^2 P_\nr}{|h_{\nr,\none}|^2 P_\nr +1}
\end{align}
To obtain the regions numerically, we optimize $\Delta_1$,
$\Delta_2$, $\Delta_3$, $\Delta_4$, ${\bf\Lambda}_{\nzero}^\pa$,
${\bar\beta}_\nzero^\pa$,
${\bf\Lambda}_{\none}^\pb$,
${\bar\beta}_\none^\pb$, ${\bf\Lambda}_{\nr}^\pd$,
${\bar\beta}_\nr^\pd$, $P_{\hat \none}$ and $\sigma_\none$ for the
given channel mutual informations to maximize the achievable rate
regions.

\item PTDBC Protocol
\begin{align}
&\text{From \eqref{eq:PTDBC:1}}\left\{
  \begin{array}{l}
    R_{\nzero,\none} + R_{\nzero,\ntwo} < \Delta_1 C(|h_{\nzero,\nr}|^2 P_\nzero)\\
  \end{array}
\right.\\
&\text{From \eqref{eq:PTDBC:2}}\left\{
  \begin{array}{l}
    R_{\none,\nzero} < \Delta_2 C(|h_{\none,\nr}|^2 P_\none)\\
    R_{\ntwo,\nzero} < \Delta_2 C(|h_{\ntwo,\nr}|^2 P_\ntwo)\\
    R_{\none,\nzero} + R_{\ntwo,\nzero} < \Delta_2 C(|h_{\none,\nr}|^2 P_\none+|h_{\ntwo,\nr}|^2 P_\ntwo)\\
  \end{array}
\right.\\
&\text{From \eqref{eq:PTDBC:3}}\left\{
  \begin{array}{l}
    R_{\none,\nzero} + R_{\ntwo,\nzero} < \Delta_3 C_{B}(P_\nr,h_{\nr,\nzero},{\bar\lambda}_{\nr1}^\pc,{\bar \beta}_\nr^\pc)\\
    R_{\nzero,\none} < \Delta_3 C_{B}(P_\nr,h_{\nr,\nzero},{\bar\lambda}_{\nr2}^\pc,{\bar \beta}_\nr^\pc)\\
    R_{\nzero,\ntwo} < \Delta_3 C_{B}(P_\nr,h_{\nr,\nzero},{\bar\lambda}_{\nr3}^\pc,{\bar \beta}_\nr^\pc)\\
    R_{\none,\nzero} + R_{\ntwo,\nzero} + R_{\nzero,\none} <  \Delta_3 (C_{B}(P_\nr,h_{\nr,\nzero},{\bar\lambda}_{\nr1}^\pc,{\bar \beta}_\nr^\pc) + C_{B}(P_\nr,h_{\nr,\nzero},{\bar\lambda}_{\nr2}^\pc,{\bar \beta}_\nr^\pc)) - \\
    ~~~~~~~~~~~~~~~~~~~~~~~~~~\Delta_3 C_{BE} ({\bar\lambda}_{\nr1}^\pc,{\bar\lambda}_{\nr2}^\pc,{\bar \beta}_\nr^\pc)\\
    R_{\none,\nzero} + R_{\ntwo,\nzero} + R_{\nzero,\ntwo} <  \Delta_3 (C_{B}(P_\nr,h_{\nr,\nzero},{\bar\lambda}_{\nr1}^\pc,{\bar \beta}_\nr^\pc) + C_{B}(P_\nr,h_{\nr,\nzero},{\bar\lambda}_{\nr3}^\pc,{\bar \beta}_\nr^\pc)) - \\
    ~~~~~~~~~~~~~~~~~~~~~~~~~~\Delta_3 C_{BE} ({\bar\lambda}_{\nr1}^\pc,{\bar\lambda}_{\nr3}^\pc,{\bar \beta}_\nr^\pc)\\
    R_{\nzero,\none} + R_{\nzero,\ntwo} <  \Delta_3 (C_{B}(P_\nr,h_{\nr,\nzero},{\bar\lambda}_{\nr2}^\pc,{\bar \beta}_\nr^\pc) + C_{B}(P_\nr,h_{\nr,\nzero},{\bar\lambda}_{\nr3}^\pc,{\bar \beta}_\nr^\pc)) - \\
    ~~~~~~~~~~~~~~~~~~~~~~~~~~\Delta_3 C_{BE} ({\bar\lambda}_{\nr2}^\pc,{\bar\lambda}_{\nr3}^\pc,{\bar \beta}_\nr^\pc)\\
    R_{\none,\nzero} + R_{\ntwo,\nzero} + R_{\nzero,\none} + R_{\nzero,\ntwo} <  \Delta_3 (C_{B}(P_\nr,h_{\nr,\nzero},{\bar\lambda}_{\nr1}^\pc,{\bar \beta}_\nr^\pc) + C_{B}(P_\nr,h_{\nr,\nzero},{\bar\lambda}_{\nr2}^\pc,{\bar \beta}_\nr^\pc) + \\
     ~~~C_{B}(P_\nr,h_{\nr,\nzero},{\bar\lambda}_{\nr3}^\pc,{\bar \beta}_\nr^\pc) - C_{BE2} ({\bar\lambda}_{\nr3}^\pc,{\bar\lambda}_{\nr1}^\pc,{\bar\lambda}_{\nr2}^\pc,{\bar \beta}_\nr^\pc) - C_{BE} ({\bar\lambda}_{\nr1}^\pc,{\bar\lambda}_{\nr2}^\pc,{\bar \beta}_\nr^\pc))\\
  \end{array}
\right.
\end{align}
To obtain the regions numerically, we optimize $\Delta_1$, $\Delta_2$, $\Delta_3$, ${\bf\Lambda}_\nr^\pc$ and ${\bar\beta}_{\nr}^\pc$ for the given channel mutual informations to maximize the achievable rate regions.

\item PTDBC-NR Protocol
\begin{align}
&\text{From \eqref{eq:PTDBC:GC:1}}\left\{
  \begin{array}{l}
    R_{\nzero,\none} < \Delta_1  C_{C}(P_\nzero,h_{\nzero,\nr},{\bar \lambda}_{\nzero 1}^\pa,{\bar \lambda}_{\nzero 2}^\pa,{\bar\beta}_\nzero^\pa)\\
    R_{\nzero,\ntwo} < \Delta_1  C_{C}(P_\nzero,h_{\nzero,\nr},{\bar \lambda}_{\nzero 2}^\pa,{\bar \lambda}_{\nzero 1}^\pa,{\bar\beta}_\nzero^\pa)\\
    R_{\nzero,\none} + R_{\nzero,\ntwo} < \Delta_1 C(|h_{\nzero,\nr}|^2 P_\nzero)\\
  \end{array}
\right.\\
&\text{From \eqref{eq:PTDBC:GC:2}}\left\{
  \begin{array}{l}
   R_{\none,\nzero} <  \Delta_2 C(|h_{\none,\nr}|^2 P_\none)\\
   R_{\ntwo,\nzero} < \Delta_2 C(|h_{\ntwo,\nr}|^2 P_\ntwo)\\
   R_{\none,\nzero} + R_{\ntwo,\nzero} < \Delta_2 C(|h_{\none,\nr}|^2 P_\none+|h_{\ntwo,\nr}|^2 P_\ntwo)\\
  \end{array}
\right.\\
&\text{From \eqref{eq:PTDBC:GC:4}}\left\{
  \begin{array}{l}
   R_{\nzero,\none} < \Delta_1 C_{B}(P_\nzero,h_{\nzero,\none},{\bar \lambda}_{\nzero 1}^\pa,{\bar\beta}_\nzero^\pa) + \Delta_3 C_{B}(P_\nr,h_{\nr,\none},{\bar \lambda}_{\nr 1}^\pc,{\bar\beta}_\nr^\pc) \\
    R_{\nzero,\ntwo} < \Delta_1 C_{B}(P_\nzero,h_{\nzero,\ntwo},{\bar \lambda}_{\nzero 2}^\pa,{\bar\beta}_\nzero^\pa) + \Delta_3 C_{B}(P_\nr,h_{\nr,\ntwo},{\bar \lambda}_{\nr 2}^\pc,{\bar\beta}_\nr^\pc)\\
   R_{\nzero,\none} + R_{\nzero,\ntwo} < \Delta_1(C_{B}(P_\nzero,h_{\nzero,\none},{\bar \lambda}_{\nzero 1}^\pa,{\bar\beta}_\nzero^\pa) + C_{B}(P_\nzero,h_{\nzero,\ntwo},{\bar \lambda}_{\nzero 2}^\pa,{\bar\beta}_\nzero^\pa)) +\\
   ~~~~~~~~~~~~~~~~~\Delta_3(C_{B}(P_\nr,h_{\nr,\none},{\bar \lambda}_{\nr 1}^\pc,{\bar\beta}_\nr^\pc) + C_{B}(P_\nr,h_{\nr,\ntwo},{\bar \lambda}_{\nr 2 }^\pc,{\bar\beta}_\nr^\pc))- \\
   ~~~~~~~~~~~~~~~~~\Delta_1 C_{BE}({\bar \lambda}_{\nzero 1}^\pa,{\bar \lambda}_{\nzero 2}^\pa,{\bar\beta}_\nzero^\pa) - \Delta_3 C_{BE}({\bar \lambda}_{\nr 1}^\pc,{\bar \lambda}_{\nr 2}^\pc,{\bar\beta}_\nr^\pc)
  \end{array}
\right.\\
&\text{From \eqref{eq:PTDBC:GC:5}}\left\{
  \begin{array}{l}
    R_{\none,\nzero} <  \Delta_2 C(|h_{\none,\nzero}|^2 P_\none) + \Delta_3 C_{C}(P_\nr,h_{\nr,\nzero},{\bar \lambda}_{\nr 1}^\pc,{\bar \lambda}_{\nr 2 }^\pc,{\bar\beta}_\nr^\pc)\\
    R_{\ntwo,\nzero} <  \Delta_2 C(|h_{\ntwo,\nzero}|^2 P_\ntwo) + \Delta_3 C_{C}(P_\nr,h_{\nr,\nzero},{\bar \lambda}_{\nr 2}^\pc,{\bar \lambda}_{\nr 1}^\pc,{\bar\beta}_\nr^\pc)\\
    R_{\none,\nzero} + R_{\ntwo,\nzero} < \Delta_2 C(|h_{\none,\nzero}|^2 P_\none +|h_{\ntwo,\nzero}|^2 P_\ntwo) + \Delta_3 C\left(|h_{\nr,\nzero}|^2 P_\nr\right)\\
  \end{array}
\right.
\end{align}
 To obtain the regions numerically, we optimize $\Delta_1$, $\Delta_2$, $\Delta_3$, ${\bf \Lambda}_{\nzero}^\pa$, ${\bar\beta}_\nzero^\pa$, ${\bf\Lambda}_{\nr}^\pc$ and ${\bar\beta}_\nr^\pc$  for the given channel mutual informations to maximize the achievable rate regions.

\end{itemize}

\subsection{Outer bounds in the Gaussian channel with $m=2$}

\begin{itemize}
\item FMABC Protocol
\begin{align}
&\text{From \eqref{eq:FMABC:out:1}}\left\{
  \begin{array}{l}
    R_{\nzero,\none} + R_{\nzero,\ntwo} < \Delta_1 C(|h_{\nzero,\nr}|^2 P_\nzero)\\
  \end{array}
\right.\\
&\text{From \eqref{eq:FMABC:out:2}}\left\{
  \begin{array}{l}
    R_{\none,\nzero} + R_{\ntwo,\nzero} < \Delta_2 C(|h_{\nr,\nzero}|^2 P_\nr)
  \end{array}
\right.\\
&\text{From \eqref{eq:FMABC:out:3}}\left\{
  \begin{array}{l}
    R_{\none,\nzero} < \Delta_1 C(|h_{\none,\nr}|^2 P_\none)\\
    R_{\ntwo,\nzero} < \Delta_1 C(|h_{\ntwo,\nr}|^2 P_\ntwo)\\
    R_{\none,\nzero} + R_{\ntwo,\nzero} < \Delta_1 C(|h_{\none,\nr}|^2 P_\none + |h_{\ntwo,\nr}|^2 P_\ntwo)\\
  \end{array}
\right.\\
&\text{From \eqref{eq:FMABC:out:4}}\left\{
  \begin{array}{l}
   R_{\nzero,\none} < \Delta_2  C(|h_{\nr,\none}|^2 P_\nr)\\
   R_{\nzero,\ntwo} < \Delta_2  C(|h_{\nr,\ntwo}|^2 P_\nr)\\
   R_{\nzero,\none} + R_{\nzero,\ntwo} < \Delta_2  C(|h_{\nr,\none}|^2 P_\nr + |h_{\nr,\ntwo}|^2 P_\nr)\\
  \end{array}
\right.
\end{align}
To obtain the regions numerically, we optimize $\Delta_1$ and
$\Delta_2$ for the given channel mutual informations to maximize the
achievable rate regions.

\item PMABC Protocol
\begin{align}
&\text{From \eqref{eq:PMABC:out:1}}\left\{
  \begin{array}{l}
    R_{\nzero,\none} + R_{\nzero,\ntwo} < (\Delta_1 + \Delta_2) C(|h_{\nzero,\nr}|^2 P_\nzero)\\
  \end{array}
\right.\\
&\text{From \eqref{eq:PMABC:out:2}}\left\{
  \begin{array}{l}
    R_{\none,\nzero} + R_{\ntwo,\nzero} < \Delta_3 C(|h_{\nr,\nzero}|^2 P_\nr)
  \end{array}
\right.\\
&\text{From \eqref{eq:PMABC:out:3}}\left\{
  \begin{array}{l}
    R_{\none,\nzero} < \Delta_1 C(|h_{\none,\nr}|^2 P_\none + |h_{\none,\ntwo}|^2 P_\none)\\
    R_{\ntwo,\nzero} < \Delta_2 C(|h_{\ntwo,\nr}|^2 P_\ntwo + |h_{\ntwo,\none}|^2 P_\ntwo)\\
    R_{\none,\nzero} + R_{\ntwo,\nzero} < \Delta_1 C(|h_{\none,\nr}|^2 P_\none) + \Delta_2 C(|h_{\ntwo,\nr}|^2 P_\ntwo)\\
  \end{array}
\right.\\
&\text{From \eqref{eq:PMABC:out:4}}\left\{
  \begin{array}{l}
   R_{\nzero,\none} < \Delta_2 C(|h_{\nzero,\none}|^2 P_\nzero + |h_{\ntwo,\none}|^2 P_\ntwo) + \Delta_3  C(|h_{\nr,\none}|^2 P_\nr)\\
   R_{\nzero,\ntwo} < \Delta_1 C(|h_{\nzero,\ntwo}|^2 P_\nzero + |h_{\none,\ntwo}|^2 P_\none) + \Delta_3  C(|h_{\nr,\ntwo}|^2 P_\nr)\\
   R_{\nzero,\none} + R_{\nzero,\ntwo} < \Delta_1 C(|h_{\nzero,\ntwo}|^2 P_\nzero) + \Delta_2 C(|h_{\nzero,\none}|^2 P_\nzero) + \Delta_3 C(|h_{\nr,\none}|^2 P_\nr + |h_{\nr,\ntwo}|^2 P_\nr)\\
  \end{array}
\right.
\end{align}
To obtain the regions numerically, we optimize $\Delta_1$,
$\Delta_2$ and $\Delta_3$ for the given channel mutual informations
to maximize the achievable rate regions.

\item FTDBC Protocol
\begin{align}
&\text{From \eqref{eq:FTDBC:out:1}}\left\{
  \begin{array}{l}
    R_{\nzero,\none} < \Delta_1 C(|h_{\nzero,\nr}|^2 P_\nzero + |h_{\nzero,\none}|^2 P_\nzero) + \Delta_3 C(|h_{\ntwo,\nr}|^2 P_\ntwo + |h_{\ntwo,\none}|^2 P_\ntwo)\\
    R_{\nzero,\none} < \Delta_1 C(|h_{\nzero,\none}|^2 P_\nzero) + \Delta_3 C(|h_{\ntwo,\none}|^2 P_\ntwo) + \Delta_4 C(|h_{\nr,\none}|^2 P_\nr)\\
    R_{\nzero,\ntwo} < \Delta_1 C(|h_{\nzero,\nr}|^2 P_\nzero + |h_{\nzero,\ntwo}|^2 P_\nzero) + \Delta_2 C(|h_{\none,\nr}|^2 P_\none + |h_{\none,\ntwo}|^2 P_\none)\\
    R_{\nzero,\ntwo} < \Delta_1 C(|h_{\nzero,\ntwo}|^2 P_\nzero) + \Delta_2 C(|h_{\none,\ntwo}|^2 P_\none) + \Delta_4 C(|h_{\nr,\ntwo}|^2 P_\nr)\\
    R_{\nzero,\none} + R_{\nzero,\ntwo} < \Delta_1 C(|h_{\nzero,\nr}|^2 P_\nzero + |h_{\nzero,\none}|^2 P_\nzero + |h_{\nzero,\ntwo}|^2 P_\nzero)\\
    R_{\nzero,\none} + R_{\nzero,\ntwo} < \Delta_1 C(|h_{\nzero,\none}|^2 P_\nzero + |h_{\nzero,\ntwo}|^2 P_\nzero) + \Delta_4 C(|h_{\nr,\none}|^2 P_\nr + |h_{\nr,\ntwo}|^2 P_\nr)\\
  \end{array}
\right.\\
&\text{From \eqref{eq:FTDBC:out:2}}\left\{
  \begin{array}{l}
    R_{\none,\nzero} < \Delta_2 C(|h_{\none,\nzero}|^2 P_\none + |h_{\none,\ntwo}|^2 P_\none) + \Delta_4 C(|h_{\nr,\nzero}|^2 P_\nr + |h_{\nr,\ntwo}|^2 P_\nr)\\
    R_{\none,\nzero} < \Delta_2 C(|h_{\none,\nzero}|^2 P_\none + |h_{\none,\ntwo}|^2 P_\none + |h_{\none,\nr}|^2 P_\none)\\
    R_{\ntwo,\nzero} < \Delta_3 C(|h_{\ntwo,\nzero}|^2 P_\ntwo + |h_{\ntwo,\none}|^2 P_\ntwo) + \Delta_4 C(|h_{\nr,\nzero}|^2 P_\nr + |h_{\nr,\none}|^2 P_\nr)\\
    R_{\ntwo,\nzero} < \Delta_3 C(|h_{\ntwo,\nzero}|^2 P_\ntwo + |h_{\ntwo,\none}|^2 P_\ntwo + |h_{\ntwo,\nr}|^2 P_\ntwo)\\
    R_{\none,\nzero} + R_{\ntwo,\nzero} < \Delta_2 C(|h_{\none,\nzero}|^2 P_\none) + \Delta_3 C(|h_{\ntwo,\nzero}|^2 P_\ntwo) + \Delta_4 C(|h_{\nr,\nzero}|^2 P_\nr)\\
    R_{\none,\nzero} + R_{\ntwo,\nzero} < \Delta_2 C(|h_{\none,\nzero}|^2 P_\none + |h_{\none,\nr}|^2 P_\none) + \Delta_3 C(|h_{\ntwo,\nzero}|^2 P_\ntwo + |h_{\ntwo,\nr}|^2 P_\ntwo)\\
  \end{array}
\right.
\end{align}
To obtain the regions numerically, we optimize $\Delta_1$,
$\Delta_2$, $\Delta_3$ and $\Delta_4$ for the given channel mutual
informations to maximize the achievable rate regions.

\item PTDBC Protocol
\begin{align}
&\text{From \eqref{eq:PTDBC:out:1}}\left\{
  \begin{array}{l}
    R_{\nzero,\none} < \Delta_1 C(|h_{\nzero,\nr}|^2 P_\nzero + |h_{\nzero,\none}|^2 P_\nzero) + \Delta_2 C(|h_{\ntwo,\nr}|^2 P_\ntwo)\\
    R_{\nzero,\none} < \Delta_1 C(|h_{\nzero,\none}|^2 P_\nzero) + \Delta_3 C(|h_{\nr,\none}|^2 P_\nr)\\
    R_{\nzero,\ntwo} < \Delta_1 C(|h_{\nzero,\nr}|^2 P_\nzero + |h_{\nzero,\ntwo}|^2 P_\nzero) + \Delta_2 C(|h_{\none,\nr}|^2 P_\none)\\
    R_{\nzero,\ntwo} < \Delta_1 C(|h_{\nzero,\ntwo}|^2 P_\nzero) + \Delta_3 C(|h_{\nr,\ntwo}|^2 P_\nr)\\
    R_{\nzero,\none} + R_{\nzero,\ntwo} < \Delta_1 C(|h_{\nzero,\nr}|^2 P_\nzero + |h_{\nzero,\none}|^2 P_\nzero + |h_{\nzero,\ntwo}|^2 P_\nzero)\\
    R_{\nzero,\none} + R_{\nzero,\ntwo} < \Delta_1 C(|h_{\nzero,\none}|^2 P_\nzero + |h_{\nzero,\ntwo}|^2 P_\nzero) + \Delta_3 C(|h_{\nr,\none}|^2 P_\nr + |h_{\nr,\ntwo}|^2 P_\nr)\\
  \end{array}
\right.\\
&\text{From \eqref{eq:PTDBC:out:2}}\left\{
  \begin{array}{l}
    R_{\none,\nzero} < \Delta_2 C(|h_{\none,\nzero}|^2 P_\none) + \Delta_3 C(|h_{\nr,\nzero}|^2 P_\nr + |h_{\nr,\ntwo}|^2 P_\nr)\\
    R_{\none,\nzero} < \Delta_2 C(|h_{\none,\nzero}|^2 P_\none + |h_{\none,\nr}|^2 P_\none)\\
    R_{\ntwo,\nzero} < \Delta_2 C(|h_{\ntwo,\nzero}|^2 P_\ntwo) + \Delta_3 C(|h_{\nr,\nzero}|^2 P_\nr + |h_{\nr,\none}|^2 P_\nr)\\
    R_{\ntwo,\nzero} < \Delta_2 C(|h_{\ntwo,\nzero}|^2 P_\ntwo  + |h_{\ntwo,\nr}|^2 P_\ntwo)\\
    R_{\none,\nzero} + R_{\ntwo,\nzero} < \Delta_2 C(|h_{\none,\nzero}|^2 P_\none + |h_{\ntwo,\nzero}|^2 P_\ntwo) + \Delta_3 C(|h_{\nr,\nzero}|^2 P_\nr)\\
    R_{\none,\nzero} + R_{\ntwo,\nzero} < \Delta_2 C(|h_{\none,\nzero}|^2 P_\none + |h_{\none,\nr}|^2 P_\none + |h_{\ntwo,\nzero}|^2 P_\ntwo + |h_{\ntwo,\nr}|^2 P_\ntwo)\\
  \end{array}
\right.
\end{align}
To obtain the regions numerically, we optimize $\Delta_1$,
$\Delta_2$, and $\Delta_3$ for the given channel mutual
informations to maximize the achievable rate regions.

\end{itemize}

%%%%%%%%%%%%%%%%%%%%%%%%%%%%%%%%%%%%%%%%%%%%%%%%%%%%%%%%%%%%%

\section{Numerical analysis}

%%%%%%%%%%%%%%%%%%%%%%%%%%%%%%%%%%%%%%%%%%%%%%%%%%%%%%%%%%%%%
\label{sec:regions}
In this section, we numerically evaluate the rate regions obtained in
the previous section for the case of two terminal nodes $m=2$.
We first compare the achievable rate regions and outer bounds of the different protocols, using different combinations of encoding schemes, for both reciprocal channel $H_1$ and asymmetric channel $H_2$:

\begin{align}
{\bf H_1} = \left[
            \begin{array}{cccc}
              0 & 0.3 & 0.05 & 1 \\
              0.3 & 0 & 1.5 & 1 \\
              0.05 & 1.5 & 0 & 0.2 \\
              1 & 1 & 0.2 & 0 \\
            \end{array}
          \right]~~
          {\bf H_2} = \left[
            \begin{array}{cccc}
              0 & 0.9 & 0.4 & 1 \\
              0 & 0 & 0.02 & 1 \\
              0 & 0.02 & 0 & 0.5 \\
              1 & 1 & 0.5 & 0 \\
            \end{array}
          \right].
\end{align}
We then proceed to examine the maximal sum-rate $R_{\nzero,\none}+R_{\nzero,\ntwo}+R_{\none,\nzero}+R_{\ntwo,\nzero}$.
Finally, we dedicate the last section to the evaluation of the cooperation coding gain.

\subsection{Achievable rate region comparisons}
We compare the achievable rate regions of the different protocols, using different combinations of encoding schemes, with the simplest protocol. We set $P_\nzero= P_\none = P_\ntwo = P_\nr =0$ dB and $H=H_1$.
In \Fig \ref{fig:comp_protocol}, there are three achievable rate regions; 1) the simplest protocol (Simple), 2) convex hull of the FMABC, PMABC, FTDBC, and PTDBC protocols (MB) and 3) convex hull of the FMABC-N, PMABC-NR, FTDBC-NR, and PTDBC-NR protocols (MB-NR).
The 4-dimensional rate regions  $(R_{\nzero,\none},R_{\nzero,\ntwo},R_{\none,\nzero},R_{\ntwo,\nzero})$ are projected onto $(R_{\nzero,\none}+R_{\nzero,\ntwo}, R_{\none,\nzero}+R_{\ntwo,\nzero})$ 2-dimensional space.
For more realistic comparison, we add lower limits of individual data rates, i.e., $R_{\nzero,\none}\geq 0.01, R_{\nzero,\ntwo}\geq 0.01,R_{\none,\nzero}\geq 0.01, R_{\ntwo,\nzero}\geq 0.01$ to guarantee minimum information flow in each data link.
Without this limitation, the maximum sum-rates will be reached when both the transmission rates $R_{\nzero,\ntwo} $ and $R_{\ntwo,\nzero}$ equal zero at least in the Simplest case because the link between the relay and the node 2 is very poor. We want to emphasize that the value of the minimum data rate (set as 0.01 here)  do not affect the following simulation outcomes.
The Simple region is outer bounded by the MB region. This implies
that the proposed protocols using only conventional MAC and extended Marton's broadcasting coding largely enhance the performance. Furthermore, we can significantly  improve the achievable rate region by Network coding and Random binning schemes (in MB-NR).

%--------
\begin{figure}[t]
 \begin{center}
  \epsfig{figure=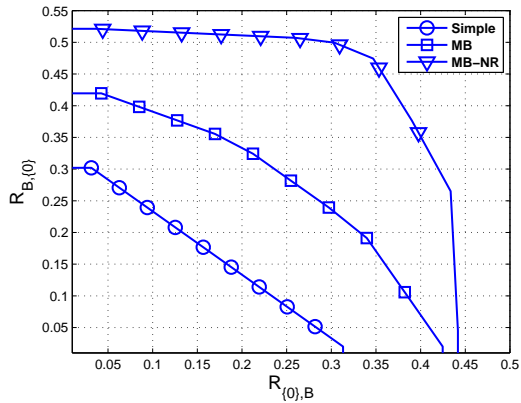, width=7.5cm}
  \caption{Comparison of protocol and coding gains with $P_\nzero= P_\none = P_\ntwo = P_\nr =0$ dB, $H=H_1$ and rate constraints ($R_{\nzero,\none}\geq 0.01, R_{\nzero,\ntwo}\geq 0.01,R_{\none,\nzero}\geq 0.01, R_{\ntwo,\nzero}\geq 0.01$).}
  \label{fig:comp_protocol}
 \end{center}
\end{figure}

%\subsection{Rate region comparisons with $m=2$}
We next evaluate the achievable rate regions and outer bounds for different SNR regimes and channel conditions. %are plotted in \Figs (\ref{fig:region_low}, \ref{fig:region_low_nr}) and \Figs (\ref{fig:region_mid2}, \ref{fig:region_high}), respectively.
We plot the inner bounds of the FMABC-N, PMABC-NR, FTDBC-NR, and PTDBC-NR protocols. The main outcome is that different protocols are optimal under different channel conditions.
This is because the amount of side information and multiple access interference is different. In the low SNR regime (\Fig \ref{fig:region_low}), the FMABC-N protocol outperforms the other protocols since the amount of both side information and multiple access interference is relatively small.  However, in the high SNR regime (\Fig \ref{fig:region_high}), the FTDBC-NR protocol becomes the best since it exploits side information more effectively.

In \Fig \ref{fig:region_low_nr} and \ref{fig:region_mid2}, the PMABC-NR protocol outperforms the other three protocols. The first case (\Fig \ref{fig:region_low_nr}) is when channel is asymmetric as $\nzero\rightarrow \none$ and $\nzero\rightarrow \ntwo$ are very good but the opposite direct links are almost disconnected. In this case, using side information in nodes $\none$ and $\ntwo$ and multiple access for node $\nzero$ would be the best choice since the quality of direct links are different. %As expected, in \Fig \ref{fig:region_low_nr} the PMABC-NR protocol achieves the best performance.
The second is when we have different power allocation.
Indeed, if we allow larger input power for the base station (node $\nzero$) and relay (node $\nr$),
the direct links from the base station are good enough to convey information. The terminal nodes can then exploit the side information efficiently.  Therefore, the PMABC-NR protocol has the best performance in this channel condition.

There is no significant different between two TDBC protocols, the FTDBC-NR and PTDBC-NR protocols. Since only difference is using multiple access or sequential transmitting for terminal nodes $(\in {\cal B})$, two type of side information, i.e., one is from node $\nzero$ to node $i$, the other is the opposite direction are both available. The small gap is from the efficiency of multiple access and generated interference.   

\begin{figure*}
  \hfill
  \begin{minipage}[t]{.45\textwidth}
    \begin{center}
       \epsfig{figure=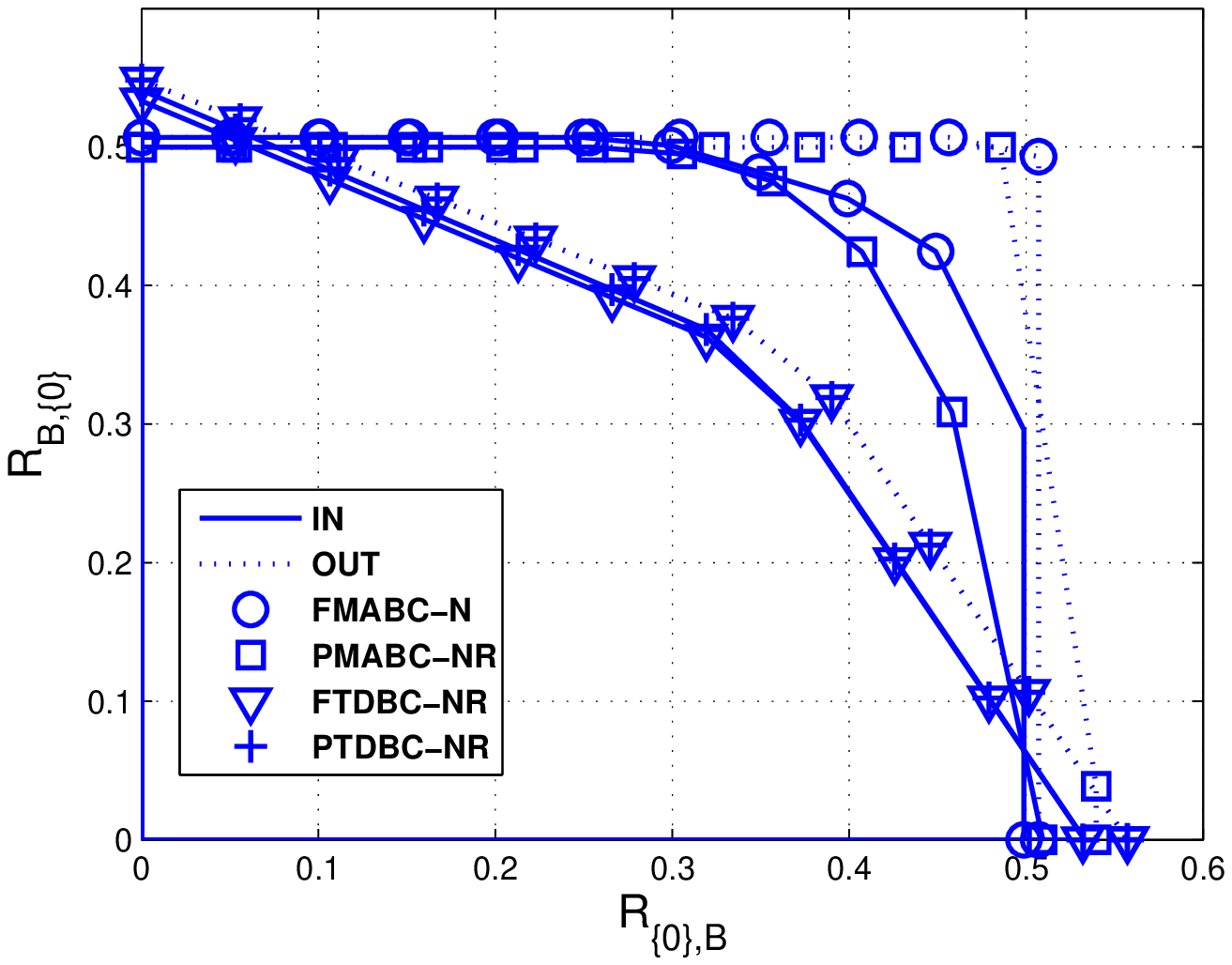, width=7.5cm}
      \caption{\baselineskip=10pt Comparison with $P_\nzero= P_\none = P_\ntwo = P_\nr =0$ dB, $H=H_1$.}
            \label{fig:region_low}
    \end{center}
  \end{minipage}
  \hfill
  \begin{minipage}[t]{.45\textwidth}
    \begin{center}
        \epsfig{figure=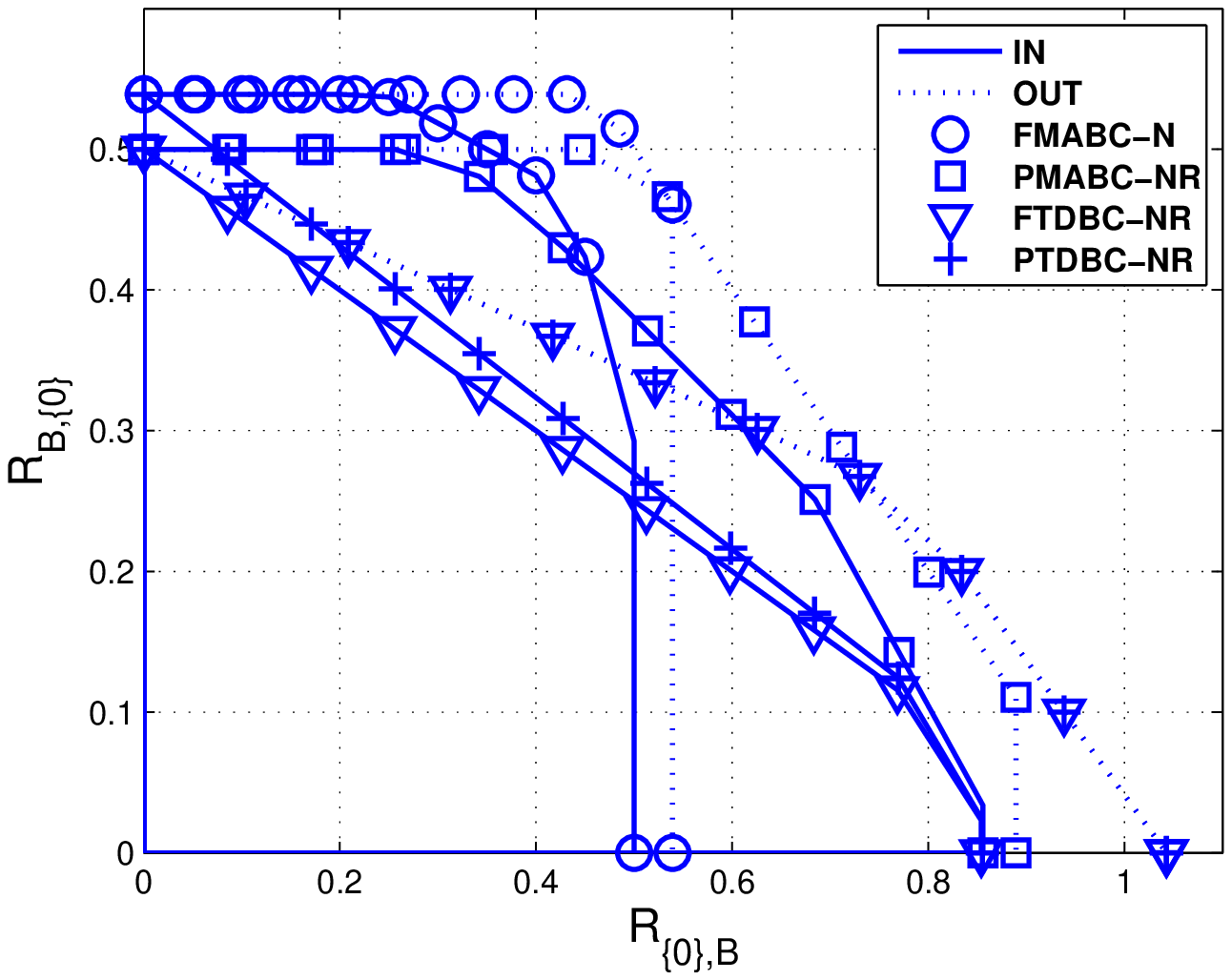, width=7.5cm}
      \caption{\baselineskip=10pt Comparison with $P_\nzero= P_\none = P_\ntwo = P_\nr =0$ dB, $H=H_2$.}
            \label{fig:region_low_nr}
    \end{center}
  \end{minipage}
  \hfill
\end{figure*}

\begin{figure*}
  \hfill
  \begin{minipage}[t]{.45\textwidth}
    \begin{center}
       \epsfig{figure=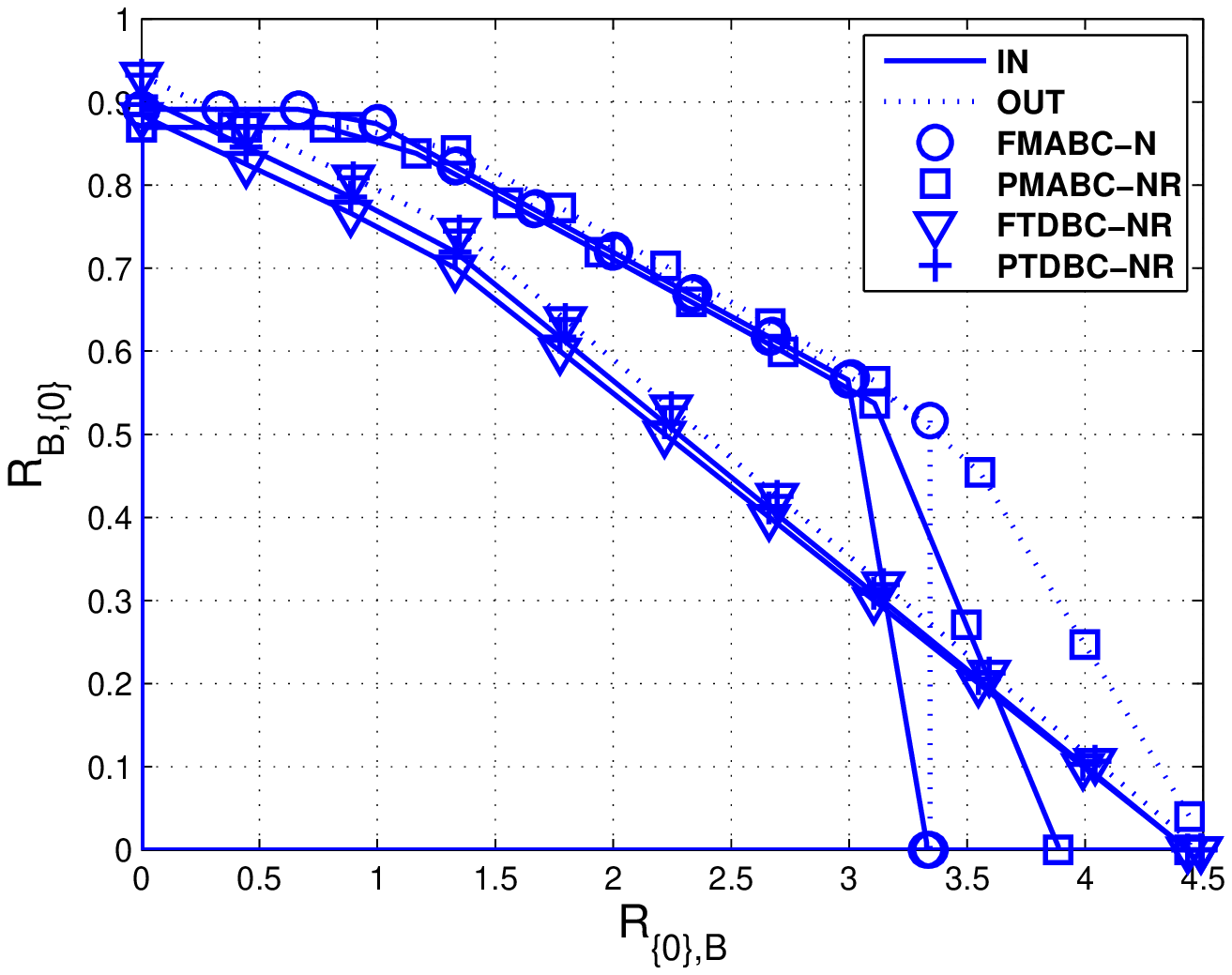, width=7.5cm}
      \caption{\baselineskip=10pt Comparison with $P_\nzero= P_\nr = 20, P_\none = P_\ntwo =0$ dB, $H=H_1$.}
            \label{fig:region_mid2}
    \end{center}
  \end{minipage}
  \hfill
  \begin{minipage}[t]{.45\textwidth}
    \begin{center}
        \epsfig{figure=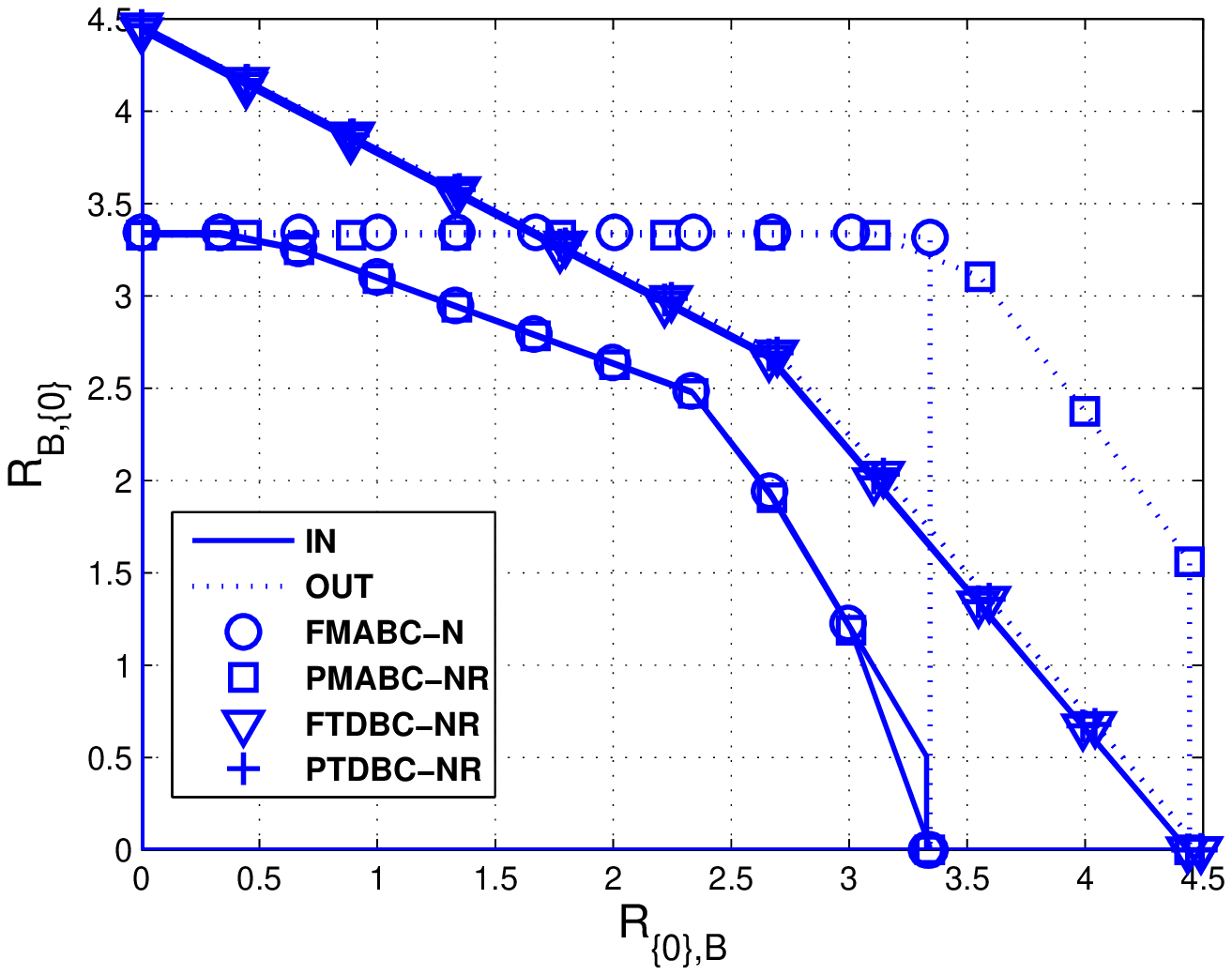, width=7.5cm}
      \caption{\baselineskip=10pt Comparison with $P_\nzero= P_\none = P_\ntwo = P_\nr=20$ dB, $H=H_1$.}
      \label{fig:region_high}
    \end{center}
  \end{minipage}
  \hfill
\end{figure*}

\subsection{Sum rate comparison}
In \Fig \ref{fig:sum_rate}, we plot the maximum sum rates $(R_{\nzero,\none}+R_{\nzero,\ntwo}+R_{\none,\nzero}+R_{\ntwo,\nzero})$ of FMABC-N, PMABC-NR, FTDBC-NR and PTDBC-NR protocols and the corresponding outer bounds.
As expected, in the low SNR regime ($\leq 15$ dB) two MABC protocols are better than the two FTDBC protocols, while the FTDBC protocols are better in the high SNR regime ($\geq 15$ dB). We also note that the slope for MABC protocols changes at around $\leq 15$ dB.
Indeed, the multiple access interference increases with SNR  and limit the maximum sum rate of the MABC protocols. However, the slopes are the same in the outer bounds since the sum rate constraints don't affect to the outer bounds.

\begin{figure}[t]
 \begin{center}
  \epsfig{figure=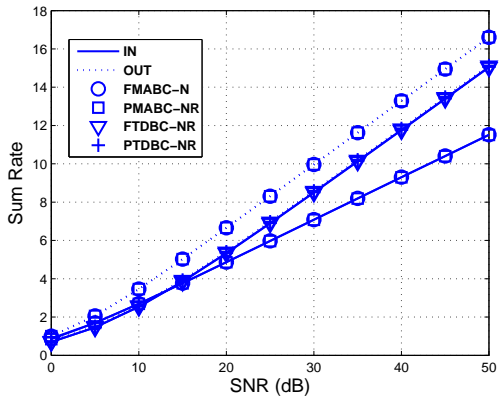, width=7.5cm}
      \caption{\baselineskip=10pt Sum rate comparison with $H=H_1$.}
            \label{fig:sum_rate}
 \end{center}
\end{figure}

\subsection{Cooperation coding gain}
To show the cooperation coding gain, we plot the achievable rate region of the different protocols with and without cooperation.  In \Fig \ref{fig:coop_ind}, we fixed the data rates
 $(R_{\nzero,\none},R_{\ntwo,\nzero})$ to the rate pair ((0.19,0.01) and plot rate regions in the $(R_{\none,\nzero},R_{\nzero,\ntwo})$ domain.  We do this to highlight the cooperation gain, which comes from re-allocating node $\none$'s transmission resources (i.e. relative power)  to the two information flows; $\none\rightarrow \nr ~(R_{\none,\nzero})$ and $ \none\rightarrow \ntwo~(R_{\nzero,\ntwo})$.  As expected -NRC protocols achieve much better performance than -NR protocols. Notably, the cooperation protocols improve $R_{\nzero,\ntwo}$ without any degradation of $R_{\none,\nzero}$ in the FTDBC protocol. In contrast, the maximum $R_{\none,\nzero}$ of the PMABC-NRC protocol is less than that of its PMABC-NR only protocol.
We explain this by the fact that in our achievable rate region, we used a simplified and sub-optimal (successive decoding like)  receiver in the PMABC-NRC protocol instead of using a fully general joint-decoder (as is done in the simpler PMABC-NR protocol), which limits the $R_{\none,\nzero}$. If we were to enhance the PMABC-NRC scheme by using the general joint decoder, the maximum $R_{\none,\nzero}$ would be reached and the overall performance would improve -  a technically challenging task left for future work.
% However, we will remain here this improvement for the further work.
We furthermore expect the gains of cooperation to increase if many more terminal nodes are able to exploit node $\none$'s cooperative broadcasting; however these situations  with current regions are too complex to be evaluated numerically.

\begin{figure}[t]
 \begin{center}
   \epsfig{figure=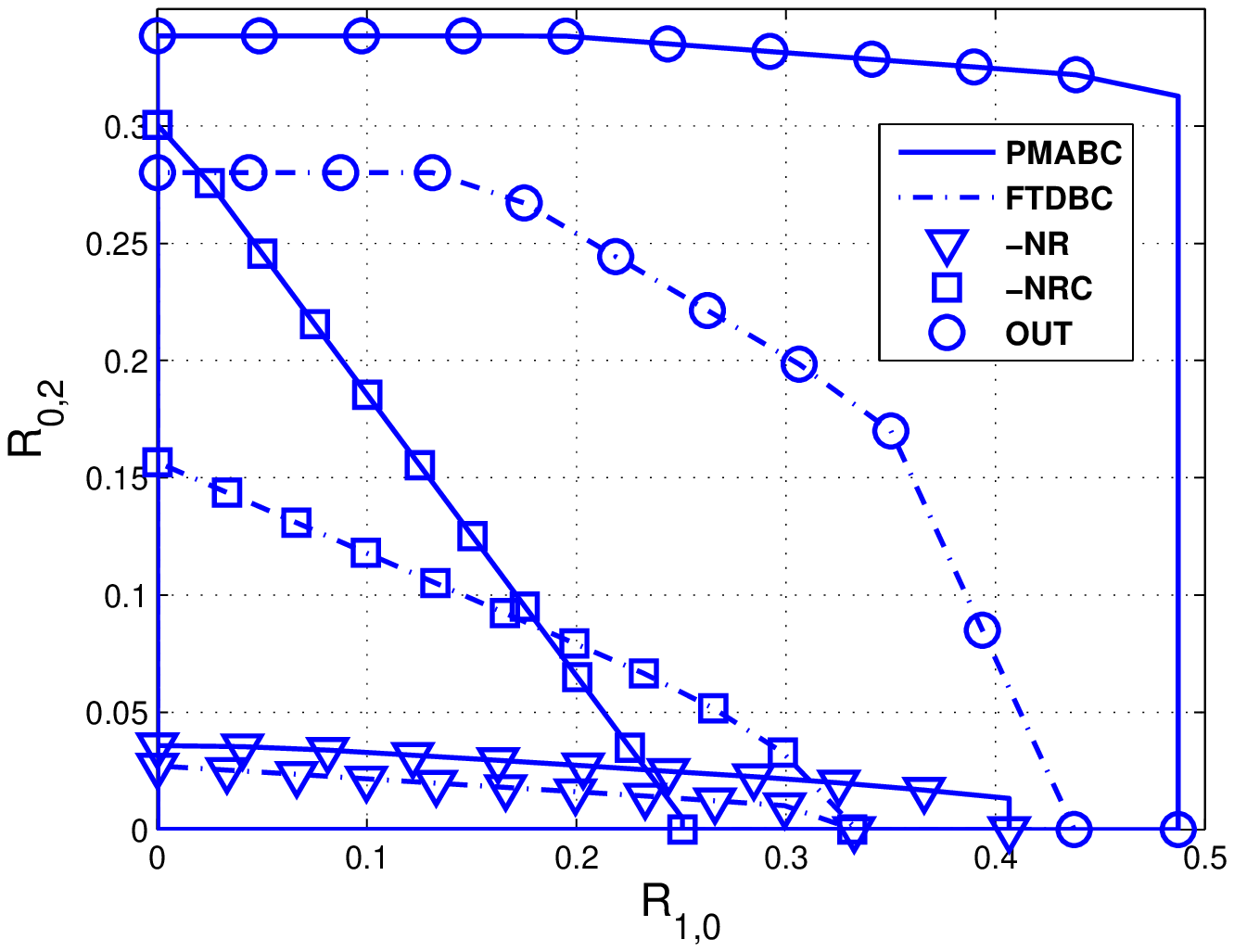, width=7.5cm}
  \caption{Comparison with $P_\nzero= P_\none = P_\ntwo = P_\nr=0$ dB, $H=H_1$ and $R_{\nzero,\none} =0.19$, $R_{\ntwo,\nzero} = 0.01$.}
  \label{fig:coop_ind}
 \end{center}
\end{figure}

%%%%%%%%%%%%%%%%%%%%%%%%%%%%%%%%%%%%%%%%%%%%%%%%%%%%%%

\section{Conclusion}

%%%%%%%%%%%%%%%%%%%%%%%%%%%%%%%%%%%%%%%%%%%%%%%%%%%%%%
\label{sec:conclusion}
In this paper, we introduced four new temporal protocols which
fully exploit the two-way nature of the data and outperform simple routing or multi-hop communication schemes:  the FMABC, PMABC, FTDBC, and PTDBC protocols.
We also proposed various coding schemes: Network coding, Random binning and user cooperations which exploit
over-heard and own-message side information.
 We derived achievable rate regions as well as outer bounds for the proposed protocols, using different combinations of encoding schemes, for a decode and forward relay. We compared these regions in the Gaussian noise channel. With numerical evaluations, we verified that different protocols achieve the best performance in different channel conditions and we highlighted the relative gains achieved by network coding, random binning and compress-and-forward-type
cooperation between terminal nodes.

\appendices
%%%%%%%%%%%%%%%%%%%%%%%%%%%%%%%%%%%%%%%%%%%%%%%%%%%%%%

\section{Lemma for Theorem \ref{theorem:gbc}}

%%%%%%%%%%%%%%%%%%%%%%%%%%%%%%%%%%%%%%%%%%%%%%%%%%%%%%
\label{app:lemma}

\begin{lemma}
\label{lemma:gbc}
For a given subset $S\subseteq {\cal B}$, $|S|>1$, we define ${\bf w} = \{w_i,i\in S\}$, ${\bf w}_0 = \{w_{i0}|i\in S\}$, ${\bf u}({\bf w}_0) = \{{\bf u}_i(w_{i0})|i\in S\}$, ${\bf U} = \{U_i|i\in S\}$ and the set
\begin{align}
D_{\bf w} \eqdef \left\{{\bf u}({\bf w}_0) \in A({\bf U}) | {\bf w}_0 \in \bigotimes_{i\in S} B^i_{w_i} \right\}.
\end{align}
Then for any choice of ${\bf w}$, $\epsilon>0$ and sufficiently large $n$ :
\begin{align}
P[\|D_{\bf w} \| = 0] \leq \epsilon
\end{align}
with
\begin{align}
\sum_{i\in S} R_i &< \sum_{i\in S} \left(I(U_i;Y_i) - I(U_i;U_{S(i)}) \right) -|S|\epsilon -\delta(\epsilon)
\end{align}
where $\delta(\epsilon)\rightarrow 0$ as $\epsilon\rightarrow 0$. \thmend
\end{lemma}

\begin{proof}
We use the similar proofs to Lemma in \cite{Gamal:1981}. From Chebychev's inequality, we have
\begin{align}
P[\|D_{\bf w} \| = 0] \leq P[|\|D_{\bf w}\| - E[\|D_{\bf w}\|]| > \epsilon E[\|D_{\bf w}\|]] \leq \frac{(\sigma[\|D_{\bf w} \|])^2}{\epsilon^2 (E[\|D_{\bf w}\|])^2}
\end{align}
and
\begin{align}
P[{\bf u}({\bf w}_0)\in D_{\bf w}] &\geq 2^{n(H({\bf U}) - \sum_{i\in S} H(U_i) -\delta(\epsilon))}\\
& = 2^{-n(\sum_{i\in S} I(U_i;U_{S(i)}) + \delta(\epsilon))}
\end{align}
Therefore,
\begin{align}
E[\|D_{\bf w}\|] &= \prod_{i\in S} \|B^i_{w_i}\| \cdot P[{\bf u}({\bf w}_0)\in D_{\bf w}]\\
&\geq 2^{n(\sum_{i\in S} (I(U_i;Y_i) - R_i -\epsilon) - \sum_{j\in S} I(U_j;U_{S(j)}) - \delta(\epsilon))}\label{lemma:gbc:1}
\end{align}
To find an upper bound of $(\sigma[\|D_{\bf w} \|])^2$, we first estimate $E[\|D_{\bf w}\|^2]$.
\begin{align}
E[\|D_{\bf w}\|^2] =&\sum_{{\bf u}({\bf w}_0) \in D_{\bf w}} P[{\bf u}({\bf w}_0)\in D_{\bf w}] + \nonumber\\ &\sum_{\twolines{{\bf u}({\bf w}_0),{\bf u}({\bf w}_1) \in D_{\bf w}}{{\bf u}({\bf w}_0)\neq {\bf u}({\bf w}_1)}} P[{\bf u}({\bf w}_0)\in D_{\bf w}]\cdot P[{\bf u}({\bf w}_1)\in D_{\bf w}]\\
\leq& \prod_{i\in S} \|B^i_{w_i}\| \cdot 2^{-n(\sum_{j\in S} I(U_j;U_{S(j)}) - \delta(\epsilon))}+ \nonumber\\
& \left( \prod_{i\in S} \|B^i_{w_i}\|^2 - \prod_{i\in S} \|B^i_{w_i}\|\right) \cdot 2^{-2n(\sum_{j\in S} I(U_j;U_{S(j)}) - \delta(\epsilon))} \label{lemma:gbc:2}
\end{align}
Thus, from \eqref{lemma:gbc:1} and \eqref{lemma:gbc:2} we have
\begin{align}
(\sigma[\|D_{\bf w} \|])^2 &= E[\|D_{\bf w}\|^2] - E[\|D_{\bf w}\|]^2\\
&\leq \prod_{i\in S} \|B^i_{w_i}\| \cdot 2^{-n(\sum_{j\in S} I(U_j;U_{S(j)}) - \delta(\epsilon))} \\
&\leq 2^{n(\sum_{i\in S} (I(U_i;Y_i) - R_i -\epsilon) - \sum_{j\in S} I(U_j;U_{S(j)}) + \delta(\epsilon))}\label{lemma:gbc:3}
\end{align}
From \eqref{lemma:gbc:1} and \eqref{lemma:gbc:3} for sufficiently large $n$
\begin{align}
P[\|D_{\bf w} \| = 0] \leq \epsilon
\end{align}
with
\begin{align}
\sum_{i\in S} R_i &< \sum_{i\in S} \left(I(U_i;Y_i) - I(U_i;U_{S(i)})\right)  -|S|\epsilon -\delta(\epsilon)
\end{align}
\end{proof}

%%%%%%%%%%%%%%%%%%%%%%%%%%%%%%%%%%%%%%%%%%%%%%%%%%%%%%

\section{Proof of Theorem \ref{theorem:FMABC-N}}

%%%%%%%%%%%%%%%%%%%%%%%%%%%%%%%%%%%%%%%%%%%%%%%%%%%%%%
\label{app:FMABC-N}

\begin{proof}
{\em Random code generation: } For simplicity
of exposition we take $|{\cal Q}| = 1$.
\begin{enumerate}
\item Phase 1 variables: Generate random $(n\cdot\Delta_{1,n})$-length sequences
\begin{itemize}
  \item ${\bf x}^\pa_\nzero(w_{\{\nzero\},{\cal B}})$ i.i.d. with $p^\pa(x_\nzero)$, $w_{\{\nzero\},{\cal B}} \in {\cal S}_{\{\nzero\},{\cal B}}$
  \item ${\bf x}^\pa_i(w_{i,\nzero})$ i.i.d. with $p^\pa(x_i)$, $w_{i,\nzero} \in {\cal S}_{i,\nzero}$, $\forall i\in{\cal B}$
\end{itemize}
\item Phase 2 variables: Generate random $(n\cdot\Delta_{2,n})$-length sequences
\begin{itemize}
  \item ${\bf u}^\pb_i(w_{\nr_i})$ with $p^\pb({\bf u}_i)$, $w_{\nr_i}  \in \{0,1,\cdots,\lfloor 2^{nR_{\nr_i}} \rfloor-1\} \eqdef  {\cal S}_{\nr_i}$, where
  \begin{align}
p^\pb({\bf u}_i) =\left\{
              \begin{array}{ll}
                \frac{1}{\|A^\pb(U_i) \|}, & {\bf u}_i\in A^\pb(U_i) \\
                0, & \hbox{otherwise.}
              \end{array}
            \right.
\end{align}
and $R_{\nr_i} = \max\{R_{\nzero,i},R_{i,\nzero}\} + (\Delta_2 I(U_i^\pb;Y_i^\pb) -4\epsilon - R_{\nzero,i})$ for $i\in[1,m]$.
\end{itemize}
and define bin  $B_j^i \eqdef \{w_{\nr_i} | w_{\nr_i} \in [j\cdot\lfloor2^{n(\Delta_2 I(U_i^\pb;Y_i^\pb)-R_{\nzero,i}-4\epsilon)}\rfloor, (j+1)\cdot\lfloor2^{n(\Delta_2 I(U_i^\pb;Y_i^\pb)-R_{\nzero,i}-4\epsilon)}\rfloor -1]\}$ for $j\in \{0,1,\cdots,\lfloor 2^{n\max\{R_{\nzero,i},R_{i,\nzero}\}}\rfloor -1\}$.
\end{enumerate}

{\em Encoding: } During phase 1, encoders of terminal nodes send the codewords ${\bf x}^\pa_\nzero(w_{\{\nzero\},{\cal B}})$, ${\bf x}^\pa_i(w_{i,\nzero})$. Relay $\nr$ estimates $\hat{w}_{\nzero,i}$ and $\hat{w}_{i,\nzero}$ after phase 1 using jointly typical decoding, then constructs $w_{i} =\hat{w}_{\nzero,i} \oplus \hat{w}_{i,\nzero}$. To transmit  messages $(w_1,\cdots,w_m)$, pick  $(w_{\nr_1},\cdots,w_{\nr_m})\in \bigotimes_{i=1}^m B^i_{w_i}$ such that ${\bf u}_T^\pb(w_{\nr_T})\in A^\pb(U_T)$, $\forall T\subseteq {\cal B}$, $|T|>1$, where $w_{\nr_T} \eqdef \{w_{\nr_i}|i\in T\}$. Such a $(w_{\nr_1},\cdots,w_{\nr_m})$ exists with high probability if
\begin{align}
\sum_{i\in T} R_{\nzero,i} &< \sum_{i\in T} \left(\Delta_2 I(U_i^\pb;Y_i^\pb) - \Delta_2 I(U_i^\pb;U_{T(i)}^\pb) \right) -|T|\epsilon -\delta(\epsilon)~~~~\forall T\subseteq {\cal B}~,~|T|>1  \label{eq:FMABC:DF:GC:7}
\end{align}
from Lemma \ref{lemma:gbc}. Then the relay finds a ${\bf x}_\nr^\pb$ jointly typical with $({\bf u}_1^\pb(w_{\nr_1}),\cdots,{\bf u}_m^\pb(w_{\nr_m}))$ and assigns it as the codeword corresponding to $(w_1,\cdots,w_m)$. Relay $\nr$ sends ${\bf x}^\pb_\nr(w_1,\cdots,w_m)$ during phase 2.

{\em Decoding: } Node $\nzero$ estimates $\tilde{w}_{i,\nzero}$ after phase 2 using jointly typical decoding. Since $w_{i} = w_{\nzero,i} \oplus w_{i,\nzero}$ and $\nzero$ knows $w_{\nzero,i}$, node $\nzero$ can reduce the number of possible $w_{\nr_i}$  and likewise at node $i$, the cardinality of $w_{\nr_i}$ is $2^{n(\Delta_2 I(U_i^\pb;Y_i^\pb)-4\epsilon)}$.

{\em Error analysis: }
Recall that $w_{\nr_T} = \{w_{\nr_i}| i\in T\}$.
Then, %$\forall i\in [1,m]$,
\begin{align}
P[E] &= P[E_{enc} \cup E_{dec}] \leq P[E_{enc}] + P[E_{dec}]\\
     &\leq \epsilon + \sum_{i=1}^m P[E_{\{\nzero\},\{i\}}] + P[E_{\{i\},\{\nzero\}}]
\end{align}
where $E$ is the entire decoding error event, $E_{enc}$ is the set of encoding error events, and $E_{dec}$ is the set of decoding error events. $P[E_{enc}]$ may be driven to $0$ as $n\rightarrow \infty$ by \eqref{eq:FMABC:DF:GC:7} and $E_{dec}$ may be decomposed into individual decoding error events on each link as:
\begin{align}
  P[E_{\{\nzero\},\{i\}}] & \leq P[E_{{\cal M},\{\nr\}}^\pa \cup E_{\{\nr\},\{i\}}^\pb]\\
  & \leq P[E_{{\cal M},\{\nr\}}^\pa] + P[E_{\{\nr\},\{i\}}^\pb | \bar{E}_{{\cal M},\{\nr\}}^\pa]\\
  P[E_{\{i\},\{\nzero\}}] & \leq P[E_{{\cal M},\{\nr\}}^\pa \cup E_{\{\nr\},\{\nzero\}}^\pb]\\
  & \leq P[E_{{\cal M},\{\nr\}}^\pa] + P[E_{\{\nr\},\{\nzero\}}^\pb | \bar{E}_{{\cal M},\{\nr\}}^\pa]
\end{align}
$P[E_{{\cal M},\{\nr\}}^\pa]$ tends to zero as $n \rightarrow \infty$ by Theorem 15.3.6 in \cite{Cover:2006} and the condition in \eqref{eq:FMABC:DF:GC:1}.
\begin{align}
 P[E_{\{\nr\},\{i\}}^\pb | \bar{E}_{{\cal M},\{\nr\}}^\pa] \leq& P[\bar{D}^\pb({\bf u}_i (w_{\nr_i}),{\bf y}_i)] +
  P[\cup_{\tilde{w}_{\nr_i} \in \cup B^i_{{\tilde w}_{\nzero,i}\oplus w_{i,\nzero}}} D^\pb({\bf u}_i( \tilde{w}_{\nr_i}),{\bf y}_i)]\\
  \leq& 2\epsilon  \label{eq:FMABC:DF:GC:4}\\
 P[E_{\{\nr\},\{\nzero\}}^\pb | \bar{E}_{{\cal M},\{\nr\}}^\pa] \leq& P[\bar{D}^\pb({\bf u}_{\cal B}(w_{\nr_{\cal B}}),{\bf y}_\nzero)] +\nonumber\\
  & \sum_{T\subseteq{\cal B}} P[\cup_{\tilde{w}_{T,\{\nzero\}} \not= w_{T,\{\nzero\}}} D^\pb({\bf u}_T(\tilde{w}_{\nr_T}),{\bf u}_{\bar T}(w_{\nr_{\bar T}}),{\bf y}_\nzero)]\\
  \leq& \epsilon + \sum_{T\subseteq {\cal B}} 2^{nR_{T,\{\nzero\}}} 2^{-n\cdot\Delta_{2,n}(I(U_T^\pb;Y_\nzero^\pb,U_{\bar T}^\pb)-\epsilon')} \label{eq:FMABC:DF:GC:5}
\end{align}
The total cardinality of the case (${\tilde w}_{\nr_T} \neq w_{\nr_T})$ is bounded by $2^{nR_{T,\{\nzero\}}}$ since ${\tilde w}_{\nr_T}$ is uniquely specified by  ${\tilde w}_{T,\{\nzero\}}$ if $w_{{\bar T},\{\nzero\}}$ is given.

Since $\epsilon > 0$ is arbitrary, the conditions of Theorem \ref{theorem:FMABC-N} and the AEP property ensure that the right hand sides of \eqref{eq:FMABC:DF:GC:4} and \eqref{eq:FMABC:DF:GC:5} tend
to 0 as $n \rightarrow \infty$.  By
Fenchel-Bunt's theorem in \cite{Hiriart:2001}, it is sufficient to
restrict $|{\cal Q}| \leq 2^{m+1} -1$ since we have $2^{m+1} -1$ in \eqref{eq:FMABC:DF:GC:1}.
\end{proof}

%%%%%%%%%%%%%%%%%%%%%%%%%%%%%%%%%%%%%%%%%%%%%%%%%%%%%%

\section{Proof of Theorem \ref{theorem:PMABC-NR}}

%%%%%%%%%%%%%%%%%%%%%%%%%%%%%%%%%%%%%%%%%%%%%%%%%%%%%%
\label{app:PMABC-NR}

\begin{proof}
{\em Random code generation: } For simplicity
of exposition only, we take $|{\cal Q}| = 1$. For all $i\in [1,m]$, first we generate a partition of ${\cal S}_{\nzero,i}$  randomly by independently assigning every index $w_{\nzero,i} \in {\cal S}_{\nzero,i}$ to a set ${\cal S}_{\nzero,i}(k)$, with a uniform distribution over the indices $k \in
\{0, \ldots, \lfloor 2^{nR_{\nzero,i:1}} \rfloor - 1\} \eqdef {\cal S}_{\nzero,i:1}$. We denote by $s_{\nzero,i}(w_{\nzero,i})$
the index $k$ of ${\cal S}_{\nzero,i}(k)$ to which $w_{\nzero,i}$  belongs.

\begin{enumerate}
\item Phase $j$ ($\in [1,m]$): Generate random $(n\cdot\Delta_{j,n})$-length sequences
\begin{itemize}
  \item ${\bf v}^{(j)}_{\nzero i}(w_{\nzero,i:2})$ with $p^{(j)}({\bf v}_{\nzero i})$, $w_{\nzero,i:2}  \in \{0,1,\cdots,\lfloor 2^{nR_{\nzero,i:2}} \rfloor-1\} \eqdef  {\cal S}_{\nzero,i:2}$, where
  \begin{align}
p^{(j)}({\bf v}_{\nzero i}) =\left\{
              \begin{array}{ll}
                \frac{1}{\|A^{(j)}(V_{\nzero i}) \|}, & {\bf v}_{\nzero i}\in A^{(j)}(V_{\nzero i}) \\
                0, & \hbox{otherwise.}
              \end{array}
            \right.
\end{align}
and $R_{\nzero,i:2} = R_{\nzero,i:1} + \sum_{j=1}^m \Delta_j I(V_{\nzero i}^{(j)};Y_i^{(j)}) -\epsilon'$. Then we define bin $B_k^{i} \eqdef \{w_{\nzero,i:2} | w_{\nzero,i:2} \in [k\cdot\lfloor2^{n(\sum_{j=1}^m\Delta_j I(V_{\nzero i}^{(j)};Y_i^{(j)})+R_{\nzero,i:1}-R_{\nzero,i}-\epsilon')}\rfloor, (k+1)\cdot\lfloor2^{n(\sum_{j=1}^m \Delta_j I(V_{\nzero i}^{(j)};Y_i^{(j)})+R_{\nzero,i:1}-R_{\nzero,i}-\epsilon')}\rfloor -1]\}$ for $k\in \{0,1,\cdots,\lfloor 2^{n R_{\nzero,i}}\rfloor -1\}$.
 \item ${\bf x}^{(j)}_j(w_{j,\nzero})$ i.i.d. with $p^{(j)}(x_j)$, $w_{j,\nzero} \in {\cal S}_{j,\nzero}$
\end{itemize}

\item Phase $m+1$: Generate random $(n\cdot\Delta_{m+1,n})$-length sequences
\begin{itemize}
  \item ${\bf u}^{(m+1)}_i(w_{\nr_i})$ with $p^{(m+1)}({\bf u}_i)$, $w_{\nr_i}  \in \{0,1,\cdots,\lfloor 2^{nR_{\nr_i}} \rfloor-1\} \eqdef  {\cal S}_{\nr_i}$, where
  \begin{align}
p^{(m+1)}({\bf u}_i) =\left\{
              \begin{array}{ll}
                \frac{1}{\|A^{(m+1)}(U_i) \|}, & {\bf u}_i\in A^{(m+1)}(U_i) \\
                0, & \hbox{otherwise.}
              \end{array}
            \right.
\end{align}
and $R_{\nr_i} = \max\{R_{\nzero,i:1},R_{i,\nzero}\} + (\Delta_{m+1} I(U_i^{(m+1)};Y_i^{(m+1)}) -4\epsilon - R_{\nzero,i:1})$.
\end{itemize}
and define bin  $C_k^i \eqdef \{w_{\nr_i} | w_{\nr_i} \in [k\cdot\lfloor2^{n(\Delta_{m+1} I(U_i^{(m+1)};Y_i^{(m+1)})-R_{\nzero,i:1}-4\epsilon)}\rfloor ,$ \\$(k+1)\cdot\lfloor2^{n(\Delta_{m+1} I(U_i^{(m+1)};Y_i^{(m+1)})-R_{\nzero,i:1}-4\epsilon)}\rfloor -1]\}$ for $k\in \{0,1,\cdots,\lfloor 2^{n\max\{R_{\nzero,i:1},R_{i,\nzero}\}}\rfloor -1\}$.
\end{enumerate}

{\em Encoding: }
\begin{enumerate}
\item To transmit $(w_{\nzero,1},\cdots,w_{\nzero,m})$, node $\nzero$ picks $(w_{\nzero,1:2},\cdots,w_{\nzero,m:2})\in \bigotimes_{i=1}^m B^i_{w_{\nzero,i}}$ such that \\${\bf v}_{\nzero S}^{(j)}(w_{\{\nzero\} ,S:2})\in A^{(j)}(V_{\nzero S} )$, $\forall j\in[1,m], \forall S \subseteq {\cal B}~,~ |S|>1$. Such a $(w_{\nzero,1:2},\cdots,w_{\nzero,m:2})$ exists with high probability if
\begin{align}
\sum_{i\in S} \left(R_{\nzero,i} -R_{\nzero,i:1}\right) &< \sum_{i\in S} \sum_{j=1}^m \Delta_j I(V_{\nzero i}^{(j)};Y_i^{(j)}) - \Delta_j I(V_{\nzero i}^{(j)};V_{\nzero S(i)}^{(j)}) -\delta(\epsilon) \label{eq:PMABC:GC:10}
\end{align}
for $S\subseteq {\cal B}~,~ |S|>1$ from Lemma \ref{lemma:gbc}.  Then node $\nzero$ finds a ${\bf x}_\nzero^{(i)}$ jointly typical with \\$({\bf v}_{\nzero 1}^{(i)}(w_{\nzero,1:2}),\cdots,{\bf v}_{\nzero m}^{(i)}(w_{\nzero,m:2}))$ and designate it as the codeword corresponding to $(w_{\nzero,1},\cdots,w_{\nzero,m})$ for all $i\in[1,m]$. Node $\nzero$ sends ${\bf x}^{(i)}_\nzero(w_{\nzero,1},\cdots,w_{\nzero,m})$ during phase $i$.

\item
During phase $i$, encoder of terminal node $i$ sends the codeword ${\bf x}^{(i)}_i(w_{i,\nzero})$ for $i\in[1,m]$.

\item
Relay $\nr$ estimates $\hat{w}_{\nzero,i}$ and $\hat{w}_{i,\nzero}$ after phase $m$ using jointly typical decoding, then constructs $w_{i} =\hat{w}_{\nzero,i:1} \oplus \hat{w}_{i,\nzero}$. To transmit a pair of messages $(w_1,\cdots,w_m)$, pick a pair $(w_{\nr_1},\cdots,w_{\nr_m})\in \bigotimes_{i=1}^m C^i_{w_i} $ such that ${\bf u}_S^{(m+1)}(w_{\nr_S})\in A^{(m+1)}(U_S)$. Such a $(w_{\nr_1},\cdots,w_{\nr_m})$ exists with high probability if
\begin{align}
\sum_{i\in S} R_{\nzero,i:1} &< \sum_{i\in S} \Delta_{m+1} I(U_i^{(m+1)};Y_i^{(m+1)}) - \Delta_{m+1} I(U_i^{(m+1)};U_{S(i)}^{(m+1)})  -|S|\epsilon -\delta(\epsilon) \label{eq:PMABC:GC:11}
\end{align}
for $S\subseteq {\cal B}~,~|S| >1$ from Lemma \ref{lemma:gbc}. Then the relay finds a ${\bf x}_\nr^{(m+1)}$ jointly typical with \\$({\bf u}_1^{(m+1)}(w_{\nr_1}),\cdots,{\bf u}_m^{(m+1)}(w_{\nr_m}))$ and designate it as the codeword corresponding to $(w_1,\cdots,w_m)$. Relay $\nr$ sends ${\bf x}^{(m+1)}_\nr(w_1,\cdots,w_m)$ during phase $m+1$.
\end{enumerate}

{\em Decoding: } For all $i\in [1,m]$,
\begin{enumerate}
\item Node $\nzero$ estimates $\tilde{w}_{i,\nzero}$ after phase $m+1$ using jointly typical decoding. Since $w_{i} = w_{\nzero,i:1} \oplus w_{i,\nzero}$ and $\nzero$ knows $w_{\nzero,i}$, node $\nzero$ can reduce the number of possible $w_{i}$.

\item Node $i$ estimates ${\tilde w}_{\nr_i}$ after phase $m+1$ using jointly typical decoding. Similar to the case of node $\nzero$, node $i$ can reduce the cardinality of $w_{\nr_i}$ to $2^{n(\Delta_{m+1} I(U_i^{(m+1)};Y_i^{(m+1)})-4\epsilon)}$ and node $i$ decodes ${\tilde w}_{\nzero,i:1}$ from the bin index of ${\tilde w}_{\nr_i}$. Then, node $i$ estimates ${\tilde w}_{\nzero,i:2}$ using jointly typical decoding of the sequence $({\bf v}_{\nzero i}^{(j)}({\tilde w}_{\nzero,i:2}),{\bf y}_i^{(j)})$, for all $j\neq i$. Since node $i$ knows the bin index $s_{{\tilde w}_{\nzero,i}}({\tilde w}_{\nzero,i})$ as ${\tilde w}_{\nzero,i:1}$, it can reduce the cardinality of $w_{\nzero,i:2}$ to $2^{n(\sum_{j\neq i}\Delta_j I(V_{\nzero i}^{(j)};Y_i^{(j)}) - \epsilon')}$. After decoding ${\tilde w}_{\nzero,i:2}$, node $i$ finally decodes ${\tilde w}_{\nzero,i}$ from the bin index of ${\tilde w}_{\nzero,i:2}$.
\end{enumerate}

{\em Error analysis: }
\begin{align}
P[E] &= P[E_{enc} \cup E_{dec}] \leq P[E_{enc}] + P[E_{dec}]\\
     &\leq \epsilon + \sum_{i=1}^m P[E_{\{\nzero\},\{i\}}] + P[E_{\{i\},\{\nzero\}}]
\end{align}
where $E$ is the  entire error event, $E_{enc}$ is the set of encoding error events, and $E_{dec}$ is the set of decoding error events. $P[E_{enc}]$ is upper bounded by sufficiently small number $\epsilon$ from  \eqref{eq:PMABC:GC:10} and \eqref{eq:PMABC:GC:11}. $E_{dec}$ can be separated by individual decoding error events in each link.
Then $\forall i \in [1,m]$,
\begin{align}
  P[E_{\{\nzero\},\{i\}}]  \leq& P[(\cup_{j=1}^m E_{\{\nzero,j\},\{\nr\}}^{(j)}) \cup E_{\{\nr\},\{i\}}^{(m+1)} \cup E_{\{\nzero\},\{i\}}^{(m+1)} ]\\
   \leq & P[\cup_{j=1}^m E_{\{\nzero,j\},\{\nr\}}^{(j)}] + P[ E_{\{\nr\},\{i\}}^{(m+1)}|\cap_{j=1}^m {\bar E}_{\{\nzero,j\},\{\nr\}}^{(j)}] + \nonumber\\
   &P[ E_{\{\nzero\},\{i\}}^{(m+1)}|(\cap_{j=1}^m {\bar E}_{\{\nzero,j\},\{\nr\}}^{(j)}) \cap \bar{ E}_{\{\nr\},\{i\}}^{(m+1)}] \\
  P[E_{\{i\},\{\nzero\}}]  \leq& P[(\cup_{j=1}^m E_{\{\nzero,j\},\{\nr\}}^{(j)}) \cup E_{\{\nr\},\{\nzero\}}^{(m+1)}]\\
   \leq & P[\cup_{j=1}^m E_{\{\nzero,j\},\{\nr\}}^{(j)}] + P[E_{\{\nr\},\{\nzero\}}^{(m+1)} | \cap_{j=1}^m {\bar E}_{\{\nzero,j\},\{\nr\}}^{(j)}]
\end{align}

Now we will show that $P[\cup_{j=1}^m E_{\{\nzero,j\},\{\nr\}}^{(j)}]$, $P[ E_{\{\nr\},\{i\}}^{(m+1)}|\cap_{j=1}^m {\bar E}_{\{\nzero,j\},\{\nr\}}^{(j)}]$, $P[ E_{\{\nzero\},\{i\}}^{(m+1)}|(\cap_{j=1}^m {\bar E}_{\{\nzero,j\},\{\nr\}}^{(j)}) \cap \bar{ E}_{\{\nr\},\{i\}}^{(m+1)}]$ and $P[E_{\{\nr\},\{\nzero\}}^{(m+1)} | \cap_{j=1}^m {\bar E}_{\{\nzero,j\},\{\nr\}}^{(j)}]$ tend to zero as $n\rightarrow \infty$. For the convenience of analysis we define $B^i_R(w_{\nzero,i:1})= \cup_{w_{\nzero,i} \in {\cal S}_{\nzero,i}(w_{\nzero,i:1})} B^i_{w_{\nzero,i}}$ and $C^i_R(w_{i,\nzero}) = \cup_{w_{\nzero,i:1} \in {\cal S}_{\nzero,i:1}} C^i_{w_{\nzero,i:1}\oplus w_{i,\nzero}}$. Then,
\begin{align}
 P[\cup_{j=1}^m &E_{\{\nzero,j\},\{\nr\}}^{(j)}] \nonumber\\
 \leq &\sum_{j=1}^m P[\bar{D}^{(j)}({\bf v}_{\nzero {\cal B}}(w_{\{\nzero\},{\cal B}:2}),{\bf x}_j( w_{j,\nzero}),{\bf y}_\nr)] +
  P\left[\bigcup_{\tilde{w}_{{\cal M},{\cal M}} \not= w_{{\cal M},{\cal M}}} \left(\bigcap_{j=1}^m D^{(j)}({\bf v}_{\nzero {\cal B}}( \tilde{w}_{\{\nzero\},{\cal B}:2}),{\bf x}_j( \tilde{w}_{j,\nzero}),{\bf y}_\nr)\right) \right]\\
 \leq& m\cdot \epsilon + \sum_{T,S\in {\cal B}}\sum_{\tilde{w}_{\{\nzero\},T}} \sum_{\tilde{w}_{S,\{\nzero\}}} \left[ \prod_{s\in S} P[D^{(s)}({\bf v}_{\nzero T}( \tilde{w}_{\{\nzero\},T:2}),{\bf v}_{\nzero {\bar T}}( w_{\{\nzero\},{\bar T}:2}),{\bf x}_s(\tilde{w}_{s,\nzero}),{\bf y}_\nr)]\cdot\right.\nonumber\\
 &~~~~~~~~~~~~~~~~~~~~~~~~~~~~~~\left.\prod_{s\in {\bar S}} P[D^{(s)}({\bf v}_{\nzero T}( \tilde{w}_{\{\nzero\},T:2}),{\bf v}_{\nzero {\bar T}}( w_{\{\nzero\},{\bar T}:2}),{\bf x}_s(w_{s,\nzero}),{\bf y}_\nr)]\right]\\
 \leq& m\epsilon + \sum_{T,S\in {\cal B}} 2^{n(R_{\{\nzero\},T}+R_{S,\{\nzero\}})}\cdot 2^{-n(\sum_{s\in S}\Delta_{s,n}I(V_{\nzero T} ^{(s)},X_s^{(s)};Y_\nr^{(s)},V_{\nzero {\bar T}}^{(s)}) +\sum_{s\in {\bar S}}\Delta_{s,n}I(V_{\nzero T} ^{(s)};Y_\nr^{(s)}, V_{\nzero {\bar T}}^{(s)} |X_{s}^{(s)})-\epsilon'')}\label{eq:PMABC:GC:6}
\end{align}
The total cardinality of the case (${\tilde w}_{\{\nzero\},T:2} \neq w_{\{\nzero\},T:2})$ is bounded by $2^{nR_{\{\nzero\},T}}$ since ${\tilde w}_{\{\nzero\},T:2}$ is uniquely specified with ${\tilde w}_{\{\nzero\},T}$ if $w_{\{\nzero\},{\bar T}:2}$ is given.
 Also,
\begin{align}
 P[ E_{\{\nr\},\{i\}}^{(m+1)}|\cap_{j=1}^m {\bar E}_{\{\nzero,j\},\{\nr\}}^{(j)}] \leq& P[\bar{D}^{(m+1)}({\bf u}_i (w_{\nr_i}),{\bf y}_i)] + P[\cup_{\tilde{w}_{\nr_i}\in C^i_R(w_{i,\nzero})} D^{(m+1)}({\bf u}_i(\tilde{w}_{\nr_i}),{\bf y}_i)]\\
 \leq& \epsilon + |C^i_R(w_{i,\nzero})|\cdot 2^{-n(\Delta_{m+1,n} I(U_i^{(m+1)};Y_i^{(m+1)}) - 3\epsilon)}
\label{eq:PMABC:GC:7}
\end{align}
\begin{align}
P[ E_{\{\nzero\},\{i\}}^{(m+1)}|(\cap_{j=1}^m {\bar E}_{\{\nzero,j\},\{\nr\}}^{(j)}) \cap \bar{ E}_{\{\nr\},\{i\}}^{(m+1)}] \leq& \sum_{j\neq i} P[\bar{D}^{(j)}({\bf v}_{\nzero i} (w_{\nzero,i:2}),{\bf y}_i)] +\nonumber\\
 & P\left[\cup_{\tilde{w}_{\nzero,i:2}\in B^i_R(w_{\nzero,i:1})} \left(\cap_{j\neq i} D^{(j)}({\bf v}_{\nzero i}(\tilde{w}_{\nzero,i:2}),{\bf y}_i)\right)\right]\\
 \leq& (m-1)\epsilon + |B^i_R(w_{\nzero,i:1})|\cdot 2^{-n(\sum_{j\neq i} \Delta_{j,n} I(V_{\nzero i}^{(j)};Y_i^{(j)}) - \epsilon'')}\label{eq:PMABC:GC:8}
\end{align}
\begin{align}
 P[E_{\{\nr\},\{\nzero\}}^{(m+1)} | \cap_{j=1}^m {\bar E}_{\{\nzero,j\},\{\nr\}}^{(j)}] \leq &P[\bar{D}^{(m+1)}({\bf u}_{\cal B}(w_{\nr_{\cal B}}),{\bf y}_\nzero)] + \nonumber\\
 & \sum_{S\subseteq{\cal B}} P[\cup_{\tilde{w}_{S,\{\nzero\}} \not= w_{S,\{\nzero\}}} D^{(m+1)}({\bf u}_S(\tilde{w}_{\nr_S}),{\bf u}_{\bar S}(w_{\nr_{\bar S}}),{\bf y}_\nzero)]\\
  \leq& \epsilon + \sum_{S\subseteq {\cal B}} 2^{nR_{S,\{\nzero\}}} 2^{-n\cdot\Delta_{m+1,n}(I(U_S^{(m+1)};Y_\nzero^{(m+1)},U_{\bar S}^{(m+1)})-\epsilon')} \label{eq:PMABC:GC:9}
\end{align}
The total cardinality of the case (${\tilde w}_{\nr_S} \neq w_{\nr_S})$ is bounded by $2^{nR_{S,\{\nzero\}}}$ since ${\tilde w}_{\nr_S}$ is uniquely specified with ${\tilde w}_{S,\{\nzero\}}$ if $w_{{\bar S},\{\nzero\}}$ is given.

Since $\epsilon > 0$ is arbitrary, with the conditions of Theorem \ref{theorem:PMABC-NR}, a proper choice of $\{R_{\nzero,i:1}\}$ and the AEP property, we can make the right hand sides of \eqref{eq:PMABC:GC:6}, \eqref{eq:PMABC:GC:7}, \eqref{eq:PMABC:GC:8} and \eqref{eq:PMABC:GC:9} tend
to 0 as $n \rightarrow \infty$.  By
Fenchel-Bunt's theorem in \cite{Hiriart:2001}, it is sufficient to
restrict $|{\cal Q}| \leq 2^{2m}+2^m$ since $2^{2m}$ inequalities are from \eqref{eq:PMABC:GC:1} and  $2^m$ inequalities are from \eqref{eq:PMABC:GC:4}.
\end{proof}

%%%%%%%%%%%%%%%%%%%%%%%%%%%%%%%%%%%%%%%%%%%%%%%%%%%%%%

\section{Proof of Theorem \ref{theorem:PMABC-NRC}}

%%%%%%%%%%%%%%%%%%%%%%%%%%%%%%%%%%%%%%%%%%%%%%%%%%%%%%
\label{app:PMABC-NRC}

\begin{proof}
{\em Random code generation: }For simplicity
of exposition only, we take $|{\cal Q}| = 1$. For all $i\in [1,m]$, nodes $\nzero$ and $i$ divide $w_{\nzero,i }$ and $w_{i,\nzero}$ into $K$ blocks, then nodes $\nzero$ and $i$ have message set $\{w_{\nzero,i|(1)},w_{\nzero,i|(2)},\cdots,w_{\nzero,i|(K)}\}$ and $\{w_{i,\nzero|(1)},w_{i,\nzero|(2)},\cdots,w_{i,\nzero|(K)}\}$, respectively. Then we generate a partition of ${\cal S}_{\nzero,i}$ randomly by independently assigning every index $w_{\nzero,i|(k)} \in {\cal S}_{\nzero,i}$ to a set ${\cal S}_{\nzero,i}(h)$, with a uniform distribution over the indices $h \in
\{0, \ldots, \lfloor 2^{nR_{\nzero,i:1}} \rfloor - 1\} \eqdef {\cal S}_{\nzero,i:1}$. We denote by $s_{\nzero,i}(w_{\nzero,i|(k)})$
the index $h$ of ${\cal S}_{\nzero,i}(h)$ to which $w_{\nzero,i|(k)}$  belongs.

\begin{enumerate}
\item Phase $j$ ($\in [1,m]$): Generate random $(n\cdot\Delta_{j,n})$-length sequences
\begin{itemize}
  \item ${\bf v}^{(j)}_{\nzero i}(w_{\nzero,i:2})$ with $p^{(j)}({\bf v}_{\nzero i})$, $w_{\nzero,i:2}  \in \{0,1,\cdots,\lfloor 2^{nR_{\nzero,i:2}} \rfloor-1\} \eqdef  {\cal S}_{\nzero,i:2}$, where
  \begin{align}
p^{(j)}({\bf v}_{\nzero i}) =\left\{
              \begin{array}{ll}
                \frac{1}{\|A^{(j)}(V_{\nzero i}) \|}, & {\bf v}_{\nzero i}\in A^{(j)}(V_{\nzero i}) \\
                0, & \hbox{otherwise.}
              \end{array}
            \right.
\end{align}
For convenience of analysis we define $R_{\nzero,i:3} = \sum_{j\in {\cal J}_i} \Delta_j I(V_{\nzero i}^{(j)};Y_i^{(j)}|V_{j 2}^{(j)}) + \sum_{j\not \in {\cal J}_i} \Delta_j I(V_{\nzero i}^{(j)};Y_i^{(j)})$ and take $R_{\nzero,i:2} = R_{\nzero,i:1} + R_{\nzero,i:3} -\epsilon'$. Then we define bin $B_h^{\nzero i} \eqdef \{w_{\nzero,i:2} | w_{\nzero,i:2} \in [h\cdot\lfloor2^{n(R_{\nzero,i:3}+R_{\nzero,i:1}-R_{\nzero,i}-\epsilon')}\rfloor , (h+1)\cdot\lfloor2^{n(R_{\nzero,i:3}+R_{\nzero,i:1}-R_{\nzero,i}-\epsilon')}\rfloor -1]\}$ for $h\in \{0,1,\cdots,\lfloor 2^{n R_{\nzero,i}}\rfloor -1\}$.
 \item${\bf v}^{(j)}_{j 1}(w_{j,\nzero:1})$ with $p^{(j)}({\bf v}_{j1})$, $w_{j,\nzero:1}  \in \{0,1,\cdots,\lfloor 2^{nR_{j,\nzero:1}} \rfloor-1\} \eqdef  {\cal S}_{j,\nzero:1}$, where
  \begin{align}
p^{(j)}({\bf v}_{j1}) =\left\{
              \begin{array}{ll}
                \frac{1}{\|A^{(j)}(V_{j1}) \|}, & {\bf v}_{j1}\in A^{(j)}(V_{j 1}) \\
                0, & \hbox{otherwise,}
              \end{array}
            \right.
\end{align}
and ${\bf v}^{(j)}_{j 2}(w_{\{j\},{\cal B}:1})$ with $p^{(j)}({\bf v}_{j2})$, $w_{\{j\},{\cal B}:1}  \in \{0,1,\cdots,\lfloor 2^{nR_{\{j\},{\cal B}:1}} \rfloor-1\} \eqdef  {\cal S}_{\{j\},{\cal B}:1}$, where
  \begin{align}
p^{(j)}({\bf v}_{j 2}) =\left\{
              \begin{array}{ll}
                \frac{1}{\|A^{(j)}(V_{j 2}) \|}, & {\bf v}_{j 2}\in A^{(j)}(V_{j 2}) \\
                0, & \hbox{otherwise.}
              \end{array}
            \right.
\end{align}
and $R_{j,\nzero:1} = \Delta_j I(V_{j1}^{(j)};Y_\nr^{(j)}) -4\epsilon$ , $R_{\{j\},{\cal B}:1} = R_{\{j\},{\cal B}} + \Delta_j I(V_{j2}^{(j)};Y_{{\cal I}_j^{\min}}^{(j)})$ \\$- \Delta_{m+1} I(Y_j^{(m+1)};{\hat Y}_j^{(m+1)}|Y_{{\cal I}_j^{\min}}^{(m+1)})-4\epsilon$. Then we define bin $B_h^{j1} \eqdef \{w_{j,\nzero:1} | w_{j,\nzero:1} \in [h\cdot\lfloor2^{n(\Delta_j I(V_{j1}^{(j)};Y_\nr^{(j)})-R_{j,\nzero}-\epsilon')}\rfloor , (h+1)\cdot\lfloor2^{n(\Delta_j I(V_{j1}^{(j)};Y_\nr^{(j)})-R_{j,\nzero}-\epsilon')}\rfloor -1]\}$ for $h\in \{0,1,\cdots,\lfloor 2^{n R_{j,\nzero}}\rfloor -1\}$. Similarly, $B_h^{j2} \eqdef \{w_{\{j\},{\cal B}:1} | w_{\{j\},{\cal B}:1} \in [h\cdot\lfloor2^{n(\Delta_j I(V_{j2}^{(j)};Y_{{\cal I}_j^{\min}}^{(j)}) - \Delta_{m+1} I(Y_j^{(m+1)};{\hat Y}_j^{(m+1)}|Y_{{\cal I}_j^{\min}}^{(m+1)})-\epsilon')}\rfloor , (h+1)\cdot\lfloor2^{n(\Delta_j I(V_{j2}^{(j)};Y_{{\cal I}_j^{\min}}^{(j)}) - \Delta_{m+1} I(Y_j^{(m+1)};{\hat Y}_j^{(m+1)}|Y_{{\cal I}_j^{\min}}^{(m+1)})-\epsilon')}\rfloor -1]\}$ for $h\in \{0,1,\cdots,\lfloor 2^{n R_{\{j\},{\cal B}}}\rfloor -1\}$.
\end{itemize}

\item Phase $m+1$: Generate random $(n\cdot\Delta_{m+1,n})$-length sequences
\begin{itemize}
  \item ${\bf u}^{(m+1)}_i(w_{\nr_i})$ with $p^{(m+1)}({\bf u}_i)$, $w_{\nr_i}  \in \{0,1,\cdots,\lfloor 2^{nR_{\nr_i}} \rfloor-1\} \eqdef  {\cal S}_{\nr_i}$, where
  \begin{align}
p^{(m+1)}({\bf u}_i) =\left\{
              \begin{array}{ll}
                \frac{1}{\|A^{(m+1)}(U_i) \|}, & {\bf u}_i\in A^{(m+1)}(U_i) \\
                0, & \hbox{otherwise.}
              \end{array}
            \right.
\end{align}
and $R_{\nr_i} = \max\{R_{\nzero,i:1},R_{i,\nzero}\} + (\Delta_{m+1} I(U_i^{(m+1)};Y_i^{(m+1)},{\hat Y}_{{\cal J}_i}^{(m+1)}) -4\epsilon - R_{\nzero,i:1})$.
\end{itemize}
and define bin  $C_h^i \eqdef \{w_{\nr_i} | w_{\nr_i} \in [h\cdot\lfloor2^{n(\Delta_{m+1} I(U_i^{(m+1)};Y_i^{(m+1)},{\hat Y}_{{\cal J}_i}^{(m+1)})-R_{\nzero,i:1}-4\epsilon)}\rfloor ,$ \\$(h+1)\cdot\lfloor2^{n(\Delta_{m+1} I(U_i^{(m+1)};Y_i^{(m+1)},{\hat Y}_{{\cal J}_i}^{(m+1)})-R_{\nzero,i:1}-4\epsilon)}\rfloor -1]\}$ for $h\in \{0,1,\cdots,\lfloor 2^{n\max\{R_{\nzero,i:1},R_{i,\nzero}\}}\rfloor -1\}$.
\end{enumerate}

{\em Encoding: } In slot $k$,
\begin{enumerate}
\item To transmit $(w_{\nzero,1|(k)},\cdots,w_{\nzero,m|(k)})$, node $\nzero$ picks $(w_{\nzero,1:2},\cdots,w_{\nzero,m:2})\in \bigotimes_{i=1}^m B^{\nzero i}_{w_{\nzero,i|(k)}}$ such that \\${\bf v}_{\nzero S}^{(j)}(w_{\{\nzero\} ,S:2})\in A^{(j)}(V_{\nzero S})$, $\forall j\in[1,m],~S\subseteq {\cal B},~|S|>1$. Such a $(w_{\nzero,1:2},\cdots,w_{\nzero,m:2})$ exists with high probability if
\begin{align}
\sum_{i\in S} \left(R_{\nzero,i} -R_{\nzero,i:1}\right) <& \sum_{i\in S} \left(\sum_{j\in {\cal J}_i}\Delta_j I(V_{\nzero i}^{(j)};Y_i^{(j)}|V_{j2}^{(i)}) - \Delta_j I(V_{\nzero i}^{(j)};V_{\nzero S(i)}^{(j)})  + \right. \nonumber\\
&~~~~~~~\left.\sum_{j\not\in {\cal J}_i} \Delta_j I(V_{\nzero i}^{(j)};Y_i^{(j)})- \Delta_j I(V_{\nzero i}^{(j)};V_{\nzero S(i)}^{(j)}) \right)  -\delta(\epsilon)\label{eq:PMABC:DF:GC:13}
\end{align}
for $S\subseteq {\cal B},~|S|>1$ from Lemma \ref{lemma:gbc}.  Then node $\nzero$ finds a ${\bf x}_\nzero^{(i)}$ jointly typical with \\$({\bf v}_{\nzero 1}^{(i)}(w_{\nzero,1:2}),\cdots,{\bf v}_{\nzero m}^{(i)}(w_{\nzero,m:2}))$ and designate it as the codeword corresponding to $(w_{\nzero,1|(k)},\cdots,w_{\nzero,m|(k)})$ for all $i\in[1,m]$. Node $\nzero$ sends ${\bf x}^{(i)}_\nzero(w_{\nzero,1|(k)},\cdots,w_{\nzero,m|(k)})$ during phase $i$.

\item Node $j (\in {\cal B})$ compresses ${\bf y}_j^{(m+1)}$ to $
\hat{\bf y}_j^{(m+1)}(w_{\{j\},{\cal B}|(k)})$ if there exists a $w_{\{j\},{\cal B}|(k)}$ such that \\$({\bf
y}_j^{(m+1)} , \hat{\bf y}_j^{(m+1)}(w_{\{j\},{\cal B}|(k)}))$ is in the jointly typical set at the end of phase $m+1$ in the slot $k$. There exists such a $w_{\{j\},{\cal B}|(k)}$
with high probability if
\begin{align}
R_{\{j\},{\cal B}} = \Delta_{m+1,n} I(Y_j^{(m+1)};\hat{Y}_j^{(m+1)}) + \epsilon \label{eq:PMABC:DF:GC:14}
\end{align}
and $n$ is sufficiently large.

To transmit $(w_{j,\nzero|(k)},w_{\{j\},{\cal B}|(k-1)})$, node $j$ picks  $(w_{j,\nzero:1},w_{\{j\},{\cal B}:1})$ such that \\$({\bf v}_{j 1 }^{(j)}(w_{j,\nzero :1}),{\bf v}_{j2}^{(j)}(w_{\{j\},{\cal B}:1}))\in A^{(j)}(V_{j1} V_{j2})$ and $(w_{j,\nzero:1},w_{\{j\},{\cal B}:1})\in B^{j1}_{w_{j,\nzero|(k)}} \times B^{j2}_{w_{\{j\},{\cal B}|(k-1)}}$. Such a $(w_{j,\nzero:1},w_{\{j\},{\cal B}:1})$ exists with high probability if
\begin{align}
R_{j,\nzero:1} + \Delta_{m+1} I(Y_j^{(m+1)};{\hat Y}_j^{(m+1)}|Y_{{\cal I}_j^{\min}}^{(m+1)}) <& \Delta_j I(V_{j1}^{(j)};Y_\nr^{(j)}) + \Delta_j I(V_{j2}^{(j)};Y_{{\cal I}_j^{\min}}^{(j)}) - \nonumber\\
&\Delta_j I(V_{j1}^{(j)};V_{j2}^{(j)}) -\delta(\epsilon)\label{eq:PMABC:DF:GC:15}
\end{align}
from Lemma \ref{lemma:gbc}.  Then node $j$ finds a ${\bf x}_j^{(j)}$ jointly typical with $({\bf v}_{j1 }^{(j)}(w_{j,\nzero:1}),{\bf v}_{j2}^{(j)}(w_{\{j\},{\cal B}:1}))$ and designate it as the codeword corresponding to $(w_{j,\nzero|(k)},w_{\{j\},{\cal B}|(k-1)})$. Node $j$ sends ${\bf x}^{(j)}_j(w_{j,\nzero|(k)},w_{\{j\},{\cal B}|(k-1)})$ during phase $j$.

\item
Relay $\nr$ estimates $\hat{w}_{\nzero,i|(k)}$ and $\hat{w}_{i,\nzero|(k)}$ after phase $m$ using jointly typical decoding, then constructs $w_{i} =\hat{w}_{\nzero,i:1} \oplus \hat{w}_{i,\nzero|(k)}$. To transmit a pair of messages $(w_1,\cdots,w_m)$, pick a pair $(w_{\nr_1},\cdots,w_{\nr_m})\in \bigotimes_{i=1}^m C^i_{w_i}$ such that ${\bf u}_S^{(m+1)}(w_{\nr_S})\in A^{(m+1)}(U_S), \forall S\subseteq {\cal B},~|S|>1$. Such a $(w_{\nr_1},\cdots,w_{\nr_m})$ exists with high probability if
\begin{align}
\sum_{i\in S} R_{\nzero,i:1} &< \sum_{i\in S} \Delta_{m+1} I(U_i^{(m+1)};Y_i^{(m+1)},{\hat Y}_{{\cal J}_i}^{(m+1)}) - \Delta_{m+1} I(U_i^{(m+1)};U_{S(i)}^{(m+1)})  -|S|\epsilon -\delta(\epsilon) \label{eq:PMABC:DF:GC:16}
\end{align}
for $S\subseteq {\cal B},~|S|>1$ from Lemma \ref{lemma:gbc}. Then the relay finds a ${\bf x}_\nr^{(m+1)}$ jointly typical with \\$({\bf u}_1^{(m+1)}(w_{\nr_1}),\cdots,{\bf u}_m^{(m+1)}(w_{\nr_m}))$ and designate it as the codeword corresponding to $(w_1,\cdots,w_m)$. Relay $\nr$ sends ${\bf x}^{(m+1)}_\nr(w_1,\cdots,w_m)$ during phase $m+1$.
\end{enumerate}

{\em Decoding: } For all $i\in [1,m]$,
\begin{enumerate}
\item Node $\nzero$ estimates $\tilde{w}_{i,\nzero|(k)}$ after phase $m+1$ using jointly typical decoding. Since $w_{i} = w_{\nzero,i:1} \oplus w_{i,\nzero|(k)}$ and $\nzero$ knows $w_{\nzero,i|(k)}$, node $\nzero$ can reduce the number of possible $w_{i|(k)}$.

\item Node $i$ decodes $\tilde{w}_{\{h\},{\cal B}|(k-1)}$ after phase $h$ if there exists a unique $\tilde{w}_{\{h\},{\cal B}|(k-1)}$ such that \\$({\bf v}_{h2}^{(h)}(\tilde{w}_{\{h\},{\cal B}|(k-1)}),{\bf y}_i^{(h)})\in A^{(h)}(V_{h2} Y_i)$ and $({\bf y}_i^{(m+1)}, {\hat{\bf y}}_h^{(m+1)}(\tilde{w}_{\{h\},{\cal B}|(k-1)})\in A^{(m+1)}(Y_i {\hat Y}_h)$ for $h \in {\cal J}_i$. Then, node $i$ estimates ${\tilde w}_{\nr_i}$ of $k-1$ slot after phase $m+1$ using jointly typical sequences $({\bf u}_\nr^{(m+1)},{\bf y}_i^{(m+1)}, {\bf y}_{{\cal J}_i}^{(m+1)})$. Similar to the case of node $\nzero$, node $i$ can reduce the cardinality of $w_{\nr_i}$ to $2^{n(\Delta_{m+1} I(U_i^{(m+1)};Y_i^{(m+1)},{\hat Y}_{{\cal J}_i}^{(m+1)})-4\epsilon)}$ and node $i$ decodes ${\tilde w}_{\nzero,i:1}$ from the bin index of ${\tilde w}_{\nr_i}$. Then, node $i$ estimates ${\tilde w}_{\nzero,i:2}$ using jointly typical decoding of the sequences $({\bf v}_{\nzero i}^{(j)}({\tilde w}_{\nzero,i:2}),{\bf v}_{j2}^{(j)}(w_{\{j\},{\cal B}:1}),{\bf y}_i^{(j)})$ for all $j\in {\cal J}_i$ and $({\bf v}_{\nzero i}^{(j)}({\tilde w}_{\nzero,i:2}),{\bf y}_i^{(j)})$ for all $j\not\in {\cal J}_i$  . Since node $i$ knows the bin index $s_{{\tilde w}_{\nzero,i}}({\tilde w}_{\nzero,i|(k-1)})$ as ${\tilde w}_{\nzero,i:1}$, it can reduce the cardinality of $w_{\nzero,i:2}$ to $2^{n(\sum_{j\in {\cal J}_i}\Delta_j I(V_{\nzero i}^{(j)};Y_i^{(j)}|V_{j2}^{(j)}) + \sum_{j\not\in {\cal J}_i}\Delta_j I(V_{\nzero i}^{(j)};Y_i^{(j)})- \epsilon')}$. After decoding ${\tilde w}_{\nzero,i:2}$, node $i$ finally decodes ${\tilde w}_{\nzero,i|(k-1)}$ from the bin index of ${\tilde w}_{\nzero,i:2}$.
\end{enumerate}

{\em Error analysis: }
\begin{align}
P[E] &= P[E_{enc} \cup E_{dec}] \leq P[E_{enc}] + P[E_{dec}]\\
     &\leq \epsilon + \sum_{i=1}^m P[E_{\{\nzero\},\{i\}}] + P[E_{\{i\},\{\nzero\}}]
\end{align}
where $E$ is the entire error event, $E_{enc}$ is the set of encoding error events, and $E_{dec}$ is the set of decoding error events. $P[E_{enc}]$ is upper bounded by sufficiently small number $\epsilon$ from \eqref{eq:PMABC:DF:GC:13}, \eqref{eq:PMABC:DF:GC:14}, \eqref{eq:PMABC:DF:GC:15} and \eqref{eq:PMABC:DF:GC:16}. $E_{dec}$ can be separated by individual decoding error events in each link.
Then
$\forall j \in [1,m]$,
\begin{align}
  P[E_{\{\nzero\},\{j\}}]  \leq& P[(\cup_{i=1}^m E_{\{\nzero,i\},\{\nr\}}^{(i)}) \cup (\cup_{h\in{\cal J}_j} E_{\{h\},\{j\}}^{(h)})\cup E_{\{\nr\},\{j\}}^{(m+1)} \cup  E_{\{\nzero\},\{j\}}^{(m+1)}]\\
   \leq & P[\cup_{i=1}^m E_{\{\nzero,i\},\{\nr\}}^{(i)}] + P[\cup_{h\in{\cal J}_j} E_{\{h\},\{j\}}^{(h)} | \cap_{i=1}^m {\bar E}_{\{\nzero,i\},\{\nr\}}^{(i)}] + \nonumber\\
  &P[ E_{\{\nr\},\{j\}}^{(m+1)}|(\cap_{h\in{\cal J}_j} {\bar E}_{\{h\},\{j\}}^{(h)}) \cap (\cap_{i=1}^m {\bar E}_{\{\nzero,i\},\{\nr\}}^{(i)})] + \nonumber\\
  &P[ E_{\{\nzero\},\{j\}}^{(m+1)}|(\cap_{h\in{\cal J}_j} {\bar E}_{\{h\},\{j\}}^{(h)}) \cap (\cap_{i=1}^m {\bar E}_{\{\nzero,i\},\{\nr\}}^{(i)}) \cap \bar{E}_{\{\nr\},\{j\}}^{(m+1)}]\\
  P[E_{\{j\},\{\nzero\}}]  \leq& P[(\cup_{i=1}^m E_{\{\nzero,i\},\{\nr\}}^{(i)}) \cup E_{\{\nr\},\{\nzero\}}^{(m+1)}]\\
   \leq & P[\cup_{i=1}^m E_{\{\nzero,i\},\{\nr\}}^{(i)}] + P[E_{\{\nr\},\{\nzero\}}^{(m+1)} | \cap_{i=1}^m {\bar E}_{\{\nzero,i\},\{\nr\}}^{(i)}]
\end{align}

Now we will show that $P[\cup_{i=1}^m E_{\{\nzero,i\},\{\nr\}}^{(i)}]$, $ P[\cup_{h\in{\cal J}_j} E_{\{h\},\{j\}}^{(h)} | \cap_{i=1}^m {\bar E}_{\{\nzero,i\},\{\nr\}}^{(i)}]$, $P[ E_{\{\nr\},\{j\}}^{(m+1)}|(\cap_{h\in{\cal J}_j} {\bar E}_{\{h\},\{j\}}^{(h)}) \cap (\cap_{i=1}^m {\bar E}_{\{\nzero,i\},\{\nr\}}^{(i)})]$, $P[ E_{\{\nzero\},\{j\}}^{(m+1)}|(\cap_{h\in{\cal J}_j} {\bar E}_{\{h\},\{j\}}^{(h)}) \cap (\cap_{i=1}^m {\bar E}_{\{\nzero,i\},\{\nr\}}^{(i)}) \cap \bar{E}_{\{\nr\},\{j\}}^{(m+1)}]$ and $P[E_{\{\nr\},\{\nzero\}}^{(m+1)} | \cap_{i=1}^m {\bar E}_{\{\nzero,i\},\{\nr\}}^{(i)}]$ tend to zero as $n\rightarrow \infty$. For the convenience of analysis we define $B^{\nzero i}_R(w_{\nzero,i:1})= \cup_{w_{\nzero,i} \in {\cal S}_{\nzero,i}(w_{\nzero,i:1})} B^{\nzero i}_{w_{\nzero,i}}$ and $C^i_R(w_{i,\nzero}) = \cup_{w_{\nzero,i:1} \in {\cal S}_{\nzero,i:1}} C^i_{w_{\nzero,i:1}\oplus w_{i,\nzero}}$. Similarly, $B^{i2}_R(j)= \cup_{w_{\{i\},{\cal B}} \in T_i^j} B^{i2}_{w_{\{i\},{\cal B}}}$ such that \\$T_i^j = \{w_{\{i\},{\cal B}} | ({\bf y}_j^{(m+1)},{\hat{\bf y}}_i^{(m+1)}(w_{\{i\},{\cal B}}))\in A^{(m+1)}(Y_j {\hat Y}_i) \}$.
\begin{align}
 P[\cup_{i=1}^m &E_{\{\nzero,i\},\{\nr\}}^{(i)}] \nonumber\\
 \leq &\sum_{i=1}^m \left(P[\bar{D}^{(i)}({\bf v}_{\nzero {\cal B}}(w_{\{\nzero\},{\cal B}:2}), {\bf v}_{i1}(w_{i,\nzero:1}),{\bf y}_\nr)] +  P[\cup_{{\tilde w}_{i,\nzero:1}\neq w_{i,\nzero:1}} D^{(i)}({\bf v}_{i1}({\tilde w}_{i,\nzero:1}),{\bf y}_\nr)]\right) + \nonumber\\
 & \sum_{S\subseteq {\cal B}} P\left[\cup_{\tilde{w}_{\{\nzero\},S}} (\cap_{i=1}^m D^{(i)}({\bf v}_{\nzero S}( \tilde{w}_{\{\nzero\},S:2}),{\bf v}_{\nzero {\bar S}}(w_{\{\nzero\},{\bar S}:2}),{\bf v}_{i1}(w_{i,\nzero:1}),{\bf y}_\nr)) \right]\\
 &\leq m\cdot\epsilon + \sum_{i=1}^m 2^{n(R_{i,\nzero:1} - \Delta_{i,n}I(V_{i1}^{(i)};Y_\nr^{(i)})+3\epsilon)} +
 \sum_{S \subseteq {\cal B}} 2^{n(R_{\{\nzero\},S} - \sum_{i=1}^m \Delta_{i,n}I(V_{\nzero S} ^{(i)};Y_\nr^{(i)},V_{\nzero {\bar S}}^{(i)} |V_{i1}^{(i)})+\epsilon')}\label{eq:PMABC:DF:GC:8}
\end{align}
The total cardinality of the case (${\tilde w}_{\{\nzero\},S:2} \neq w_{\{\nzero\},S:2})$ is bounded by $2^{nR_{\{\nzero\},S}}$ since ${\tilde w}_{\{\nzero\},S:2}$ is uniquely specified with ${\tilde w}_{\{\nzero\},S}$ if $w_{\{\nzero\},{\bar S}:2}$ is given.
\begin{align}
P[\cup_{h\in{\cal J}_j} E_{\{h\},\{j\}}^{(h)} | \cap_{i=1}^m {\bar E}_{\{\nzero,i\},\{\nr\}}^{(i)}] \leq \sum_{h\in{\cal J}_j} P[E_{\{h\},\{j\}}^{(h)} | \cap_{i=1}^m {\bar E}_{\{\nzero,i\},\{\nr\}}^{(i)}]
\end{align}
Also,
\begin{align}
P[E_{\{h\},\{j\}}^{(h)}| \cap_{i=1}^m {\bar E}_{\{\nzero,i\},\{\nr\}}^{(i)}] &\leq P[{\bar D}^{(h)}({\bf v}_{h2}(w_{h,j}),{\bf y}_j)] + \nonumber\\
&~~~P[\cup_{\tilde{w}_{\{h\},{\cal B}:1}\in B^{h2}_R(j)} D^{(h)}({\bf v}_{h2}(\tilde{w}_{\{h\},{\cal B}:1}),{\bf y}_j) ]\\
&\leq \epsilon + |B_R^{h2}(j)|\cdot 2^{-n(\Delta_{h,n}I(V_{h2}^{(h)};Y_{j}^{(h)})-3\epsilon)}\label{eq:PMABC:DF:GC:9}
\end{align}
and $|B_R^{h2}(j)| = 2^{n(\Delta_{m+1} I(Y_h^{(m+1)};{\hat Y}_h^{(m+1)}|Y_j^{(m+1)}) + \Delta_h I(V_{h2}^{(h)};Y_{{\cal I}_h^{\min}}^{(h)}) - \Delta_{m+1} I(Y_h^{(m+1)};{\hat Y}_h^{(m+1)}|Y_{{\cal I}_h^{\min}}^{(m+1)})-4\epsilon)}$.
\begin{align}
 P[E_{\{\nr\},\{j\}}^{(m+1)}|&(\cap_{h\in{\cal J}_j} {\bar E}_{\{h\},\{j\}}^{(h)}) \cap (\cap_{i=1}^m {\bar E}_{\{\nzero,i\},\{\nr\}}^{(i)})] \nonumber\\
 \leq& P[\bar{D}^{(m+1)}({\bf u}_j (w_{\nr_j}),{\bf y}_j)] + P[\cup_{\tilde{w}_{\nr_j}\in C^j_R(w_{j,\nzero})} D^{(m+1)}({\bf u}_j(\tilde{w}_{\nr_j}),{\bf y}_j,{\hat{\bf y}}_{{\cal J}_j} )]\\
 \leq& \epsilon + |C^j_R(w_{j,\nzero})|\cdot 2^{-n(\Delta_{m+1,n} I(U_j^{(m+1)};Y_j^{(m+1)},{\hat Y}_{{\cal J}_j}^{(m+1)}) - 3\epsilon)}
\label{eq:PMABC:DF:GC:10}
\end{align}
\begin{align}
P[ &E_{\{\nzero\},\{j\}}^{(m+1)}|(\cap_{h\in{\cal J}_j} {\bar E}_{\{h\},\{j\}}^{(h)}) \cap (\cap_{i=1}^m {\bar E}_{\{\nzero,i\},\{\nr\}}^{(i)}) \cap \bar{E}_{\{\nr\},\{j\}}^{(m+1)}] \nonumber\\
\leq& \sum_{i\neq j} P[\bar{D}^{(i)}({\bf v}_{\nzero j} (w_{\nzero,j:2}),{\bf y}_j)] +\nonumber\\
 & P\left[\cup_{\tilde{w}_{\nzero,j:2}\in B^j_R(w_{\nzero,i:1})} \left(\cap_{i\in {\cal J}_j} D^{(i)}({\bf v}_{\nzero j}(\tilde{w}_{\nzero,j:2}),{\bf v}_{i2}(w_{\{i\},{\cal B}:1}),{\bf y}_j) \cap_{i\not\in {\cal J}_j} D^{(i)}({\bf v}_{\nzero j}(\tilde{w}_{\nzero,j:2}),{\bf y}_j) \right)\right]\\
 \leq& (m-1)\epsilon + |B^i_R(w_{\nzero,i:1})|\cdot 2^{-n(\sum_{i\in {\cal J}_j} \Delta_{i,n} I(V_{\nzero j}^{(i)};Y_j^{(i)}|V_{i2}^{(i)}) +  \sum_{i\not \in {\cal J}_j} \Delta_{i,n} I(V_{\nzero j}^{(i)};Y_j^{(i)})- \epsilon'')}\label{eq:PMABC:DF:GC:11}
\end{align}
\begin{align}
 P[E_{\{\nr\},\{\nzero\}}^{(m+1)} | \cap_{j=1}^m {\bar E}_{\{\nzero,j\},\{\nr\}}^{(j)}] \leq &P[\bar{D}^{(m+1)}({\bf u}_{\cal B}(w_{\nr_{\cal B}}),{\bf y}_\nzero)] + \nonumber\\
 & \sum_{S\subseteq{\cal B}} P[\cup_{\tilde{w}_{S,\{\nzero\}} \not= w_{S,\{\nzero\}}} D^{(m+1)}({\bf u}_S(\tilde{w}_{\nr_S}),{\bf u}_{\bar S}(w_{\nr_{\bar S}}),{\bf y}_\nzero)]\\
  \leq& \epsilon + \sum_{S\subseteq {\cal B}} 2^{nR_{S,\{\nzero\}}} 2^{-n\cdot\Delta_{m+1,n}(I(U_S^{(m+1)};Y_\nzero^{(m+1)},U_{\bar S}^{(m+1)})-\epsilon')}\label{eq:PMABC:DF:GC:12}
\end{align}
The total cardinality of the case (${\tilde w}_{\nr_S} \neq w_{\nr_S})$ is bounded by $2^{nR_{S,\{\nzero\}}}$ since ${\tilde w}_{\nr_S}$ is uniquely specified with ${\tilde w}_{S,\{\nzero\}}$ if $w_{{\bar S},\{\nzero\}}$ is given.

Since $\epsilon > 0$ is arbitrary, with the conditions of Theorem \ref{theorem:PMABC-NRC} and the AEP property, we can make the right hand sides of \eqref{eq:PMABC:DF:GC:8}, \eqref{eq:PMABC:DF:GC:9}, \eqref{eq:PMABC:DF:GC:10}, \eqref{eq:PMABC:DF:GC:11} and \eqref{eq:PMABC:DF:GC:12} tend
to 0 as $n \rightarrow \infty$.  By
Carath\'{e}odory theorem in \cite{Hiriart:2001}, it is sufficient to
restrict $|{\cal Q}| \leq 2^{m+1}+ m^2 +m+2$ since $2^m$ inequality is from \eqref{eq:PMABC:DF:GC:1}, $m$ inequalities are from \eqref{eq:PMABC:DF:GC:2}, $2^m$ inequalities are from \eqref{eq:PMABC:DF:GC:4}, $m$ inequalities are from \eqref{eq:PMABC:DF:GC:6} and $m(m-1)$ inequalities are from \eqref{eq:PMABC:DF:GC:7}.
\end{proof}

%%%%%%%%%%%%%%%%%%%%%%%%%%%%%%%%%%%%%%%%%%%%%%%%%%%%%%

\section{Proof of Theorem \ref{theorem:FTDBC-NR}}

%%%%%%%%%%%%%%%%%%%%%%%%%%%%%%%%%%%%%%%%%%%%%%%%%%%%%%
\label{app:FTDBC-NR}

\begin{proof}
{\em Random code generation: } For all $i\in [1,m]$, first we generate a partition of ${\cal S}_{\nzero,i}$  randomly by independently assigning every index $w_{\nzero,i} \in {\cal S}_{\nzero,i}$ to a set ${\cal S}_{\nzero,i}(j)$, with a uniform distribution over the indices $j \in
\{0, \ldots, \lfloor 2^{nR_{\nzero,i:1}} \rfloor - 1\} \eqdef {\cal S}_{\nzero,i:1}$. We denote by $s_{\nzero,i}(w_{\nzero,i})$
the index $j$ of ${\cal S}_{\nzero,i}(j)$ to which $w_{\nzero,i}$  belongs. Similarly, we generate a partition of ${\cal S}_{i,\nzero}$ with the cardinality $2^{nR_{i,\nzero:1}}$.
\begin{enumerate}
\item Phase 1: Generate random $(n\cdot\Delta_{1,n})$-length sequences
\begin{itemize}
  \item ${\bf v}^\pa_{\nzero i}(w_{\nzero,i:2})$ with $p^\pa({\bf v}_{\nzero i})$, $w_{\nzero,i:2}  \in \{0,1,\cdots,\lfloor 2^{nR_{\nzero,i:2}} \rfloor-1\} \eqdef  {\cal S}_{\nzero,i:2}$, where
  \begin{align}
p^\pa({\bf v}_{\nzero i}) =\left\{
              \begin{array}{ll}
                \frac{1}{\|A^\pa(V_{\nzero i}) \|}, & {\bf v}_{\nzero i}\in A^\pa(V_{\nzero i}) \\
                0, & \hbox{otherwise.}
              \end{array}
            \right.
\end{align}
\end{itemize}
and $R_{\nzero,i:2} = R_{\nzero,i:1} + \Delta_1 I(V_{\nzero i}^\pa;Y_i^\pa) -4\epsilon$. Then we define bin  $B_j^{i} \eqdef \{w_{\nzero,i:2} | w_{\nzero,i:2} \in [j\cdot\lfloor2^{n(\Delta_1 I(V_{\nzero i}^\pa;Y_i^\pa)+R_{\nzero,i:1}-R_{\nzero,i}-4\epsilon)}\rfloor, (j+1)\cdot\lfloor2^{n(\Delta_1 I(V_{\nzero i}^\pa;Y_i^\pa)+R_{\nzero,i:1}-R_{\nzero,i}-4\epsilon)}\rfloor -1]\}$ for $j\in \{0,1,\cdots,\lfloor 2^{n R_{\nzero,i}}\rfloor -1\}$.

\item Phase $i+1$: Generate random $(n\cdot\Delta_{i+1,n})$-length sequences
\begin{itemize}
  \item ${\bf x}^{(i+1)}_i(w_{i,\nzero})$ i.i.d. with $p^{(i+1)}(x_i)$, $w_{i,\nzero} \in {\cal S}_{i,\nzero}$
\end{itemize}
\item Phase $m+2$: Generate random $(n\cdot\Delta_{m+2,n})$-length sequences
\begin{itemize}
  \item ${\bf u}^{(m+2)}_i(w_{\nr_i})$ with $p^{(m+2)}({\bf u}_i)$, $w_{\nr_i}  \in \{0,1,\cdots,\lfloor 2^{nR_{\nr_i}} \rfloor-1\} \eqdef  {\cal S}_{\nr_i}$, where
  \begin{align}
p^{(m+2)}({\bf u}_i) =\left\{
              \begin{array}{ll}
                \frac{1}{\|A^{(m+2)}(U_i) \|}, & {\bf u}_i\in A^{(m+2)}(U_i) \\
                0, & \hbox{otherwise.}
              \end{array}
            \right.
\end{align}
and $R_{\nr_i} = \max\{R_{\nzero,i:1},R_{i,\nzero:1}\} + (\Delta_{m+2} I(U_i^{(m+2)};Y_i^{(m+2)}) -4\epsilon - R_{\nzero,i:1})$.
\end{itemize}
and define bin  $C_j^i \eqdef \{w_{\nr_i} | w_{\nr_i} \in [j\cdot\lfloor2^{n(\Delta_{m+2} I(U_i^{(m+2)};Y_i^{(m+2)})-R_{\nzero,i:1}-4\epsilon)}\rfloor ,$ \\$(j+1)\cdot\lfloor2^{n(\Delta_{m+2} I(U_i^{(m+2)};Y_i^{(m+2)})-R_{\nzero,i:1}-4\epsilon)}\rfloor -1]\}$ for $j\in \{0,1,\cdots,\lfloor 2^{n\max\{R_{\nzero,i:1},R_{i,\nzero:1}\}}\rfloor -1\}$.
\end{enumerate}

{\em Encoding: }  For all $i\in [1,m]$,
\begin{enumerate}
\item During phase 1, to transmit $(w_{\nzero,1},\cdots,w_{\nzero,m})$, node $\nzero$ picks  $(w_{\nzero,1:2},\cdots,w_{\nzero,m:2})\in \bigotimes_{i=1}^m B^i_{w_{\nzero,i}}$ such that ${\bf v}_{\nzero S}^\pa(w_{\{\nzero\} ,S:2})\in A^\pa(V_{\nzero S})$, $\forall S\in {\cal B}~,~ |S|>1$. Such a $(w_{\nzero,1:2},\cdots,w_{\nzero,m:2})$ exists with high probability if
\begin{align}
\sum_{i\in S} \left(R_{\nzero,i} -R_{\nzero,i:1}\right) &< \sum_{i\in S} \Delta_1 I(V_{\nzero i}^\pa;Y_i^\pa) - \Delta_1 I(V_{\nzero i}^\pa;V_{\nzero S(i)}^\pa)  -|S|\epsilon -\delta(\epsilon) \label{eq:FTDBC:GC:12}
\end{align}
for $S\subseteq {\cal B}~,~ |S|>1$ from Lemma \ref{lemma:gbc}.  Then node $\nzero$ finds a ${\bf x}_\nzero^\pa$ jointly typical with \\$({\bf v}_{\nzero 1}^\pa(w_{\nzero,1:2}),\cdots,{\bf v}_{\nzero m}^\pa(w_{\nzero,m:2}))$ and designate it as the codeword corresponding to $(w_{\nzero,1},\cdots,w_{\nzero,m})$. Node $\nzero$ sends ${\bf x}^\pa_\nzero(w_{\nzero,1},\cdots,w_{\nzero,m})$ during phase 1.

\item
During phase $i+1$ encoder of terminal node $i$ sends the codeword ${\bf x}^{(i+1)}_i(w_{i,\nzero})$.

\item
Relay $\nr$ estimates $\hat{w}_{\nzero,i}$ and $\hat{w}_{i,\nzero}$ after phase $m+1$ using jointly typical decoding, then constructs $w_{i} =\hat{w}_{\nzero,i:1} \oplus \hat{w}_{i,\nzero:1}$. To transmit a pair of messages $(w_1,\cdots,w_m)$, pick a pair $(w_{\nr_1},\cdots,w_{\nr_m})\in \bigotimes_{i=1}^m C^i_{w_i}$ such that ${\bf u}_S^{(m+2)}(w_{\nr_S})\in A^{(m+2)}(U_S)$, $\forall S\in {\cal B},~ |S|>1$. Such a $(w_{\nr_1},\cdots,w_{\nr_m})$ exists with high probability if
\begin{align}
\sum_{i\in S} R_{\nzero,i:1} &< \sum_{i\in S} \Delta_{m+2} I(U_i^{(m+2)};Y_i^{(m+2)}) - \Delta_{m+2} I(U_i^{(m+2)};U_{S(i)}^{(m+2)})  -|S|\epsilon -\delta(\epsilon) \label{eq:FTDBC:GC:13}
\end{align}
for $S\subseteq {\cal B},~|S|>1$. from Lemma \ref{lemma:gbc}. Then the relay finds a ${\bf x}_\nr^{(m+2)}$ jointly typical with \\$({\bf u}_1^{(m+2)}(w_{\nr_1}),\cdots,{\bf u}_m^{(m+2)}(w_{\nr_m}))$ and designate it as the codeword corresponding to $(w_1,\cdots,w_m)$. Relay $\nr$ sends ${\bf x}^{(m+2)}_\nr(w_1,\cdots,w_m)$ during phase $m+2$.
\end{enumerate}

{\em Decoding: } For all $i\in [1,m]$,
\begin{enumerate}
\item Node $\nzero$ estimates $\tilde{w}_{i,\nzero:1}$ after phase $m+2$ using jointly typical decoding. Since $w_{i} = w_{\nzero,i:1} \oplus w_{i,\nzero:1}$ and $\nzero$ knows $w_{\nzero,i}$, node $\nzero$ can reduce the number of possible $w_{i}$. Then node $\nzero$ decodes $\tilde{w}_{i,\nzero}$ if there exists a unique $\tilde{w}_{i,\nzero} \in {\cal S}_{i,\nzero}({\tilde w}_{i,\nzero:1})$ such that $({\bf x}_i^{(i+1)}({\tilde w}_{i,\nzero}),{\bf y}_\nzero^{(i+1)})\in A^{(i+1)}(X_i Y_\nzero)$.

\item Node $i$ estimates ${\tilde w}_{\nr_i}$ after phase $m+2$ using jointly typical decoding. Similar to the case of node $\nzero$, node $i$ can reduce the cardinality of $w_{\nr_i}$ to $2^{n(\Delta_{m+2} I(U_i^{(m+2)};Y_i^{(m+2)})-4\epsilon)}$ and node $i$ decodes ${\tilde w}_{\nzero,i:1}$ from the bin index of ${\tilde w}_{\nr_i}$. Then, node $i$ estimates ${\tilde w}_{\nzero,i:2}$ using jointly typical decoding of the sequence $({\bf v}_{\nzero i}^\pa({\tilde w}_{\nzero,i:2}),{\bf y}_i^\pa)$. Since node $i$ knows the bin index $s_{{\tilde w}_{\nzero,i}}({\tilde w}_{\nzero,i})$ as ${\tilde w}_{\nzero,i:1}$, it can reduce the cardinality of $w_{\nzero,i:2}$ to $2^{n(\Delta_1 I(V_{\nzero i}^\pa;Y_i^\pa) - 4 \epsilon)}$. After decoding ${\tilde w}_{\nzero,i:2}$, node $i$ finally decodes ${\tilde w}_{\nzero,i}$ from the bin index of ${\tilde w}_{\nzero,i:2}$.
\end{enumerate}

{\em Error analysis: }
\begin{align}
P[E] &= P[E_{enc} \cup E_{dec}] \leq P[E_{enc}] + P[E_{dec}]\\
     &\leq \epsilon + \sum_{i=1}^m P[E_{\{\nzero\},\{i\}}] + P[E_{\{i\},\{\nzero\}}]
\end{align}
where $E$ is the entire error event, $E_{enc}$ is the set of encoding error events, and $E_{dec}$ is the set of decoding error events. $P[E_{enc}]$ is upper bounded by sufficiently small number $\epsilon$ from \eqref{eq:FTDBC:GC:12} and \eqref{eq:FTDBC:GC:13}. $E_{dec}$ can be separated by individual decoding error events in each link.
Then $\forall i\in [1,m]$,
\begin{align}
  P[E_{\{\nzero\},\{i\}}] & \leq P[E_{\{\nzero\},\{\nr\}}^\pa \cup E_{\{\nr\},\{i\}}^{(m+2)} \cup E_{\{\nzero\},\{i\}}^{(m+2)}]\\
  & \leq P[E_{\{\nzero\},\{\nr\}}^\pa] + P[E_{\{\nr\},\{i\}}^{(m+2)} | \bar{E}_{\{\nzero\},\{\nr\}}^\pa ] +  P[E_{\{\nzero\},\{i\}}^{(m+2)} | \bar{E}_{\{\nzero\},\{\nr\}}^\pa \cap \bar{E}_{\{\nr\},\{i\}}^{(m+2)} ]\\
  P[E_{\{i\},\{\nzero\}}] & \leq P[E_{\{i\},\{\nr\}}^{(i+1)} \cup E_{\{\nr\},\{\nzero\}}^{(m+2)} \cup E_{\{i\},\{\nzero\}}^{(m+2)}]\\
  & \leq P[E_{\{i\},\{\nr\}}^{(i+1)}] + P[ E_{\{\nr\},\{\nzero\}}^{(m+2)} | \bar{E}_{\{i\},\{\nr\}}^{(i+1)}] + P[ E_{\{i\},\{\nzero\}}^{(m+2)} | \bar{E}_{\{i\},\{\nr\}}^{(i+1)}\cap \bar{E}_{\{\nr\},\{\nzero\}}^{(m+2)}]
\end{align}
Also, for the convenience of analysis we define $B^i_R(w_{\nzero,i:1})= \cup_{w_{\nzero,i} \in {\cal S}_{\nzero,i}(w_{\nzero,i:1})} B^i_{w_{\nzero,i}}$ and $C^i_R(w_{i,\nzero:1}) = \cup_{w_{\nzero,i:1} \in {\cal S}_{\nzero,i:1}} C^i_{w_{\nzero,i:1}\oplus w_{i,\nzero:1}}$. Then,
\begin{align}
 P[E_{\{\nzero\},\{\nr\}}^\pa] \leq&  P[\bar{D}^\pa({\bf v}_{\nzero {\cal B}} (w_{\{\nzero\},{\cal B}}),{\bf y}_\nr)] + \sum_{S\subseteq {\cal B}} P[\cup_{\tilde{w}_{\{\nzero\},S}} D^\pa({\bf v}_{\nzero S}(\tilde{w}_{\{\nzero\},S:2}),{\bf v}_{\nzero {\bar S}}(\tilde{w}_{\{\nzero\},{\bar S}:2}),{\bf y}_\nr)]\\
 \leq& \epsilon + \sum_{S\subseteq {\cal B}} 2^{n(R_{\{\nzero\},S}- \Delta_{1,n}(I(V_{\nzero S}^\pa;Y_\nr^\pa V_{\nzero {\bar S}})+\epsilon'))} \label{eq:FTDBC:GC:6}
\end{align}
The total cardinality of the case (${\tilde w}_{\{\nzero\},S:2} \neq w_{\{\nzero\},S:2})$ is bounded by $2^{nR_{\{\nzero\},S}}$ since ${\tilde w}_{\{\nzero\},S:2}$ is uniquely specified with ${\tilde w}_{\{\nzero\},S}$ if $w_{\{\nzero\},{\bar S}:2}$ is given.
\begin{align}
 P[E_{\{\nr\},\{i\}}^{(m+2)} | \bar{E}_{\{\nzero\},\{\nr\}}^\pa ] \leq& P[\bar{D}^{(m+2)}({\bf u}_i (w_{\nr_i}),{\bf y}_i)] + P[\cup_{\tilde{w}_{\nr_i}\in C^i_R(w_{i,\nzero:1})} D^{(m+2)}({\bf u}_i(\tilde{w}_{\nr_i}),{\bf y}_i)]\\
 \leq& \epsilon + |C^i_R(w_{i,\nzero:1})|\cdot 2^{-n(\Delta_{m+2,n} I(U_i^{(m+2)};Y_i^{(m+2)}) - 3\epsilon)}\label{eq:FTDBC:GC:7}\\
 P[E_{\{\nzero\},\{i\}}^{(m+2)} | \bar{E}_{\{\nzero\},\{\nr\}}^\pa \cap \bar{E}_{\{\nr\},\{i\}}^{(m+2)} ] \leq& P[\bar{D}^\pa({\bf v}_{\nzero i} (w_{\nzero,i:2}),{\bf y}_i)] +\nonumber\\
 & P[\cup_{\tilde{w}_{\nzero,i:2}\in B^i_R(w_{\nzero,i:1})} D^\pa({\bf v}_{\nzero i}(\tilde{w}_{\nzero,i:2}),{\bf y}_i)]\\
 \leq& \epsilon + |B^i_R(w_{\nzero,i:1})|\cdot 2^{-n(\Delta_{1,n} I(V_{\nzero i}^\pa;Y_i^\pa) - 3\epsilon)}\label{eq:FTDBC:GC:8}\\
 P[E_{\{i\},\{\nr\}}^{(i+1)}] \leq&  P[\bar{D}^{(i+1)}({\bf x}_i (w_{i,\nzero}),{\bf y}_\nr)] + P[\cup_{\tilde{w}_{i,\nzero} \not= w_{i,\nzero}} D^{(i+1)}({\bf x}_i(\tilde{w}_{i,\nzero}),{\bf y}_\nr)]\\
 \leq& \epsilon + 2^{n(R_{i,\nzero}- \Delta_{i+1,n}(I(X_i^{(i+1)};Y_\nr^{(i+1)})+3\epsilon))}\label{eq:FTDBC:GC:9} \\
 P[ E_{\{\nr\},\{\nzero\}}^{(m+2)} | \bar{E}_{\{i\},\{\nr\}}^{(i+1)}] \leq &P[\bar{D}^{(m+2)}({\bf u}_{\cal B}(w_{\nr_{\cal B}}),{\bf y}_\nzero)] + \nonumber\\
 & \sum_{S\subseteq{\cal B}} P[\cup_{\tilde{w}_{S,\{\nzero\}} \not= w_{S,\{\nzero\}}} D^{(m+2)}({\bf u}_S(\tilde{w}_{\nr_S}),{\bf u}_{\bar S}(w_{\nr_{\bar S}}),{\bf y}_\nzero)]\\
  \leq& \epsilon + \sum_{S\subseteq {\cal B}} 2^{nR_{S,\{\nzero\}:1}} 2^{-n\cdot\Delta_{m+2,n}(I(U_S^{(m+2)};Y_\nzero^{(m+2)},U_{\bar S}^{(m+2)})-\epsilon')}\label{eq:FTDBC:GC:10}
\end{align}
The total cardinality of the case (${\tilde w}_{\nr_S} \neq w_{\nr_S})$ is bounded by $2^{nR_{S,\{\nzero\}}:1}$ since ${\tilde w}_{\nr_S}$ is uniquely specified with ${\tilde w}_{S,\{\nzero\}:1}$ if $w_{{\bar S},\{\nzero\}:1}$ is given.
\begin{align}
 P[E_{\{i\},\{\nzero\}}^{(m+2)} | \bar{E}_{\{i\},\{\nr\}}^{(i+1)}\cap \bar{E}_{\{\nr\},\{\nzero\}}^{(m+2)}] \leq& P[\bar{D}^{(i+1)}({\bf x}_i (w_{i,\nzero}),{\bf y}_\nzero)] + \nonumber\\
 & P[\cup_{\tilde{w}_{i,\nzero} \not= w_{i,\nzero}} D^{(i+1)}({\bf x}_i(\tilde{w}_{i,\nzero}),{\bf y}_\nzero) \cap ({\tilde w}_{i,\nzero} \in S_{i,\nzero}(w_{i,\nzero:1}))]\\
 \leq & \epsilon +  2^{n(R_{i,\nzero}- \Delta_{i+1,n}(I(X_i^{(i+1)};Y_\nzero^{(i+1)})- R_{i,\nzero:1}+3\epsilon))}\label{eq:FTDBC:GC:11}
\end{align}

Since $\epsilon > 0$ is arbitrary, with the conditions of Theorem \ref{theorem:FTDBC-NR}, proper choices of $\{R_{\nzero,i:1}\}$ and $\{R_{i,\nzero:1}\}$ and the AEP property, we can make the right hand sides of \eqref{eq:FTDBC:GC:6}, \eqref{eq:FTDBC:GC:7}, \eqref{eq:FTDBC:GC:8}, \eqref{eq:FTDBC:GC:9}, \eqref{eq:FTDBC:GC:10} and \eqref{eq:FTDBC:GC:11} tend
to 0 as $n \rightarrow \infty$.
\end{proof}

%%%%%%%%%%%%%%%%%%%%%%%%%%%%%%%%%%%%%%%%%%%%%%%%%%%%%%

\section{Apply extended Marton's bound to the Gaussian channel}

%%%%%%%%%%%%%%%%%%%%%%%%%%%%%%%%%%%%%%%%%%%%%%%%%%%%%%
\label{app:evaluation}

For convenience of analysis, we define a function as follow:
\begin{align}
f ({\bf a},{\bf b},{\bf c}) &\eqdef \sum_{i=1}^m a_i b_i c_i
\end{align}
where ${\bf a},{\bf b},{\bf c}$ are vectors length $m$. Similar to Costa's setup in \cite{Costa:1983}, we will apply the broadcast scheme in the previous section to the Gaussian channel. We use the following relationship between input signals and auxiliary vairiables of transmitter $\na$:
\begin{align}
{\bf U}_\na & = {\bf{\Lambda}}_\na {\bf V}_\na\\
X_\na &= \sum_{i=1}^m V_{\na i}
\end{align}
where ${\bf U}_\na$ and ${\bf V}_\na$ are vectors length $m$ and ${\bf \Lambda}_\na \in \reals^{m\times m}$. Also $V_{\na i}'s$ follow the distributions $V_{\na i} \sim {\cal C N}(0,\beta_{\na i} P)$, where $(0\leq \beta_{\na i}\leq 1)$, $\sum_{i=1}^m \beta_{\na i} = 1$ and $V_{\na i}'s$ are independent. We define ${\bar \lambda}_{\na i}$ as the $i^{th}$ row vector of ${\bf \Lambda}_\na $, i.e., ${\bar \lambda}_{\na i} = (\lambda_{\na i1},\lambda_{\na i2},\cdots, \lambda_{\na im})$ and ${\bar \beta}_{\na} = (\beta_{\na 1},\cdots,\beta_{\na m})$.

Now we define some useful functions to evaluate mutual information terms in the Gaussian broadcast channel. We consider two scenarios; the single transmitter case and the double transmitter case. The latter one is only for the multiple access period of the PMABC protocols. We assume that there are two senders $\na$, $\nb$ and two receivers $\nc$ and $\nd$. Node $\na$ has $m$ independent messages and corresponding $m$ auxiliary random variables $U_{\na i}'s$. Similarly, Node $\nb$ constructs $U_{\nb i}'s$.

\begin{enumerate}
\item Single transmitter (when node $\na$ is only broadcasting)
\begin{align}
I(&U_{\na i};Y_\nc) = C\left(\frac{|h_{\na,\nc}|^2 P_\na f^2({\bar \lambda}_{\na i},{\bar 1},{\bar \beta}_\na)}{|h_{\na,\nc}|^2P_\na (f({\bar \lambda}_{\na i},{\bar \lambda}_{\na i},{\bar \beta}_\na)-f^2({\bar \lambda}_{\na i},{\bar 1},{\bar \beta}_\na))  +f({\bar \lambda}_{\na i},{\bar \lambda}_{\na i},{\bar \beta}_\na)} \right)\nonumber \\
&\eqdef C_{B}(P_\na,h_{\na,\nc},{\bar \lambda}_{\na i},{\bar \beta}_\na)\\
I(&U_{\na i};U_{\na j}) = C\left(\frac{f^2({\bar \lambda}_{\na i},{\bar \lambda}_{\na j},{\bar \beta}_\na)}{f({\bar \lambda}_{\na i},{\bar \lambda}_{\na i},{\bar \beta}_\na)f({\bar \lambda}_{\na j},{\bar \lambda}_{\na j},{\bar \beta}_\na)-f^2({\bar \lambda}_{\na i},{\bar \lambda}_{\na j} ,{\bar \beta}_\na)} \right)\nonumber \\
&\eqdef C_{BE}({\bar \lambda}_{\na i},{\bar \lambda}_{\na j},{\bar \beta}_\na)\\
I(&U_{\na i};U_{\na j},U_{\na k}) = C\left(\frac{K_{BE2}({\bar \lambda}_{\na i},{\bar \lambda}_{\na j},{\bar \lambda}_{\na k},{\bar \beta}_\na)}{K_{BE1}({\bar \lambda}_{\na i},{\bar \lambda}_{\na j},{\bar \lambda}_{\na k},{\bar \beta}_\na)-K_{BE2}({\bar \lambda}_{\na i},{\bar \lambda}_{\na j},{\bar \lambda}_{\na k},{\bar \beta}_\na)} \right)\nonumber \\
&\eqdef C_{BE2}({\bar \lambda}_{\na i},{\bar \lambda}_{\na j},{\bar \lambda}_{\na k},{\bar \beta}_\na)\\
I(&U_{\na i};Y_\nc,U_{\na j}) = \nonumber\\&C\left(\frac{|h_{\na,\nc}|^2 K_{C2}({\bar \lambda}_{\na i},{\bar \lambda}_{\na j},{\bar \beta}_\na) P_\na + K_{C4}({\bar \lambda}_{\na i},{\bar \lambda}_{\na j},{\bar \beta}_\na)}{|h_{\na,\nc}|^2 (K_{C1}({\bar \lambda}_{\na i},{\bar \lambda}_{\na j},{\bar \beta}_\na) - K_{C2}({\bar \lambda}_{\na i},{\bar \lambda}_{\na j},{\bar \beta}_\na)) P_\na + (K_{C3}({\bar \lambda}_{\na i},{\bar \lambda}_{\na j},{\bar \beta}_\na) - K_{C4}({\bar \lambda}_{\na i},{\bar \lambda}_{\na j},{\bar \beta}_\na))}\right)\nonumber\\
&\eqdef C_C(P_\na,h_{\na,\nc}, {\bar \lambda}_{\na i},{\bar \lambda}_{\na j},{\bar \beta}_\na)\\
I(&U_{\na i};Y_\nc,{\hat Y}_\nd) = \nonumber\\
&C\left(\frac{f^2({\bar \lambda}_{\na i},{\bar 1},{\bar \beta}_\na)K_{BC2}(P_\na,h_{\na,\nc},h_{\na,\nd},P_{\hat \nd},\sigma_\nd)}{f({\bar \lambda}_{\na i},{\bar \lambda}_{\na i},{\bar \beta}_\na)K_{BC1}(P_\na,h_{\na,\nc},h_{\na,\nd},P_{\hat \nd},\sigma_\nd) -f^2({\bar \lambda}_{\na i},{\bar 1},{\bar \beta}_\na)K_{BC2}(P_\na,h_{\na,\nc},h_{\na,\nd},P_{\hat \nd},\sigma_\nd)} \right)\nonumber \\
&\eqdef C_{BC}(P_\na,h_{\na,\nc},h_{\na,\nd},{\bar \lambda}_{\na i},{\bar \beta}_\na,P_{\hat \nd},\sigma_\nd)
\end{align}
where $P_{\hat \nd} = E[{\hat Y}_\nd^2]$, $\sigma_\nd = E[Y_\nd {\hat Y}_\nd]$ and
\begin{align}
K_{BE1}({\bar \lambda}_{\na i},{\bar \lambda}_{\na j},{\bar \lambda}_{\na k},{\bar \beta}_\na) = &f({\bar \lambda}_{\na i},{\bar \lambda}_{\na i},{\bar \beta}_\na) f({\bar \lambda}_{\na j},{\bar \lambda}_{\na j},{\bar \beta}_\na) f({\bar \lambda}_{\na k},{\bar \lambda}_{\na k},{\bar \beta}_\na) - \nonumber\\
& f({\bar \lambda}_{\na i},{\bar \lambda}_{\na i},{\bar \beta}_\na) f^2({\bar \lambda}_{\na j},{\bar \lambda}_{\na k},{\bar \beta}_\na)\\
K_{BE2}({\bar \lambda}_{\na i},{\bar \lambda}_{\na j},{\bar \lambda}_{\na k},{\bar \beta}_\na) = &f({\bar \lambda}_{\na j},{\bar \lambda}_{\na j},{\bar \beta}_\na) f^2({\bar \lambda}_{\na i},{\bar \lambda}_{\na k},{\bar \beta}_\na) + f({\bar \lambda}_{\na k},{\bar \lambda}_{\na k},{\bar \beta}_\na) f^2({\bar \lambda}_{\na i},{\bar \lambda}_{\na j},{\bar \beta}_\na) - \nonumber\\
&2 f({\bar \lambda}_{\na i},{\bar \lambda}_{\na j},{\bar \beta}_\na) f({\bar \lambda}_{\na j},{\bar \lambda}_{\na k},{\bar \beta}_\na)f({\bar \lambda}_{\na i},{\bar \lambda}_{\na k},{\bar \beta}_\na)\\
K_{C1}({\bar \lambda}_{\na i},{\bar \lambda}_{\na j},{\bar \beta}_\na) = &f({\bar \lambda}_{\na i},{\bar \lambda}_{\na i},{\bar \beta}_\na) f({\bar \lambda}_{\na j},{\bar \lambda}_{\na j},{\bar \beta}_\na) - f({\bar \lambda}_{\na i},{\bar \lambda}_{\na i},{\bar \beta}_\na)f^2({\bar \lambda}_{\na j},{\bar 1},{\bar \beta}_\na)\\
K_{C2}({\bar \lambda}_{\na i},{\bar \lambda}_{\na j},{\bar \beta}_\na) = &f^2({\bar \lambda}_{\na i},{\bar \lambda}_{\na j},{\bar \beta}_\na)  + f({\bar \lambda}_{\na j},{\bar \lambda}_{\na j},{\bar \beta}_\na)f^2({\bar \lambda}_{\na i},{\bar 1},{\bar \beta}_\na) -\nonumber\\
&2 f({\bar \lambda}_{\na i},{\bar \lambda}_{\na j},{\bar \beta}_\na)f({\bar \lambda}_{\na i},{\bar 1},{\bar \beta}_\na) f({\bar \lambda}_{\na j},{\bar 1},{\bar \beta}_\na) \\
K_{C3}({\bar \lambda}_{\na i},{\bar \lambda}_{\na j},{\bar \beta}_\na) = &f({\bar \lambda}_{\na i},{\bar \lambda}_{\na i},{\bar \beta}_\na)f({\bar \lambda}_{\na j},{\bar \lambda}_{\na j},{\bar \beta}_\na)\\
K_{C4} ({\bar \lambda}_{\na i},{\bar \lambda}_{\na j},{\bar \beta}_\na)= & f^2({\bar \lambda}_{\na i},{\bar \lambda}_{\na j},{\bar \beta}_\na)\\
K_{BC1}(P_\na,h_{\na,\nc},h_{\na,\nd},P_{\hat \nd},\sigma_\nd) = & |h_{\na,\nc}|^2 P_\na \left(P_{\hat \nd} - \sigma_\nd^2 \frac{|h_{\na,\nd}|^2 P_\na }{|h_{\na,\nd}|^2 P_\na +1}\right)+ P_{\hat \nd}\\
K_{BC2}(P_\na,h_{\na,\nc},h_{\na,\nd},P_{\hat \nd},\sigma_\nd) = & |h_{\na,\nc}|^2 P_\na \left(P_{\hat \nd} - \sigma_\nd^2 \frac{|h_{\na,\nd}|^2 P_\na }{|h_{\na,\nd}|^2 P_\na +1}\right)+ \sigma_\nd^2 \frac{|h_{\na,\nd}|^2 P_\na }{|h_{\na,\nd}|^2 P_\na +1}
\end{align}
\item Double transmitter (when both $\na$ and $\nb$ are broadcasting)
\begin{align}
I(&X_\na; Y_\nc,U_{\nb i}) = C\left(\frac{|h_{\na,\nc}|^2 P_\na f({\bar \lambda}_{\nb i},{\bar \lambda}_{\nb i},{\bar \beta}_\nb)}{|h_{\nb,\nc}|^2P_\nb (f({\bar \lambda}_{\nb i},{\bar \lambda}_{\nb i},{\bar \beta}_\nb)-f^2({\bar \lambda}_{\nb i},{\bar 1},{\bar \beta}_\nb))  +f({\bar \lambda}_{\nb i},{\bar \lambda}_{\nb i},{\bar \beta}_\nb)} \right)\nonumber \\
&\eqdef C_{M}(P_\na,P_\nb,h_{\na,\nc},h_{\nb,\nc},{\bar \lambda}_{\nb i},{\bar \beta}_\nb)\\
I(&U_{\na i}; Y_\nc) = C\left(\frac{f^2({\bar \lambda}_{\na i},{\bar 1},{\bar \beta}_\na) |h_{\na,\nc}|^2 P_\na }{f({\bar \lambda}_{\na i},{\bar \lambda}_{\na i},{\bar \beta}_\na) (|h_{\na,\nc}|^2 P_\na + |h_{\nb,\nc}|^2 P_\nb +1)-f^2({\bar \lambda}_{\na i},{\bar 1},{\bar \beta}_\na) |h_{\na,\nc}|^2 P_\na} \right)\nonumber \\
&\eqdef C_{BI}(P_\na,P_\nb,h_{\na,\nc},h_{\nb,\nc},{\bar \lambda}_{\na i},{\bar \beta}_\na)\\
I(&U_{\na i}; Y_\nc | U_{\nb j}) = \nonumber\\
& C\left(\frac{|h_{\na,\nc}|^2 P_\na K_{BM1}({\bar \lambda}_{\na i},{\bar \lambda}_{\nb j},{\bar \beta}_\na,{\bar \beta}_\nb) }{|h_{\na,\nc}|^2 P_\na K_{BM2}({\bar \lambda}_{\na i},{\bar \lambda}_{\nb j},{\bar \beta}_\na,{\bar \beta}_\nb) + |h_{\nb,\nc}|^2 P_\nb K_{BM3}({\bar \lambda}_{\na i},{\bar \lambda}_{\nb j},{\bar \beta}_\na,{\bar \beta}_\nb)+K_{BM4}({\bar \lambda}_{\na i},{\bar \lambda}_{\nb j},{\bar \beta}_\na,{\bar \beta}_\nb)} \right)\nonumber \\
&\eqdef C_{BM}(P_\na,P_\nb,h_{\na,\nc},h_{\nb,\nc},{\bar \lambda}_{\na i},{\bar \lambda}_{\nb j},{\bar \beta}_\na,{\bar \beta}_\nb)\\
I(&U_{\na i}; Y_\nc ,U_{\na j}| U_{\nb k}) =  C\left(\frac{K_{CI,N}(P_\na,P_\nb,h_{\na,\nc},h_{\nb,\nc},{\bar \lambda}_{\na i},{\bar \lambda}_{\na j},{\bar \lambda}_{\nb k},{\bar \beta}_\na,{\bar \beta}_\nb) }{K_{CI,D}(P_\na,P_\nb,h_{\na,\nc},h_{\nb,\nc},{\bar \lambda}_{\na i},{\bar \lambda}_{\na j},{\bar \lambda}_{\nb k},{\bar \beta}_\na,{\bar \beta}_\nb)} \right)\nonumber \\
&\eqdef C_{CI}(P_\na,P_\nb,h_{\na,\nc},h_{\nb,\nc},{\bar \lambda}_{\na i},{\bar \lambda}_{\na j},{\bar \lambda}_{\nb k},{\bar \beta}_\na,{\bar \beta}_\nb)
\end{align}
where
\begin{align}
&K_{BM1}({\bar \lambda}_{\na i},{\bar \lambda}_{\nb j},{\bar \beta}_\na,{\bar \beta}_\nb) =  f({\bar \lambda}_{\nb j},{\bar \lambda}_{\nb j},{\bar \beta}_\nb) f^2({\bar \lambda}_{\na i},{\bar 1},{\bar \beta}_\na)\\
&K_{BM2}({\bar \lambda}_{\na i},{\bar \lambda}_{\nb j},{\bar \beta}_\na,{\bar \beta}_\nb) = f({\bar \lambda}_{\nb j},{\bar \lambda}_{\nb j},{\bar \beta}_\nb)(f({\bar \lambda}_{\na i},{\bar \lambda}_{\na i},{\bar \beta}_\na) -f^2({\bar \lambda}_{\na i},{\bar 1},{\bar \beta}_\na))\\
&K_{BM3}({\bar \lambda}_{\na i},{\bar \lambda}_{\nb j},{\bar \beta}_\na,{\bar \beta}_\nb)= f({\bar \lambda}_{\na i},{\bar \lambda}_{\na i},{\bar \beta}_\na)(f({\bar \lambda}_{\nb j},{\bar \lambda}_{\nb j},{\bar \beta}_\nb) -f^2({\bar \lambda}_{\nb j},{\bar 1},{\bar \beta}_\nb))\\
&K_{BM4}({\bar \lambda}_{\na i},{\bar \lambda}_{\nb j},{\bar \beta}_\na,{\bar \beta}_\nb)= f({\bar \lambda}_{\na i},{\bar \lambda}_{\na i},{\bar \beta}_\na)f({\bar \lambda}_{\nb j},{\bar \lambda}_{\nb j},{\bar \beta}_\nb)\\
&K_{CI,N}(P_\na,P_\nb,h_{\na,\nc},h_{\nb,\nc},{\bar \lambda}_{\na i},{\bar \lambda}_{\na j},{\bar \lambda}_{\nb k},{\bar \beta}_\na,{\bar \beta}_\nb) =|h_{\na,\nc}|^2 P_\na K_{CI1}({\bar \lambda}_{\na i},{\bar \lambda}_{\na j},{\bar \lambda}_{\nb k},{\bar \beta}_\na,{\bar \beta}_\nb) + \nonumber\\
&~~~~~~~~~~~~~~~~~~~~~ |h_{\nb,\nc}|^2 P_\nb K_{CI2}({\bar \lambda}_{\na i},{\bar \lambda}_{\na j},{\bar \lambda}_{\nb k},{\bar \beta}_\na,{\bar \beta}_\nb) + K_{CI3}({\bar \lambda}_{\na i},{\bar \lambda}_{\na j},{\bar \lambda}_{\nb k},{\bar \beta}_\na,{\bar \beta}_\nb)\\
&K_{CI,D}(P_\na,P_\nb,h_{\na,\nc},h_{\nb,\nc},{\bar \lambda}_{\na i},{\bar \lambda}_{\na j},{\bar \lambda}_{\nb k},{\bar \beta}_\na,{\bar \beta}_\nb) = |h_{\na,\nc}|^2 P_\na K_{CI4}({\bar \lambda}_{\na i},{\bar \lambda}_{\na j},{\bar \lambda}_{\nb k},{\bar \beta}_\na,{\bar \beta}_\nb) + \nonumber\\
&~~~~~~~~~~~~~~~~~~~~~ |h_{\nb,\nc}|^2 P_\nb K_{CI5}({\bar \lambda}_{\na i},{\bar \lambda}_{\na j},{\bar \lambda}_{\nb k},{\bar \beta}_\na,{\bar \beta}_\nb) + K_{CI6}({\bar \lambda}_{\na i},{\bar \lambda}_{\na j},{\bar \lambda}_{\nb k},{\bar \beta}_\na,{\bar \beta}_\nb)\\
&K_{CI1}({\bar \lambda}_{\na i},{\bar \lambda}_{\na j},{\bar \lambda}_{\nb k},{\bar \beta}_\na,{\bar \beta}_\nb) = f({\bar \lambda}_{\nb k},{\bar \lambda}_{\nb k},{\bar \beta}_\nb) f^2({\bar \lambda}_{\na i},{\bar 1},{\bar \beta}_\na)f({\bar \lambda}_{\na j},{\bar \lambda}_{\na j},{\bar \beta}_\na) + \nonumber\\
&f({\bar \lambda}_{\nb k},{\bar \lambda}_{\nb k},{\bar \beta}_\nb) f^2({\bar \lambda}_{\na i},{\bar \lambda}_{\na j},{\bar \beta}_\na) - 2f({\bar \lambda}_{\nb k},{\bar \lambda}_{\nb k},{\bar \beta}_\nb) f({\bar \lambda}_{\na i},{\bar \lambda}_{\na j},{\bar \beta}_\na) f({\bar \lambda}_{\na i},{\bar 1},{\bar \beta}_\na) f({\bar \lambda}_{\na j},{\bar 1},{\bar \beta}_\na)\\
&K_{CI2}({\bar \lambda}_{\na i},{\bar \lambda}_{\na j},{\bar \lambda}_{\nb k},{\bar \beta}_\na,{\bar \beta}_\nb) =  f^2({\bar \lambda}_{\na i},{\bar \lambda}_{\na j},{\bar \beta}_\na)(f({\bar \lambda}_{\nb k},{\bar \lambda}_{\nb k},{\bar \beta}_\nb) -f^2({\bar \lambda}_{\nb k},{\bar 1},{\bar \beta}_\nb) )\\
&K_{CI3}({\bar \lambda}_{\na i},{\bar \lambda}_{\na j},{\bar \lambda}_{\nb k},{\bar \beta}_\na,{\bar \beta}_\nb) =  f^2({\bar \lambda}_{\na i},{\bar \lambda}_{\na j},{\bar \beta}_\na)f({\bar \lambda}_{\nb k},{\bar \lambda}_{\nb k},{\bar \beta}_\nb) \\
&K_{CI4}({\bar \lambda}_{\na i},{\bar \lambda}_{\na j},{\bar \lambda}_{\nb k},{\bar \beta}_\na,{\bar \beta}_\nb) = f({\bar \lambda}_{\na i},{\bar \lambda}_{\na i},{\bar \beta}_\na)f({\bar \lambda}_{\nb k},{\bar \lambda}_{\nb k},{\bar \beta}_\nb)(f({\bar \lambda}_{\na j},{\bar \lambda}_{\na j},{\bar \beta}_\na)-f^2({\bar \lambda}_{\na j},{\bar 1},{\bar \beta}_\na)) - \nonumber\\
&~~~~~~~~~~~~~~~~~~~~~~~~~~~~~~~~~~K_{CI1}({\bar \lambda}_{\na i},{\bar \lambda}_{\na j},{\bar \lambda}_{\nb k},{\bar \beta}_\na,{\bar \beta}_\nb)\\
&K_{CI5}({\bar \lambda}_{\na i},{\bar \lambda}_{\na j},{\bar \lambda}_{\nb k},{\bar \beta}_\na,{\bar \beta}_\nb) =  f({\bar \lambda}_{\na i},{\bar \lambda}_{\na i},{\bar \beta}_\na)f({\bar \lambda}_{\na j},{\bar \lambda}_{\na j},{\bar \beta}_\na)(f({\bar \lambda}_{\nb k},{\bar \lambda}_{\nb k},{\bar \beta}_\nb)-f^2({\bar \lambda}_{\nb k},{\bar 1},{\bar \beta}_\nb)) - \nonumber\\
&~~~~~~~~~~~~~~~~~~~~~~~~~~~~~~~~~~K_{CI2}({\bar \lambda}_{\na i},{\bar \lambda}_{\na j},{\bar \lambda}_{\nb k},{\bar \beta}_\na,{\bar \beta}_\nb)\\
&K_{CI6}({\bar \lambda}_{\na i},{\bar \lambda}_{\na j},{\bar \lambda}_{\nb k},{\bar \beta}_\na,{\bar \beta}_\nb) = f({\bar \lambda}_{\na i},{\bar \lambda}_{\na i},{\bar \beta}_\na)f({\bar \lambda}_{\na j},{\bar \lambda}_{\na j},{\bar \beta}_\na)f({\bar \lambda}_{\nb k},{\bar \lambda}_{\nb k},{\bar \beta}_\nb) - \nonumber\\
&~~~~~~~~~~~~~~~~~~~~~~~~~~~~~~~~~~K_{CI3}({\bar \lambda}_{\na i},{\bar \lambda}_{\na j},{\bar \lambda}_{\nb k},{\bar \beta}_\na,{\bar \beta}_\nb)
\end{align}
\end{enumerate}

For example, if we take
\begin{align}
{\bf \Lambda}_\nr = \left[
                     \begin{array}{cc}
                       1 & \frac{\beta_{\nr1} |h_{\nr,\none}|^2 P_\nr}{\beta_{\nr1} |h_{\nr\none}|^2 P_\nr +1} \\
                       0 & 1 \\
                     \end{array}
                   \right]
\end{align}
and assume the degraded channel
$X_\nr \rightarrow Y_\none \rightarrow Y_\ntwo$ in the second phase of the FMABC-N protocol, then we can achieve the following rate regions in the second phase:
\begin{align}
R_{\nzero,\none} &< \Delta_2 I(U_\none^\pb;Y_\none^\pb) - \Delta_2 I(U_\none^\pb;U_\ntwo^\pb) = \Delta_2 C(\beta_{\nr1}|h_{\nr,\none}|^2P_\nr)\\
R_{\nzero,\ntwo} &< \Delta_2 I(U_\ntwo^\pb;Y_\ntwo^\pb) = \Delta_2 C\left(\frac{\beta_{\nr2}|h_{\nr,\ntwo}|^2 P_\nr}{\beta_{\nr1} |h_{\nr,\ntwo}|^2 P_\nr +1 }\right)
\end{align}
This is the well-known capacity region for the degraded channel
$X_\nr \rightarrow Y_\none \rightarrow Y_\ntwo$. However, this may
not be optimum due to the opposite data rates $R_{\none,\nzero}$ and
$R_{\ntwo,\nzero}$. Thus, generally it is necessary to optimize
${\bf\Lambda}_{\nr}$ and ${\bar \beta}_\nr$ to
maximize the entire achievable region.

%=====================BIBLIOGRAPHY===================
\bibliographystyle{IEEEtranS}
\bibliography{sang}
\end{document}